\documentclass[11pt]{article}
\usepackage{amsmath}
\usepackage{amsthm}
\usepackage{amssymb}
\usepackage{algorithm}
\usepackage{subfig}
\usepackage{color}
\usepackage[english]{babel}
\usepackage{graphicx}

\usepackage{wrapfig,epsfig}
\usepackage{epstopdf}
\usepackage{url}
\usepackage{graphicx}
\usepackage{color}
\usepackage{epstopdf}
\usepackage{algpseudocode}
\usepackage{scrextend}
\usepackage[T1]{fontenc}
\usepackage{bbm}
\usepackage{comment}

\usepackage{tikz}
\usepackage{hyperref}  
\hypersetup{colorlinks=true,citecolor=blue,linkcolor=blue} 
\usetikzlibrary{arrows}
\usepackage[margin=1in]{geometry}
\graphicspath{{./figs/}}

\usepackage{tikz}
\usetikzlibrary{positioning}
\definecolor{processblue}{cmyk}{0.96,0,0,0}

\author{
  Alexandr Andoni\thanks{Research supported in part by Simons Foundation (\#491119 to
Alexandr Andoni), NSF (CCF-1617955, CCF-1740833), and Google
Research Award.}\\
  \texttt{andoni@cs.columbia.edu}\\
  Columbia University
  \and
  Clifford Stein\thanks{Research supported in part by NSF grants CCF-1714818 and CCF-1822809.} \\
  \texttt{cliff@cs.columbia.edu}\\
  Columbia University
  \and
  Peilin Zhong\thanks{Research supported in part by NSF grants
    CCF-1740833, CCF-1703925, CCF-1714818 and CCF-1822809, as well as
    the Google PhD Fellowship.}\\
  \texttt{peilin.zhong@columbia.edu}\\
  Columbia University
}

\date{}
\title{Parallel Approximate Undirected Shortest Paths\\ Via Low Hop Emulators}
\newtheorem{theorem}{Theorem}[section]
\newtheorem{question}[theorem]{Question}
\newtheorem{lemma}[theorem]{Lemma}
\newtheorem{definition}[theorem]{Definition}

\newtheorem{corollary}[theorem]{Corollary}

\newtheorem{observation}[theorem]{Observation}
\newtheorem{fact}[theorem]{Fact}

\newtheorem{claim}[theorem]{Claim}

\newcommand{\wh}{\widehat}
\newcommand{\wt}{\widetilde}

\newcommand{\eps}{\epsilon}

\renewcommand{\varepsilon}{\epsilon}
\renewcommand{\tilde}{\wt}
\renewcommand{\hat}{\wh}

\DeclareMathOperator*{\E}{{\bf {E}}}

\DeclareMathOperator{\OPT}{OPT}

\DeclareMathOperator{\poly}{poly}

\DeclareMathOperator{\nnz}{nnz}

\DeclareMathOperator{\p}{par}
\DeclareMathOperator{\rt}{root}
\DeclareMathOperator{\map}{map}

\DeclareMathOperator{\dist}{dist}

\DeclareMathOperator{\EMD}{EMD}

\DeclareMathOperator{\B}{Ball}

\DeclareMathOperator{\diam}{diam}

\DeclareMathOperator{\lev}{level}

\DeclareMathOperator{\sgn}{sgn}
\DeclareMathOperator{\one}{ \mathbf{1}}

\makeatletter
\newcommand*{\RN}[1]{\expandafter\@slowromancap\romannumeral #1@}
\makeatother

\usepackage{lineno}

\usepackage{etoolbox}

\makeatletter

\newcommand{\define}[4][ignore]{%
  \ifstrequal{#1}{ignore}{}{
  \@namedef{thmtitle@#2}{#1}}%
  \@namedef{thm@#2}{#4}%
  \@namedef{thmtypen@#2}{lemma}%
  \newtheorem{thmtype@#2}[theorem]{#3}%
  \newtheorem*{thmtypealt@#2}{#3~\ref{#2}}%
}

\newcommand{\state}[1]{%
  \@namedef{curthm}{#1}
  \@ifundefined{thmtitle@#1}{
  \begin{thmtype@#1}
    }{
  \begin{thmtype@#1}[\@nameuse{thmtitle@#1}]
  }
    \label{#1}
    \@nameuse{thm@#1}
  \end{thmtype@#1}
  \@ifundefined{thmdone@#1}{
  \@namedef{thmdone@#1}{stated}%
  }{}
}

\newcommand{\restate}[1]{%
  \@namedef{curthm}{#1}
  \@ifundefined{thmtitle@#1}{
    \begin{thmtypealt@#1}
    }{
  \begin{thmtypealt@#1}[\@nameuse{thmtitle@#1}]
  }
    \@nameuse{thm@#1}
  \end{thmtypealt@#1}
  \@ifundefined{thmdone@#1}{
  \@namedef{thmdone@#1}{stated}%
  }{}
}

\newcommand{\thmlabel}[1]{
  \@ifundefined{thmdone@\@nameuse{curthm}}{\label{#1}
    }{\tag*{\eqref{#1}}}
}
\makeatother

\begin{document}

\begin{titlepage}
  \maketitle
%
\begin{abstract}
	We present a $(1+\varepsilon)$-approximate parallel algorithm for
computing shortest paths in undirected graphs, achieving $\poly(\log
n)$ depth and $m\poly(\log n)$ work for $n$-nodes $m$-edges graphs.
Although sequential algorithms with (nearly) optimal running time have
been known for several decades, near-optimal parallel algorithms have
turned out to be a much tougher challenge.  For
$(1+\varepsilon)$-approximation, all prior algorithms with $\poly(\log
n)$ depth perform at least $\Omega(mn^{c})$ work for some constant
$c>0$.  Improving this long-standing upper bound obtained by
Cohen (STOC'94) has been open for $25$ years.

We develop several new tools of independent interest. One of them is a
new notion beyond hopsets --- low hop emulator --- a $\poly(\log n)$-approximate emulator graph in which every shortest path has at most $O(\log\log n)$ hops (edges).
Direct applications of the low hop emulators are parallel
algorithms for $\poly(\log n)$-approximate single source shortest path
(SSSP), Bourgain's embedding, metric tree embedding, and low diameter
decomposition, all with $\poly(\log n)$ depth and $m\poly(\log n)$ work.

To boost the approximation ratio to $(1+\varepsilon)$, we introduce
compressible preconditioners and apply it inside Sherman's framework
(SODA'17) to solve the more general problem of uncapacitated minimum
cost flow (a.k.a., transshipment problem).  Our algorithm computes a
$(1+\varepsilon)$-approximate uncapacitated minimum cost flow in
$\poly(\log n)$ depth using $m\poly(\log n)$ work.  As a consequence,
it also improves the state-of-the-art sequential running time from
$m\cdot 2^{O(\sqrt{\log n})}$ to $m\poly(\log n)$.

\end{abstract}

  \thispagestyle{empty}
\end{titlepage}











\section{Introduction}
The problem of finding the shortest path between two vertices in an
undirected weighted graph is one of the most fundamental problems in computer science. 
Standard sequential algorithms with (nearly) optimal running time have been known for several decades~\cite{ft87,t99,hktz15}.
In contrast, parallelizing these algorithms has been a challenge, and existing parallel algorithms are far from attaining the
efficiency we would like.
Two standard measures of the efficiency of a parallel algorithm in the
standard PRAM model of parallelism are work (total time over the processors)\footnote{More precisely, 
	the work is the running time required when only one processor can be
	used, i.e., the sequential running time when the algorithm is
	implemented in the vanilla RAM model.
}
and depth (parallel time).
The exact shortest path can be computed by the standard path-doubling
(Floyd-Warshall) algorithm in $\poly(\log n)$ parallel time using
$O(n^3)$ total work, for an $n$-node $m$-edge graph.
This result has been improved in a long line of work~\cite{s97,ks97,c97,btz98,ss99,bgst16,sd18}.
Nevertheless, the state-of-the-art algorithms have either
$\Omega(n^{2.1})$ work or $\Omega(n^{0.1})$ depth.



 In order to achieve algorithms with better
bounds on work and depth, researchers have turned to approximation
algorithms.  Building on the idea of hopsets \cite{c94}, a series
of papers, including~\cite{ks92,c94,c00,mpvx15,en16} give $(1 +
\varepsilon)$-approximation algorithms.
Yet again, every prior algorithm with $m\poly(\log n)$ work has at
least $\Omega(n^{\rho})$ depth, and the ones with $\poly(\log n)$ depth do $\Omega(mn^{\rho})$ work, where $\rho>0$ is an arbitrary small constant. 
In particular, none of the prior algorithms achieve $\poly(\log n)$ depth and $m\poly(\log n)$ work simultaneously.
In fact, there was no known parallel algorithm with $\poly(\log
n)$ parallel time and $m\poly(\log n)$ work that approximates the
shortest path even up to a $\poly(\log n)$ factor. Hence, after~\cite{c94},
a major question  which remained open for more than 25 years  is:
\begin{question}
	Is there a parallel algorithm which computes an approximate
        shortest path in $\poly(\log n)$ depth and $m\poly(\log n)$ work?
\end{question}
In this paper, we answer this question positively by developing a
parallel $(1+\varepsilon)$-approximate shortest path algorithm with
$\poly(\log n)$ depth and $m\poly(\log n)$ work. 

	\subsection{Our Results and Comparison to Prior Approaches}

        To obtain our main result, we develop new tools, which we
        present next and which may be of independent interest. It is most natural to
        present these results in the context of two related approaches
        to parallel shortest path algorithms --- hopsets and
        continuous optimization techniques.

We note that some of our results have new consequences beyond parallel algorithms,
including faster sequential algorithms and constructions where none
were previously known. 
Our input is a connected $n$-vertex $m$-edge undirected weighted graph $G=(V,E,w)$ with weights $w:E\rightarrow \mathbb{Z}_{\geq 0}$ and $\max_{e\in E}w(e)\leq \poly(n)$.
The parallel algorithms from this paper are in the EREW PRAM model.


\paragraph{Hopsets.}

 One iteration of Bellman-Ford can be implemented efficiently in parallel,
	and therefore, for graphs in which an approximate shortest path has a small number of hops (edges) we already have an efficient algorithm. 
	Motivated by this insight, researchers have proposed adding edges to a graph in order to make an approximate shortest path with a small number of edges between every pair of vertices. 
Formally, for a given graph $G=(V,E,w)$ with weights
$w:E\to \mathbb{R}_{\geq 0}$, a hopset is an edge-set $H$ with weights
$w_H:H\rightarrow \mathbb{R}_{\geq 0}$. Let $\wt{G}$ be the union
graph $(V,E\cup H,w\cup w_H)$.  We define $\dist_{\wt{G}}^{(h)}(u,v)$,
the $h$-hop distance in $\wt{G}$, to be the length of the shortest
path between $u,v\in V$ which uses at most $h$ hops (edges) in
$\wt{G}$.  Then $H$ is an {\em $(h,\varepsilon)$-hopset} of $G$ if
$\forall u,v\in V$, $\dist^{(h)}_{\wt{G}}(u,v)$ is always a
$(1+\varepsilon)$-approximation to the shortest distance between $u$
and $v$ in the graph $G$. There is a three-way trade-off between $h$,
$\varepsilon$, and $|H|$, which was studied
in~\cite{c94,tz06,mpvx15,en16,hp19}, leading to some of the aforementioned
 algorithms.

Surprisingly, a hard barrier arose: \cite{abp18} showed that the size
of $H$ must be $\Omega(n^{1+\rho})$ for any $h\leq \poly(\log n)$,
$\varepsilon<\tfrac{1}{\log n}$ and some constant $\rho > 0$.  Thus, it is
impossible to directly apply hopsets to compute a 
$(1+\varepsilon)$-approximate shortest path in $\poly(\tfrac{\log
	n}{\eps})$ parallel time using 
 $m\poly(\frac{\log n}{\varepsilon})$ work
for sparse
graphs $G$, when $|E|=O(n)$.

 \paragraph{Low Hop Emulator.}
To bypass this hardness, we introduce a new notion --- {\em low
	hop emulator} --- which has a  weaker approximation guarantee than hopsets, but has stronger guarantees in other ways.
A low hop emulator $G'=(V,E',w')$ of $G$ is a sparse graph with $n\poly(\log n)$ edges satisfying two properties.
First, the distance between every pair of vertices in $G'$ is a $\poly(\log n)$ approximation to the distance in $G$.
The second property is that $G'$ has a low hop diameter, i.e., a shortest path between every pair of two vertices in $G'$ only contains $O(\log\log n)$ number of hops (edges). 


	We give an efficient parallel sparse low hop emulator construction
	algorithm.  To the best of our knowledge, it was not even clear
	whether sparse low hop emulators exist, and thus no previous
	algorithm was known even in the sequential setting.
	
	\begin{theorem}[Low hop emulator, restatement of Theorem~\ref{thm:parallel_low_hop_emulator}]
		For any $k\geq 1$, any graph $G$ admits a low hop emulator
		$G'$, with expected size of $\wt{O}\left(n^{1+\frac{1}{k}}\right)$,\footnote{$\wt{O}(f(n))$ denotes $f(n)\cdot \poly\log(f(n))$.} satisfying:
		\begin{align*}
		\forall u,v\in V, \dist_{G}(u,v)\leq \dist_{G'}(u,v)\leq \poly(k)\cdot \dist_G(u,v),
		\end{align*}
		and with the hop diameter at most $O(\log k)$.
		Furthermore, there is a PRAM algorithm computing the emulator $G'$ in
		$\poly\log(n)$ parallel time using $\wt{O}(m+n^{1+\frac{2}{k}})$ expected work.
	\end{theorem}
	
	Notice that, setting $k=\log n$, we can compute a low hop emulator with expected size $\wt{O}(n)$ and hop diameter $O(\log\log n)$ in $\poly\log(n)$ parallel time using $\wt{O}(m)$ expected work.
	The approximation ratio in this case is $\poly\log(n)$.

	We now highlight two main features that make a low hop emulator stronger than hopsets. 
	Firstly, the low hop emulator can be computed in $\poly(\log n)$ parallel time using $m\poly(\log n)$ work while the same guarantees cannot be simultaneously achieved by hopsets.
	Secondly, the $O(\log\log n)$-hop distances in low hop emulator $G'$
	satisfy the {\em triangle inequality} while the $h$-hop distances in the union graph $\wt{G}$ of original graph $G$ and the $(h,\varepsilon)$-hopset do not.


An immediate application of the first feature is a $\poly(\log n)$-approximate single source shortest path (SSSP) algorithm in $\poly(\log n)$ parallel time using $m\poly(\log n)$ work 
 (Corollary~\ref{cor:parallel_sssp}).
{We remark that when we use the hop-distance to approximate the exact distance in $G$, we only need to use the edges from the low hop emulator while we also need to use original edges if we use hopsets.}

The second feature is crucial for designing 
 parallel algorithms for Bourgain's embedding~\cite{b85}  (Corollary~\ref{cor:parallel_embed_into_l1}), metric tree embedding~\cite{frt04,fl18}  (Corollary~\ref{cor:parallel_metric_tree}) and low diameter decomposition~\cite{mpx13} (Corollary~\ref{cor:parallel_ldd}), using $\poly(\log n)$ depth and $m\poly(\log n)$ work. 
\cite{fl18} introduced a notion similar to low hop emulators, and 
it also has the second feature mentioned above.
In contrast, their emulator graph is a complete graph, and the construction is based on $\left(\poly(\log n),\tfrac{1}{\poly(\log n)}\right)$-hopsets.

\paragraph{Continuous optimization.} To boost the
approximation ratio of shortest path from $\poly(\log n)$ to
$(1+\varepsilon)$, we employ continuous optimization techniques.
Recently, continuous optimization techniques have been successfully
applied to design new efficient algorithms for many classic
combinatorial graph problems,
e.g.,~\cite{ds08,ckmst11,s13,mad13,klos14,ls14,cmsv17,s17,s17b,knp19}.
Most of them can be seen as ``boosting'' a coarse approximation
algorithm to a more accurate approximation algorithm.  Oftentimes, to
fit into a general optimization framework, the ``coarse''
approximation must be for a more general problem --- in our case, for
the uncapacitated minimum cost flow, also known as the transshipment
problem.  Following this approach, the work of \cite{bkkl17} develops
near-optimal uncapacitated min-cost flow algorithms in the distributed and streaming
settings based on the gradient descent algorithm.  Their algorithm can
be seen as boosting a $\poly(\log n)$ approximate solver for the
uncapacitated min-cost flow problem to an $(1+\varepsilon)$
approximate solver, but with one crucial difference: it requires a
$\poly(\log n)$ approximate solver for the {\em dual problem}.  Hence
it is not clear how to leverage their algorithm for our goal as the
aforementioned techniques do not seem applicable to the dual of
uncapacitated min-cost flow.

We develop an algorithm for the uncapacitated min-cost flow problem by opening up Sherman's framework~\cite{s17} and combining it with new techniques. 
There is a fundamental challenge in adopting Sherman's framework,
beyond implementing it in the parallel setting.
Sherman's original algorithm solves the uncapacitated minimum cost
flow problem in $m\cdot 2^{O(\sqrt{\log n})}$ sequential time. Hence,
if we obtain
a parallel uncapacitated min-cost flow algorithm with $m\poly(\log n)$ total work,
we cannot avoid improving this best-known running time of $m\cdot 2^{O(\sqrt{\log n})}$ to $m\poly(\log n)$. 

\paragraph{Uncapacitated minimum cost flow and approximate $s-t$ shortest path.}
To handle  the challenge mentioned above, we develop a novel {\em compressible} preconditioner.
By using our compressible preconditioner inside Sherman's framework,
we improve the running time of $(1+\varepsilon)$-approximate uncapacitated
min-cost flow from $m\cdot 2^{O(\sqrt{\log n})}$ to $m\poly(\log n)$.
Furthermore, we show that such compressible preconditioner can be
computed in $\poly(\log n)$ parallel time using $m\poly(\log n)$ work.
This preconditioner relies crucially on our low hop emulator ideas.


 Formally, in the uncapacitated minimum cost flow problem, given a demand vector $b\in\mathbb{R}^n$ satisfying $\sum_{v\in V} b_v = 0$,
the goal is to determine the flow on each edge such that the demand of
each vertex is satisfied and the cost of the flow is minimized. 
\begin{theorem}[Parallel uncapacitated minimum cost flow, restatement of Theorem~\ref{thm:parallel_min_cost_flow}]
	Given a graph $G=(V,E,w)$, a demand vector $b\in\mathbb{R}^{n}$ and an error parameter $\varepsilon\in(0,0.5)$, there is a PRAM algorithm which outputs an $(1+\varepsilon)$-approximate solution to the uncapacitated minimum cost flow problem with probability at least $0.99$ in $\varepsilon^{-2}\poly\log(n)$ parallel time using $\wt{O}(\varepsilon^{-2}m)$ expected work.
\end{theorem}


While the above techniques are sufficient for estimating the {\em value} of
the shortest path, one additional challenge arises when we want to
compute an $(1+\varepsilon)$-approximate shortest {\em path}. In particular, the continuous optimization framework
produces an approximate shortest path {\em flow}, which is not
necessary integral and, more crucially, may contain cycles.
We address this challenge by developing a novel recursive algorithm
based on random walks, and which uses a coupling argument.

\begin{theorem}[Parallel $(1+\varepsilon)$-approximate $s-t$ shortest path, restatement of Theorem~\ref{thm:st_shortest_path}]
	Given a graph $G=(V,E,w)$, two vertices $s,t\in V$ and an error parameter $\varepsilon\in (0,0.5)$, there is a PRAM algorithm which can output a $(1+\varepsilon)$-approximate $s-t$ shortest path with probability at least $0.99$ in $\varepsilon^{-2}\poly\log(n)$ parallel time using expected $\wt{O}(\varepsilon^{-3}m)$ work.
\end{theorem}

\paragraph{Massive Parallel Computing (MPC).} Although we present our
parallel algorithms in the PRAM model, they can also be implemented in
the Massive Parallel Computing (MPC) model~\cite{FMSSS10-mad, ksv10,
	gsz11, bks13, ANOY14-parallel}
 which is an abstract of massively parallel computing systems such as MapReduce~\cite{dg08}, Hadoop~\cite{w12}, Dryad~\cite{ibybf07}, Spark~\cite{zcfss10}, and others. 
In particular, MPC model allows $m^{\delta}$ space per machine for some $\delta\in(0,1)$.
The computation in MPC proceeds in rounds, where each machine can perform unlimited local computation and exchange up to $m^{\delta}$ data in one round.

By applying the simulation methods~\cite{ksv10,gsz11}, our PRAM algorithm can be directly simulated in MPC.
The obtained MPC algorithm has $\poly(\log n)$ rounds and only needs $m\cdot \poly(\log n)$ total space. 
Furthermore, it is also fully scalable, i.e., the memory size per machine can be allowed to be $m^\delta$ for any constant $\delta\in(0,1)$.
To the best of our knowledge, this is the first MPC algorithm which computes $(1+\varepsilon)$-approximate shortest path using $\poly(\log n)$ rounds and $m\poly(\log n)$ total space when the memory of each machine is upper bounded by $n^{1-\Omega(1)}$. 
{
Previous work on shortest paths in the MPC model include \cite{dn19}
when the  memory size per machine is $o(n)$, and simulations of shortest path
algorithms from the Congested Clique model~\cite{lpspp05,knpr15,hkn16,bkkl17,chdkl19,hkn19}  when the memory size per machine is $\Omega(n)$~\cite{bdh18}.
}

\subsection{Our Techniques}
In this section, we give an overview of techniques that we use in our algorithms.
Figure~\ref{fig:tech_summary} sketches the dependencies between our techniques and the main results mentioned in this paper. 
\begin{figure}[t!]
	\centering
	\includegraphics[width=0.8\textwidth]{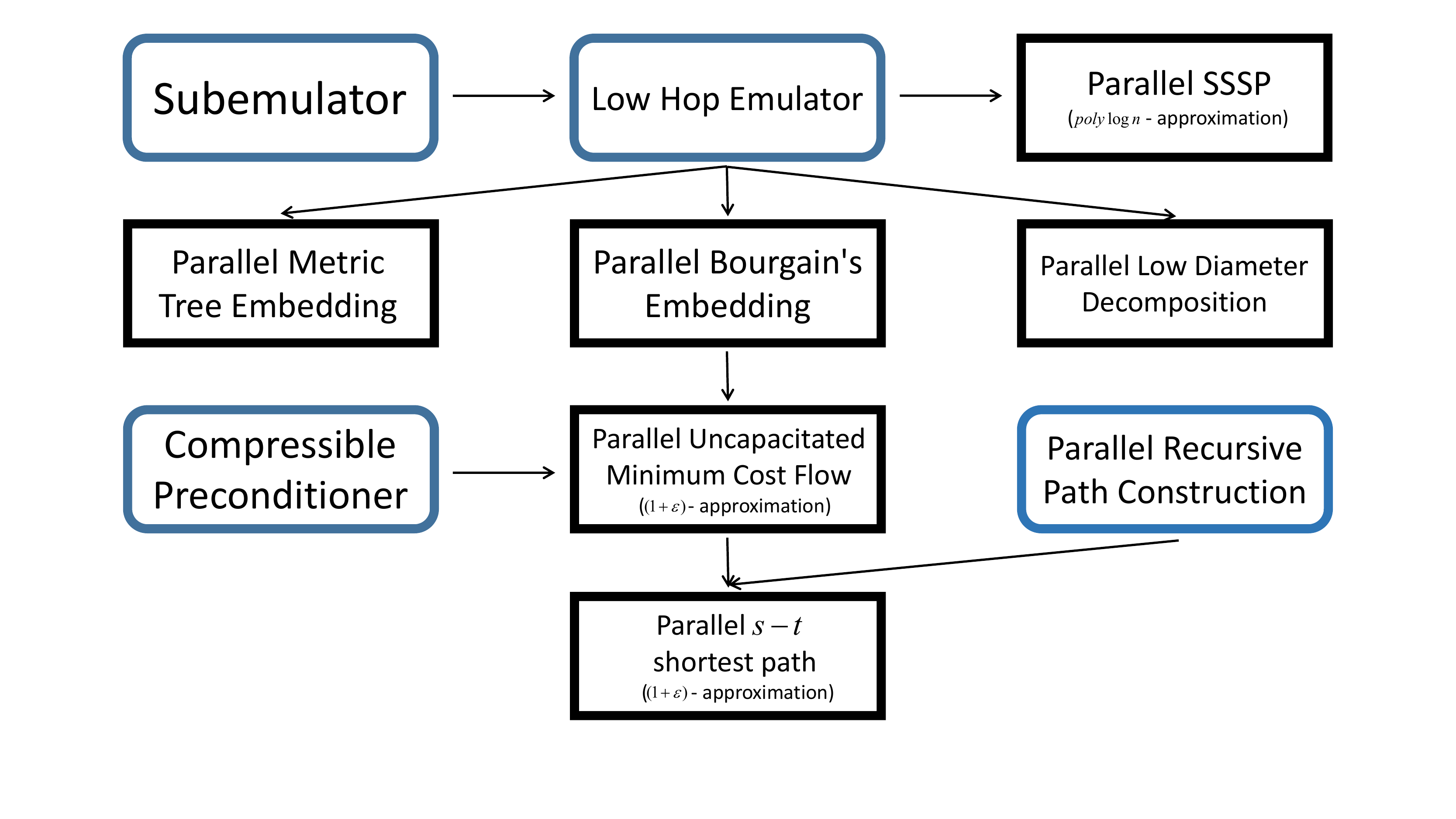}
	\vspace{-10mm}
	\caption{ \small A summary of techniques and results mentioned in this paper. 
	Blue rounded rectangles indicate new techniques developed in this paper. 
}\label{fig:tech_summary}
	\vspace{-3mm}
\end{figure}

\subsubsection{Low Hop Emulator} 
A concept closely related to low hop emulator are hopsets~\cite{c94}.
A hopset is a set of weighted shortcut edges such that for any two
vertices $s$ and $t$ we can always find an approximate shortest path
connecting them using small number of edges from the hopset and the
original graph.
Many hopset construction
methods~\cite{c94,ks97,tz05,tz06,mpvx15,b09,hkn14,hkn16,en16,hp19}
share some common features --- they all choose a layer or multiple
layers of leader vertices, and the hopset edges are some shortcut
edges connecting to these leader vertices.  However, when connecting
shortcut edges to a layer of leader vertices, none of these algorithms
can avoid processing information for all $n$ vertices from the
original graph, even though there may be a large fraction of vertices
which are not connecting any this layer's leader vertex in the final
hopset.  Furthermore, each of these algorithms needs either $n\cdot
\log^{\omega(1)} n$ work (sequential time) or $\log^{\omega(1)} n$
depth to process $n$ vertices for constructing shortcut edges for some
layers.  To improve the efficiency of these algorithms, a natural
question is: can we reduce the number of vertices needed to be
processed when constructing the shortcut edges?

\paragraph{Subemulator.} 
Motivated by the above question, we introduce a new concept called subemulator. 
For $\alpha\geq 1$ and an integer $b\geq 1$, we say $H=(V',E',w')$ is an $(\alpha,b)$-subemulator of $G=(V,E,w)$ if 1) $V'$ is a subset of $V$; 2) for any vertex $v$ in $G$, at least one of the $b$-closest neighbors\footnote{We assume that $v$ is also a neighbor of $v$ itself. Thus the closest (or $1$-closest) neighbor of $v$ is $v$ itself.} of $v$ is in $V'$; 3) for any two vertices $u,v$ in $H$, $\dist_H(u,v)$ $\alpha$-approximates $\dist_G(u,v)$. 
In addition, if we can assign each vertex $v\in V$ a leader $q(v)\in V'$ such that $q(v)$ is one of the $b$-closest neighbors of $v$ and for any two vertices $ u,v\in V$ it always satisfies
\begin{align}\label{eq:strong_property}
\dist_G(q(u),q(v))\leq \dist_H(q(u),q(v))\leq \dist_G(q(u),u) + \beta \cdot \dist_G(u,v) + \dist_G(v,q(v)) 
\end{align}
for some $\beta \geq 1$, we call $H$ a strong $(\alpha,b,\beta)$-subemulator of $G$.
A subemulator $H$ can be regarded as a sparsification of vertices of $G$.
{
	Two notions related to subemulators are vertex sparsifiers~\cite{m09,lm10} and distance-preserving minors~\cite{knz14}.
	The major difference between subemulators and vertex sparisfiers is that the vertex sparsifier approximately preserves flow/cut properties for the subset of vertices while the subemulator approximately preserves distances for the subset of vertices.
	Furthermore, both vertex sparsifiers and distance-preserving minors have given fixed vertex sets, whereas the vertex set of the subemulator is not given but should satisfy the condition 2) mentioned above, i.e., each vertex in $G$ has a $b$-closest neighbor which is in the subemulator.
	
}

To construct a strong subemulator $H=(V',E',w')$, we need to construct both a vertex set $V'$ and a edge set $E'$.
For convenience, let us consider the case for $b\gg \log n$.
Constructing $V'$ is relatively easy.
We can add each vertex of $V$ to $V'$ with probability $\Theta(\log (n) / b)$.
By Chernoff bound, with high probability, each vertex has at least one of the $b$-closest neighbors in $V'$ and the size of $V'$ is roughly $\tilde{O}(n/b)$.
For each vertex $v\in V$, it is natural to set the leader vertex $q(v)$ to be the vertex in $V'$ which is the closest vertex to $v$.
The challenge remaining is to construct the edge set $E'$ such that condition 3) and Equation~\eqref{eq:strong_property} can be satisfied.
In our construction, we add two categories of edges to $E'$:
\begin{enumerate}
\item For each edge $\{u,v\}\in E$, add an edge $\{q(u),q(v)\}$ with weight $\dist_G(q(u),u)+w(u,v)+\dist_G(v,q(v))$ to $E'$.
\item For each $v\in V$ and for each $u$ which is a $b$-closest neighbor of $v$, we add an edge $\{q(u),q(v)\}$ with weight $\dist_G(q(u),u)+\dist_G(u,v)+\dist_G(q(v),v)$ to $E'$.
\end{enumerate}
The first category of edges looks natural --- for an edge $\{u,v\}$ of which two end points $u,v$ are assigned to different leader vertices $q(u),q(v)$, we add a shortcut edge connecting those two leader vertices with a weight which is equal to the smallest length of the $q(u)-q(v)$ path crossing edge $\{u,v\}$. 
However, if we only have the edges from the first category, it is not good enough to preserve the distances between leader vertices (see Section~\ref{sec:necessity_edges_subemulator} for examples).
To fix this, we add the second category of edges.
We now sketch the analysis.
It follows from our construction that each edge in $H$ corresponds to a path in $G$.
Thus, $\forall u',v'\in V'$, $\dist_{G}(u',v')\leq \dist_H(u',v')$.
We only need to upper bound $\dist_H(u',v')$.
Let us fix a shortest path $u'=z_0\rightarrow z_1\rightarrow \cdots\rightarrow z_h=v'$ between $u',v'$ in the original graph $G$.
We want to construct a path in $H$ with a short length.
We use the following procedure to find some crucial vertices on the shortest path $z_0\rightarrow \cdots \rightarrow z_h$:
\begin{enumerate}
	\item $y_0\gets u',k\gets 0$. Repeat the following two steps:
	\item Let $x_{k+1}$ be the last vertex on the path $z_0\rightarrow \cdots\rightarrow z_h$ such that $x_{k+1}$ is one of the $b$-closest neighbors of $y_k$. 
	If $x_{k+1}$ is $z_h$, finish the procedure.
	\item Set $y_{k+1}$ to be the next vertex of $x_{k+1}$ on the path $z_0\rightarrow \cdots\rightarrow z_h$. $k\gets k + 1$.
\end{enumerate}
It is obvious that 
\begin{align*}
\dist_G(u',v')=\dist_G(y_k,x_{k+1})+\sum_{i=0}^{k-1} (\dist_G(y_i,x_{i+1})+w(x_{i+1},y_{i+1})).
\end{align*}
For $i=0,1,\cdots,k$, $x_{i+1}$ is a $b$-closest neighbor of $y_i$.
Thus, there is an edge $\{q(y_i),q(x_{i+1})\}$ in $H$ from the second category of the edges.
For $i=1,2,\cdots,k$, $y_{i}$ is adjacent to $x_i$.
Thus, there is an edge $\{q(x_i),q(y_i)\}$ in $H$ from the first category of the edges.
Thus $u'=q(y_0)\rightarrow q(x_1)\rightarrow q(y_1)\rightarrow q(x_2)\rightarrow q(y_2)\rightarrow \cdots \rightarrow q(x_{k+1})=v'$ is a valid path (see Figure~\ref{fig:path_in_subemulator}) in $H$ and the length is 
\begin{align*}
\dist_G(u',v') + 2\cdot \sum_{i=1}^k \left(\dist_G(x_i,q(x_i))+\dist_G(y_i,q(y_i))\right).
\end{align*}
\begin{figure}[t!]
	\centering
	\includegraphics[width=\textwidth]{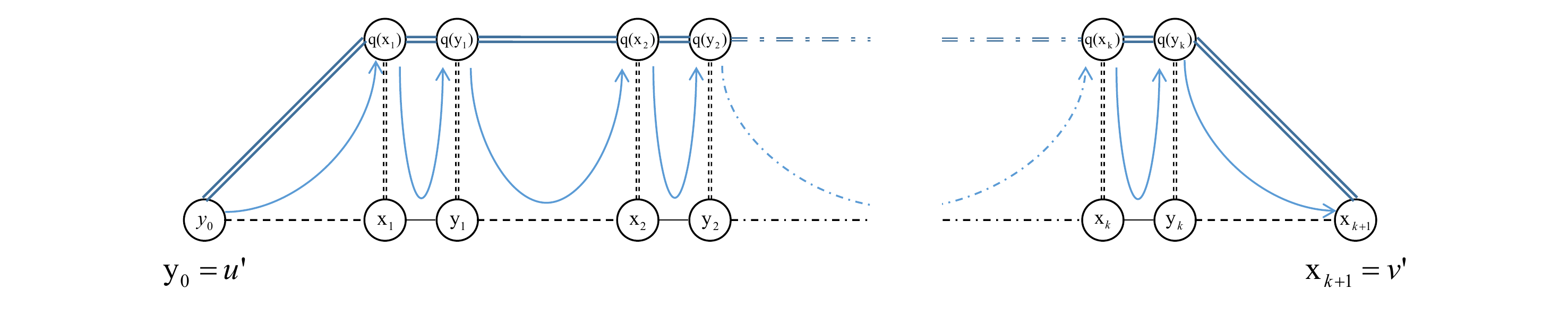}
	\vspace{-10mm}
	\caption{ \small For $u',v'\in V'$ and a shortest path between $u',v'$ in $G$, we can find a corresponding path between $u',v'$ in the subemulator $H$.
	A single dashed line denotes a shortest path in $G$ between $y_{i-1}$ and $x_i$.
	A single solid line denotes an edge $\{x_i,y_i\}$ in $G$.
	A double dashed line denotes a shortest path in $G$ between a vertex and its leader vertex. 
	A double solid blue line denotes an edge in the subemulator $H$ with a weight which is equal to the length of the path in $G$ represented by the corresponding blue arc.
	}\label{fig:path_in_subemulator}
	\vspace{-3mm}
\end{figure}
By our choice of $q(\cdot)$, we have $\forall v\in V,\dist_G(v,q(v))=\dist_G(v,V')$.
So,  
\begin{align*} 
\forall i=1,2,\cdots,k, \dist_G(x_i,q(x_i))\leq \dist_G(y_{i-1},q(y_{i-1}))+\dist_G(y_{i-1},x_i).
\end{align*}
Since $y_i$ is not a $b$-closest neighbor of $y_{i-1}$ but $q(y_{i-1})$ is a $b$-closest neighbor of $y_{i-1}$, 
\begin{align*}
\forall i=1,2,\cdots,k,\dist_G(y_{i-1},q(y_{i-1}))\leq \dist_G(y_{i-1},x_i)+w(x_i,y_i).
\end{align*}
Since $x_{k+1}\in V'$, we have $\dist_G(y_k,q(y_k))\leq \dist_G(y_k,x_{k+1})$.
Then we know $\sum_{i=1}^k \dist_G(x_i,q(x_i))\leq 2\cdot \dist_G(u',v')$ and $\sum_{i=1}^k \dist_G(y_i,q(y_i))\leq \dist_G(u',v')$.
Thus, we can conclude $\dist_H(u',v')\leq 8\cdot \dist_G(u',v')$.
We now argue that our construction of $E'$ also satisfies Equation~\eqref{eq:strong_property} with $\beta = 22$.
There are two cases.
The first case is that either $u$ is a $b$-closest neighbor of $v$ or $v$ is a $b$-closest neighbor of $u$.
In this case, $E'$ contains an edge from the second category with weight $\dist_G(q(u),u)+\dist_G(u,v)+\dist_G(v,q(v))$ which implies Equation~\eqref{eq:strong_property}.
The second case is that neither $u$ is a $b$-closest neighbor of $v$ nor $v$ is a $b$-closest neighbor of $u$.
In this case, we have 
\begin{align*}
\dist_H(q(u),q(v))&\leq 8\dist_G(q(u),q(v))\leq 8(\dist_G(q(u),u)+\dist_G(u,v)+\dist_G(v,q(v)))\\
&\leq \dist_G(q(u),u)+\dist_G(v,q(v)) + 22\dist_G(u,v),
\end{align*}
where the last step follows from $\dist_G(u,q(u)),\dist_G(v,q(v))\leq \dist_G(u,v)$.

The bottleneck of computing a subemulator is to obtain $b$-closest neighbors for each vertex.
We can use the truncated broadcasting technique~\cite{ss99,asswz18} to handle this in $\poly(\log n)$ parallel time using $\tilde{O}(m+nb^2)$ total work. 
The output subemulator has $\tilde{O}(n/b)$ vertices and $O(m+nb)$ edges.
As we can see, there is a trade-off between total work used and the number of vertices in the subemulator: if we can afford more work for the construction of the subemulator, fewer vertices appear in the subemulator. 

\paragraph{Low hop emulator via subemulator.}
Now, 
we describe how to use strong subemulators to construct a low hop emulator.
Consider a weighted undirected graph $G=(V,E,w)$. 
Suppose we obtain a sequence of subemulators $H_0 = (V_0,E_0,w_0),H_1 = (V_1,E_1,w_1),\cdots,H_t = (V_t,E_t,w_t)$ where $H_0=G$ and $\forall i=0,\cdots,t-1,$ $H_{i+1}$ is a strong $(8,b_i,22)$-subemulator of $H_i$ for some integer $b_i\geq 1$.
We have $V=V_0\supseteq V_1\supseteq V_2\supseteq \cdots \supseteq V_t$.
For $v\in V_i$, let us denote $q_i(v)\in V_{i+1}$ as the corresponding assigned leader vertex of $v$ in the subemulator $H_{i+1}$ satisfying Equation~\eqref{eq:strong_property}.
We add following three types of edges to the graph $G'=(V,E',w')$ and we will see that $G'$ is a low hop emulator of $G$:
\begin{enumerate}
	\item $\forall i=0,\cdots, t-1,~\forall v\in V_i$, add an edge $\{v,q_i(v)\}$ with weight $27^{t-i-1}\cdot \dist_{H_i}(v,q_i(v))$ to $G'$.
	\item $\forall i=0,\cdots,t,$ for each edge $\{u,v\}\in E_i$, add an edge $\{u,v\}$ with weight $27^{t-i}\cdot w_i(u,v)$ to $G'$.
	\item $\forall i=0,\cdots,t,~\forall v\in V_i$, add an edge $\{v,u\}$ with weight $27^{t-i}\cdot \dist_{H_i}(v,u)$ to $G'$ for each $u$ which is one of the $b_i$-closest neighbors of $v$ in $H_i$ (define $b_t=|V_t|$).
\end{enumerate}
Roughly speaking, we can imagine that $G'$ is obtained from flattening a graph with $t+1$ layers.
Each layer corresponds to a subemulator.
The lowest layer corresponds to the original graph $G$, and the highest layer corresponds to the last subemulator $H_t$. 
The first type of edges connect the vertices in the lower layer to the leader vertices in the higher layer.
The second type of edges correspond to the subemulators on all layers.
The third type of edges shortcut the close vertices from the same layer.
Furthermore, the weights of the edges on the lower layers have larger penalty factor, i.e., the penalty factor of the edges on the layer $i$ is $27^{t-i}$.

By Equation~\eqref{eq:strong_property} of strong subemulator, we can show that $\forall u,v\in V,\dist_G(u,v)\leq \dist_{G'}(u,v)$.
Consider the first layer.
By the second type edges, we know that $\forall u,v\in V,\dist_{G'}(u,v)\leq 27^t\dist_G(u,v)$.
In particular, for $t= O(\log\log n)$, $G'$ preserves the distances in $G$ up to a $\poly(\log n)$ factor.
Now we want to show that $\forall u,v\in V$, there is always a shortest path connecting $u,v$ in $G'$ such that the number of hops (edges) of the path is at most $4t$.
For convenience, we conceptually split each vertex of $G'$ into vertices on different layers based on the construction of $G'$.
Consider a shortest path $u=z_0\rightarrow z_1\rightarrow z_2\rightarrow \cdots \rightarrow z_h= v$ using the smallest number of hops in $G'$ with splitting vertices.
 By the constructions of three types of edges we know that $\forall j = 0,1,\cdots, h-1,$  $z_j,z_{j+1}$ are either on the same layer or on the adjacent layers
 , and $z_0,z_h$ are on the lowest layer which is corresponding to $H_0=G$.
We will claim two properties of the shortest path $z_0\rightarrow \cdots\rightarrow z_h$.
Suppose $z_j,z_{j+1}$ are on the same layer corresponding to $H_i$.
We claim that $z_{j+2}$ cannot be on the same layer as $z_j$ and $z_{j+1}$.
Intuitively, this is because if $z_{j+2}$ is on the same layer of $z_j$ then there are two cases which both lead to contradictions: in the first case, $z_{j+2}$ is close to $z_j$ such that there is a third type edge connecting $z_j,z_{j+2}$ which implies that $z_{j+1}$ is redundant; in the second case, $z_{j+2}$ is far away from $z_j$ such that $\dist_{H_i}(z_j,q_i(z_j))+\dist_{H_{i+1}}(q_i(z_j),q_i(z_{j+2}))+\dist_{H_i}(q_i(z_{j+2}),z_{j+2})$ is a good approximation to $\dist_{H_{i}}(z_j,z_{j+2})$, and due to a smaller penalty factor, the length of the path $z_j\rightarrow q_i(z_j)\rightarrow (\text{shortest path})\rightarrow q_i(z_{j+2})\rightarrow z_{j+2}$ is smaller than the length of $z_j\rightarrow z_{j+1}\rightarrow z_{j+2}$.
We claim another property of $z_0\rightarrow \cdots \rightarrow z_h$ as the following.
If the layer of $z_{j+1}$ is lower than the layer of $z_j$
, the layer of any of $z_{j+2},z_{j+3},\cdots,z_{h}$ must be lower than the layer of $z_j$.
At a high level, this is because of Equation~\eqref{eq:strong_property} and the smaller penalty factor for higher layers: if we move from higher layer to lower layer then come back to the higher layer, it is always worse than we only move in the higher layers.
Due to these two claims, the shortest path in $G'$ should have the following shape: the path starts from the lowest layer, then keeps moving to the non-lower layers until reach some vertex, and finally keeps moving to the non-higher layers until reach the target. 
Furthermore, there are no three consecutive vertices on the path which are on the same layer.
Hence we can conclude that the shortest path has number of hops at most $4t$.
Based on above analysis, the shortest path in $G'$ will never use the second type edges. 
Thus, in our final construction of $G'$, we only need the first type and the third type of edges.

The size of $G'$ is at most $\sum_{i=0}^t |V_i|\cdot b_i$.
The bottleneck of the construction of $G'$ is to compute the third type edges.
This can be done by truncated broadcasting technique~\cite{ss99,asswz18} in $t\cdot \poly(\log n)$ parallel time using $\sum_{i=0}^t \left(|E_i|+|V_i|\cdot b_i^2\right)\cdot \poly(\log n)$ total work.
The problem remaining is to determine the sequence of $b_i$.
As we discussed previously, we are able to use $\poly(\log |V_i|)$ parallel time and $\tilde{O}(|E_i|+|V_i|b_i)$ total work to construct a subemulator $H_{i+1}$ with $\tilde{O}(|V_i|/b_i)$ vertices and $O(|E_i|+|V_i| b_i)$ edges. 
By double exponential problem size reduction technique~\cite{asswz18}, we can make $b_i$ grow double exponentially fast in this situation.
More precisely, if we set $b_0\gets \poly(\log n)$, $b_{i+1}\gets b_i^{1.25}$, and $t\gets O(\log\log n)$, then in this case, the result low hop emulator can be computed in $\poly(\log n)$ parallel time and $\wt{O}(m+n)$ total work.
Furthermore, the size of the result low hop emulator is $\tilde{O}(n)$, the approximation ratio is $\poly(\log n)$, and the hop diameter is $O(\log\log n)$.

\paragraph{Applications of low hop emulator.} 
We can build a useful oracle based on a low hop emulator: given a query subset $S$ of vertices, the oracle can output a $\poly(\log n)$ approximations to $\dist_G(v,S)$ for all $v\in V$. 
Furthermore, the output approximate distances always satisfy triangle inequality.
To implement the such oracle, we preprocess an $\wt{O}(n)$ size low hop emulator $G'$ with $\poly(\log n)$ approximation ratio and $O(\log\log n)$ hop diameter in $\poly(\log n)$ parallel time using $\wt{O}(m+n)$ work. 
For each oracle query, we can just run Bellman-Ford on $G'$ with source $S$.
The work needed for each Bellman-Ford iteration is at most $\wt{O}(n)$.
Since the hop diameter is $O(\log\log n)$, the number of iterations needed is $O(\log\log n)$.
Therefore, each query can be handled in $\poly(\log n)$ parallel time and $\wt{O}(n)$ total work.
The triangle inequality is always satisfied since the output approximate distances are exact distances in the graph $G'$.
Several parallel applications such as Bourgain's embedding, metric tree embedding and low diameter decomposition directly follow the oracle. 
We refer readers to see Section~\ref{sec:parallel_app} for more details of these applications.

\subsubsection{Minimum Cost Flow and Shortest Path}

\paragraph{Uncapacitated minimum cost flow.}
At a high level, our uncapacitated minimum cost flow algorithm is based on Sherman's framework~\cite{s17}.
Sherman's algorithm has several recursive iterations.
 It first uses the multiplicative weights update method~\cite{ahk12} to find a flow which almost satisfies the demands and has nearly optimal cost.
 If the unsatisfied parts of demands are sufficiently small, it routes them naively to make the flow truly feasible without increasing the cost by too much. 
Otherwise, it updates the demands to be the unsatisfied parts of the original demands and recursively routes the new demands.  
\cite{s17} shows that if the problem is well conditioned, then the final solution can be computed by the above process efficiently.
However, most of the time the natural form of the uncapacitated minimum cost flow problem is not well-conditioned.
Thus, a preconditioner, i.e., a linear operator $P\in\mathbb{R}^{r\times n}$ applied to the flow constraints, is needed to make the problem well-condtitioned. 
Consider a given graph $G=(V,E,w)$.
Sherman shows that if for any valid demands $b\in \mathbb{R}^n$ we always have
\begin{align*}
\OPT(b)\leq \|Pb\|_1\leq \gamma \cdot \OPT(b),
\end{align*}
then $P$ can make the condition number of the flow problem on $G$ be upper bounded by $\gamma$, where $\OPT(b)$ denotes the optimal cost of the flow on $G$ satisfying the demands $b$.
Sherman gives a method to construct such $P$.
However, to have a smaller approximation ratio $\gamma$, the time of computing matrix-vector multiplication with $P$ must increase such that the running time of the multiplicative weights update step increases.
To balance the trade-off, Sherman constructs $P$ with $\gamma = 2^{O(\sqrt{\log n})}$ approximation ratio and $\nnz(x)\cdot 2^{O(\sqrt{\log n})}$ time for matrix-vector multiplication $P\cdot x$, where $\nnz(x)$ denotes the number of non-zero entries of $x$.
Thus, its final running time is $m\cdot 2^{O(\sqrt{\log n})}$.
To design a parallel minimum cost flow algorithm using $\poly(\log n)$ parallel time and $m\poly(\log n)$ work, we cannot avoid improving the sequential running time of minimum cost flow to $m\poly(\log n)$ time in sequential setting.
By above discussion, a natural way is to find a linear transformation $P$ which can embed the minimum cost flow into $\ell_1$ with $\poly(\log n)$ approximation ratio and the running time for matrix-vector multiplication $P\cdot x$ needs to be $\nnz(x)\cdot \poly(\log n)$. 
Next, we will introduce how to construct such embedding $P$.

First, we compute a mapping $\varphi$ which embeds the vertices into $\ell_1^d$ for $d=O(\log^2 n)$ such that $\forall u,v\in V$, $\|\varphi(u)-\varphi(v)\|_1$ is a $\poly(\log n)$ approximation to $\dist_G(u,v)$.
This step can be done by Bourgain's embedding.
 The parallel version of Bourgain's embedding is one of the applications of low hop emulator as we mentioned previously.
Then we can reduce the minimum cost flow problem to the geometric transportation problem.
The geometric transportation problem is also called Earth Mover's Distance (EMD) problem.
Specifically, it is the following minimization problem:
\begin{align*}
&\min_{\pi:V\times V\rightarrow \mathbb{R}_{\geq 0}} \sum_{(u,v)\in V\times V} \pi(u,v)\cdot \|\varphi(u)-\varphi(v)\|_1\\
s.t.~&\forall u\in V, \sum_{v\in V}\pi(u,v) - \sum_{v\in V}\pi(v,u) = b_u.\notag
\end{align*}
We denote $\OPT_{\EMD}(b)$ as the optimal cost of the above EMD problem.
It is easy to see that $\OPT_{\EMD}(b)$ is a $\poly(\log n)$ approximation to $\OPT(b)$.
Therefore, it suffices to construct $P$ such that for any valid demand vector $b\in\mathbb{R}^n$, 
\begin{align*}
\OPT_{\EMD}(b)\leq\|Pb\|_1\leq \poly(\log n)\cdot \OPT_{\EMD}(b).
\end{align*}
One known embedding of EMD into $\ell_1$ is based on randomly shifted grids~\cite{it03}.
We can without loss of generality assume that the coordinates of $\varphi(v)$ are integers in $\{1,\cdots,\Delta\}$ for some $\Delta$ which is a power of $2$ and upper bounded by $\poly(n)$. 
We create $1+\log \Delta$ levels of cells.
We number each level from $0$ to $\log \Delta$.
Each cell in level $\log\Delta$ has side length $\Delta$.
Each cell in level $i+1$ is partitioned into $2^d$ equal size cells in level $i$ and thus each cell in level $i$ has side length $2^i$.
Therefore each cell in level $0$ can contain at most one point $\varphi(v)$ for $v\in V$.
According to~\cite{it03}, for any valid demand vector $b\in\mathbb{R}^n$,
\begin{align}\label{eq:random_shift_eval}
\E_{\tau\sim\{0,1,\cdots,\Delta-1\}}\left[~\sum_{i=0}^{\log \Delta} \sum_{C:\text{ a cell in level }i}  
2^i\cdot \left|\sum_{v\in V:\varphi(v)+\tau\cdot \one_d\text{ is in the cell }C} b_v\right|~\right]
\end{align}
is always a $\poly(\log n)$ approximation to $\OPT_{\EMD}(b)$, where $\tau$ is drawn uniformly at random from $\{0,1,\cdots,\Delta-1\}$, and $\varphi(v)+\tau\cdot\one_d$ is the point obtained after shifting each coordinate of $\varphi(v)$ by $\tau$.
Since each cell in level $i$ has side length $2^i$, Equation~\eqref{eq:random_shift_eval} is equal to 
\begin{align}\label{eq:random_shift_eval2}
&\sum_{i=0}^{\log \Delta} \frac{1}{2^i}\sum_{\tau=0}^{2^i-1}  \sum_{C:\text{ a cell in level }i} 2^i\cdot \left|\sum_{v\in V:\varphi(v)+\tau\cdot \one_d\text{ is in the cell }C} b_v\right|\notag\\
=& \sum_{i=0}^{\log \Delta} \sum_{C:\text{ a cell in level }i} \sum_{\tau=0}^{2^i-1}   \left|\sum_{v\in V:\varphi(v)+\tau\cdot \one_d\text{ is in the cell }C} b_v\right|.
\end{align}
Equation~\eqref{eq:random_shift_eval2} can be written in the from of $\|Pb\|_1$ where each row of $P$ corresponds to a cell $C$ and a shift value $\tau$, and each column of $P$ corresponds to a vertex $v$.
\begin{figure}[t!]
	\centering
	\includegraphics[width=0.8\textwidth]{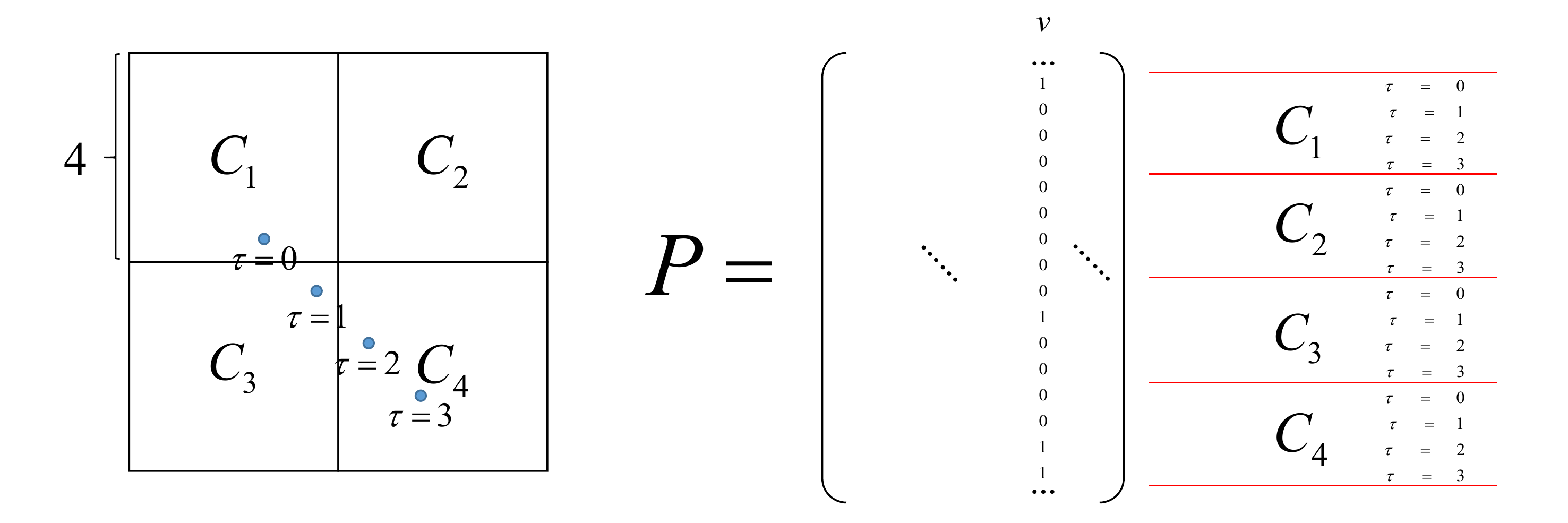}
	\vspace{-3mm}
	\caption{ \small Consider cells $C_1,C_2,C_3,C_4$ shown above with side length $4$. 
		Blue dots denote the positions of $\varphi(v)+\tau\cdot \one_d$ for some vertex $v$ and $\tau =0,1,2,3$.
		The entries of $P$ in the column corresponding to $v$ and in the rows corresponding to $(C,\tau)$ for $C=C_1,C_2,C_3,C_4$ and $\tau=0,1,2,3$ are shown on the right.
	}\label{fig:preconditioner}
\vspace{-5mm}
\end{figure}
Figure~\ref{fig:preconditioner} shows how does $P$ look like:
for a entry $P_{i,j}$ corresponding to a cell $C$, a shift value $\tau$ and a vertex $v$, we have $P_{i,j}=1$ if the point $\varphi(v)+\tau\cdot \one_d$ is in the cell $C$ and $P_{i,j}=0$ otherwise.
Therefore, $P$ can be used to precondition the minimum cost flow problem on $G$ with condition number at most $\poly(\log n)$.
However, such matrix $P$ is dense and have $\poly(n)$ number of rows.
It is impossible to naively write down the whole matrix.
Fortunately, we will show that $P$ has a good structure and we can write down a compressed representation of $P$.
Consider a cell $C$ in level $i$ and a vertex $v$.
If there exists $\tau\in \{0,1,\cdots,2^i-1\}$ such that $\varphi(v)+\tau\cdot \one_d$ is in the cell $C$, then there must exist $\tau_1,\tau_2$ such that $\varphi(v)+\tau\cdot \one_d$ is in the cell $C$ if and only if $\tau\in\{\tau_1,\tau_1+1,\cdots,\tau_2\}$.
In other words, the shift values $\tau$ that can make $\varphi(v)+\tau\cdot\one_d$ be in $C$ are consecutive.
Another important property that we can show is that the number of cells in level $i$ that can contain at least one of the shifted points $\varphi(v), \varphi(v)+ \one_d, \varphi(v)+2\cdot\one_d,\cdots,\varphi(v)+(2^i-1)\cdot \one_d$ is at most $d+1$.
Now consider a column of $P$ corresponding to some vertex $v$.
The entries with value $1$ in this column should be in several consecutive segments.
The number of such segments is at most $(d+1)\cdot(1+\log\Delta)\leq \poly(\log n)$.
Thus, for each column of $P$, we can just store the beginning and the ending positions of these segments.
The whole matrix $P$ can be represented by $n\poly(\log n)$ segments.
The only problem remaining is to use this compressed representation to do matrix-vector multiplication.
Suppose we want to compute $y=P\cdot x$ for some $x\in \mathbb{R}^n$.
It is equivalent to the following procedure:
\begin{enumerate}
\item Initialize $y$ to be an all-zero vector.
\item For each column $i$ and for each segment $[l,r]$ in column $i$, increase all $y_l,y_{l+1},\cdots,y_{r}$ by $x_i$.
\end{enumerate}
We can reduce the above procedure to the following one:
\begin{enumerate}
\item Initialize $z$ to be an all-zero vector.
\item For each column $i$ and for each segment $[l,r]$ in column $i$, increase $z_l$ by $x_i$ and increase $z_{r+1}$ by $-x_i$.
\item Compute $y_j\gets \sum_{k=1}^j z_k$.
\end{enumerate}
In the above procedure, we only need to compute a prefix sum for $z$. 
Since each column of $P$ has at most $\poly(\log n)$ segments, the total number of segments involved is at most $\nnz(x)\cdot \poly(\log n)$.
The total running time is $\tilde{O}(\nnz(x)\cdot\poly(\log n))$.
Notice that even though $y$ has a large dimension, it can be decomposed into $\wt{O}(\nnz(x)\cdot\poly(\log n))$ segments where the entries of each segment have the same value.
Thus, we just store the beginning and the ending positions of each segment of $y$.

Each step of computing the compressed representation can be implemented in $\poly(\log n)$ parallel time and each step of computing the matrix-vector multiplication can also be implemented in $\poly(\log n)$ parallel time.
We obtained a desired preconditioner. 
By plugging this preconditioner into Sherman's framework, we can obtain a parallel $(1+\varepsilon)$-approximate uncapacitated minimum cost flow algorithm with $\poly(\log n)$ depth and $\varepsilon^{-2}m\cdot\poly(\log n)$ work.

\paragraph{Parallel $(1+\varepsilon)$-approximate $s-t$ shortest path.}
$s-t$ Shortest path is closely related to uncapacitated minimum cost flow.
If we set demand $b_s=1,b_t=-1$ and $b_v=0$ for $v\not=s,t\in V$, then the optimal cost of the flow is exactly $\dist_G(s,t)$.
Thus, computing a $(1+\varepsilon)$-approximation to $\dist_G(s,t)$ can be achieved by our flow algorithm.
However, the flow algorithm can only output a flow but not a path. 
We need more effort to find a path from $s$ to $t$ with length at most $(1+\varepsilon)\cdot\dist_G(s,t)$.
As mentioned by~\cite{bkkl17}, if the $(1+\varepsilon)$-approximate flow does not contain any cycles, then for each vertex $v\not=t$ we can choose an out edge with probability proportional to the magnitude of its out flow, and the expected length of the path found from $s$ to $t$ is exactly the cost of the flow which is $(1+\varepsilon)\cdot\dist_G(s,t)$.
Unfortunately, the flow outputted by our flow algorithm may create cycles.
If we randomly choose an out edge for each vertex $v\not=t$ with probability proportional to the magnitude of the out flow, we may stuck in some cycle and may not find a path from $s$ to $t$.
To handle cycles, we propose the following procedure to find a path from $s$ to $t$.
\begin{enumerate}
\item If the graph only has constant number of vertices, find the shortest path from $s$ to $t$ directly.
\item Otherwise, compute the $(1+\varepsilon')$-approximate minimum cost flow from $s$ to $t$ for $\varepsilon'=\Theta(\varepsilon/\log n)$.
\item For each vertex except $t$, choose an out edge with probability proportional to its out flow.
\item Consider the graph with $n-1$ chosen edges. 
Each connected component in the graph is either a tree or a tree plus an edge. 
A component is a tree if and only if $t$ is in the component. 
For each component, we compute a spanning tree.
If the component contains $t$, we set $t$ as the root of the spanning tree.
Otherwise, we set an arbitrary end point of the non-tree edge as the root of the spanning tree.
\item Construct a new graph of which vertices are roots of spanning trees.
For each edge $\{u,v\}$ in the original graph, we add an edge connecting the root of $u$ and the root of $v$ with weight 
\begin{align*}
(\text{distance from }u\text{ to the root of }u\text{ on the spanning tree})+w(u,v)\\
+(\text{distance from }v\text{ to the root of }u\text{ on the spanning tree}).
\end{align*}
\item Recursively find a $(1+\varepsilon')$-approximate shortest path from the root of $s$ to $t$ in the new graph.
Recover a path in the original graph from the path in the new graph.
\end{enumerate}
In the above procedure, only $1/2$ vertices can be root vertices. 
Thus, the procedure can recurse at most $\log n$ times which implies that the parallel time of the algorithm is at most $\poly(\log n)$ and the total work is still $\sim m\poly(n)$.
Now analyze the correctness.
It is easy to see that each edge in the new graph corresponds to a path between two root vertices in the original graph.
Thus a path from the root of $s$ to $t$ in the new graph corresponds to an $s-t$ path in the original graph.
We only need to show that the distance between the root of $s$ and $t$ in the new graph can not be much larger than the distance between $s$ and $t$ in the original graph.
To prove this, we show that if we do a random walk starting from $s$ and for each step we choose the next vertex with probability proportional to the out flow, the expected length of the random walk to reach $t$ is exactly the cost of the flow.
By coupling argument, we can prove that the expected length of the distance between the root of $s$ and $t$ in the new graph is at most $(1+O(\varepsilon'))\cdot(\text{the cost of the flow})$.
Thus, the expected length of the final $s-t$ path is at most $(1+O(\varepsilon'))^{\log n}\cdot \dist_G(s,t)\leq (1+\varepsilon)\cdot\dist_G(s,t)$. 

\section{Preliminaries}\label{sec:preli}
Let $[n]$ denote the set $\{1,2,\cdots,n\}$. 
For a set $V$, $2^V$ denotes the family of all the subsets of $V$, i.e., $2^V=\{S\mid S\subseteq V\}$.
In this paper, we will only consider graphs with non-negative weights.
Let $G=(V,E,w)$ be a connected undirected weighted graph with vertex set $V$, edge set $E$, and weights of the edges $w:E\rightarrow \mathbb{Z}_{\geq 0}$.
Let both $\{u,v\},\{v,u\}$ denote an undirected edge between $u$ and $v$.
For each edge $e=\{u,v\}\in E$, let both $w(u,v),w(v,u)$ denote $w(e)$. 
For $v\in V$, let $w(v,v)$ be $0$.
Consider a tuple $p=(u_0,u_1,u_2,\cdots,u_h)\in V^{h+1}$. 
If $\forall i\in[h]$, either $u_{i}=u_{i-1}$ or $\{u_{i-1},u_{i}\}\in E$, then $p$ is a path between $u_0$ and $u_h$. 
The number of hops of $p$ is $h$, and
the length of $p$ is defined as $w(p)=\sum_{i=1}^{h} w(u_{i-1},u_i)$.
For $u,v\in V,$ let $\dist_G(u,v)$ denote the length of the shortest path between $u,v$, i.e., $\dist_G(u,v)=w(p^*)$, where the path $p^*$ between $u,v$ satisfies that $\forall$path $p$ between $u,v$, $w(p^*)\leq w(p)$.
Similarly, $\dist_G^{(h)}(u,v)$ denotes the $h$-hop distance between $u,v$, i.e., $\dist^{(h)}_G(u,v)=w(p')$, where the $h$-hop path $p'$ between $u,v$ satisfies that $\forall h$-hop path $p$ between $u,v$, $w(p')\leq w(p)$.
The diameter $\diam(G)$ of $G$ is defined as $\max_{u,v\in V} \dist_G(u,v)$.
The hop diameter of $G$ is defined as the minimum value of $h\in\mathbb{Z}_{\geq 0}$ such that $\forall u,v\in V,\dist_G(u,v)=\dist_G^{(h)}(u,v)$.
For $S\subseteq V, v\in V$, we define $\dist_{G}(v,S)=\dist_G(S,v)=\min_{u\in S}\dist(u,v)$.
Similarly, we define $\dist^{(h)}_{G}(v,S)=\dist^{(h)}_G(S,v)=\min_{u\in S}\dist^{(h)}_G(u,v)$.
If $G$ is clear in the context, we use $\dist(\cdot,\cdot)$ and $\dist^{(h)}(\cdot,\cdot)$ for short.

Consider two weighted graphs $G=(V,E,w)$ and $G'=(V,E',w')$.
If $\forall u,v\in V,\dist_G(u,v)\leq \dist_{G'}(u,v)\leq \alpha\cdot \dist_{G}(u,v)$ for some $\alpha\geq 1$, then $G'$ is called an $\alpha$-emulator of $G$. 

Given $r\in \mathbb{Z}_{\geq 0}$, for $v\in V$, we define $\B_G(v,r)=\{u\in V\mid \dist_G(u,v)\leq r\}$, and $\B_G^{\circ}(v,r)=\{u\in V\mid \dist_G(u,v)< r\}$.
Given $b\in[|V|]$, for $v\in V$, let $r_{G,b}(v)$ satisfy that $|\B_{G}(v,r_{G,b}(v))|\geq b$ and $|\B_{G}^{\circ}(v,r_{G,b}(v))|<b$. 
We define $\B_{G,b}(v)=\B_{G}(v,r_{G,b}(v))$, and $\B_{G,b}^{\circ}(v)=\B_{G}^{\circ}(v,r_{G,b}(v))$.
If there is no ambiguity, we just use $\B(v,r)$, $\B^{\circ}(v,r),$ $r_b(v)$, $\B_b(v),$ $\B_b^{\circ}(v)$ to denote $\B_G(v,r)$, $\B_G^{\circ}(v,r)$, $r_{G,b}(v)$, $\B_{G,b}(v)$, $\B_{G,b}^{\circ}(v)$ respectively for short.

For a vector $x\in\mathbb{R}^{m}$ we use $\|x\|_1$ to denote the $\ell_1$ norm of $x$, i.e., $\|x\|_1=\sum_{i=1}^m |x_i|$.
We use $\|x\|_{\infty}$ to denote the $\ell_{\infty}$ norm of $x$, i.e., $\|x\|_{\infty}=\max_{i\in[m]}|x_i|$.
Given a matrix $A\in\mathbb{R}^{n\times m}$, we use $A_i,\ A^j$ and $A_{j,i}$ to denote the $i$-th column, the $j$-th row and the entry in the $i$-th column and the $j$-th row of $A$ respectively.
We use $\|A\|_{1\rightarrow 1}$ to denote the operator $\ell_1$ norm of $A$, i.e., $\|A\|_{1\rightarrow 1}=\sup_ {x:x\not=0}\frac{\|Ax\|_1}{\|x\|_1}$.
A well-known fact is that $\|A\|_{1\rightarrow 1}=\max_{i\in[m]}\|A_i\|_1$.
We use $\one_n$ to denote an $n$ dimensional all-one vector.
We use $\sgn(a)$ to denote the sign of $a$, i.e., $\sgn(a)=1$ if $a\geq 0$, and $\sgn(a)=-1$ otherwise.
We use $\nnz(\cdot)$ to denote the number of non-zero entries of a matrix or a vector.
\section{Low Hop Emulator}
Given a weighted undirected graph $G$, we give a new construction of the graph emulator of $G$.
For any two vertices in our constructed emulator, there is always a shortest path with small number of hops.
Furthermore, our construction can be implemented in parallel efficiently.
	\subsection{Subemulator 
}
In this section, we introduce a new concept which we called \emph{subemulator}.
Later, we will show how to use subemulator to construct an emulator with low hop diameter.
\begin{definition}[Subemulator]\label{def:subemulator}
Consider two connected undirected weighted graphs $G=(V,E,w)$ and $H=(V',E',w')$. For $b\in [|V|]$ and $\alpha\geq 1$, if $H$ satisfies
\begin{enumerate}
\item $V'\subseteq V$,
\item $\forall v\in V,$ $\B_{G,b}(v)\cap V'\not = \emptyset$,
\item $\forall u,v\in V',$ $\dist_G(u,v)\leq \dist_H(u,v)\leq \alpha\cdot \dist_G(u,v)$,
\end{enumerate}
 then $H$ is an $(\alpha,b)$-subemulator of $G$.
Furthermore, if there is a mapping $q:V\rightarrow V'$ which satisfies $\forall v\in V, q(v)\in \B_{G,b}(v)$ and 
\begin{align*}
\forall u,v\in V,\dist_H(q(u),q(v))\leq \dist_G(u,q(u))+\dist_G(v,q(v))+\beta\cdot \dist_G(u,v)
\end{align*}
for some $\beta\geq 1$, then $H$ is a strong $(\alpha,b,\beta)$-subemulator of $G$, $q(\cdot)$ is called a leader mapping, and $q(v)$ is the leader of $v$. 
\end{definition}

Next, we will show how to construct a strong subemulator. 
Algorithm~\ref{alg:leader_selection} gives a construction of the vertices of the subemulator.
Algorithm~\ref{alg:edge_construction} gives a construction of the edges of the subemulator.

\begin{algorithm}[h!]
	\caption{Construction of the Vertices of the Subemulator}\label{alg:leader_selection}
	\begin{algorithmic}[1]
		\small
		\Procedure{\textsc{Samples}}{$G=(V,E,w),b\in[|V|]$}
		\State Output: $V'$
		\State Initialize $S,V'\gets \emptyset$, $n\gets |V|$
		\State For $v\in V$, add $v$ into $S$ with probability $\min(50\log(n)/b,1/2)$. \label{sta:sampled_vertices}
		\State For $v\in V$, if either $v\in S$ or $\B_{G,b}(v)\cap S=\emptyset$, $V'\gets V'\cup\{v\}$. \label{sta:check_each_vertex}
		\State Return $V'$.
		\EndProcedure
	\end{algorithmic}
\end{algorithm}

\begin{algorithm}[h!]
	\caption{Construction of the Edges of the Subemulator} \label{alg:edge_construction}
	\begin{algorithmic}[1]
		\small
		\Procedure{\textsc{Connects}}{$G=(V,E,w),V'\subseteq V,b\in[|V|]$} \Comment{$V',b$ satisfies $\forall v\in V,\B_{G,b}(v)\cap V'\not=\emptyset$.}
		\State Output: $H=(V',E',w')$, $q:V\rightarrow V'$
		\State For $v\in V$, $q(v)\leftarrow \arg \min_{u\in\B_{G,b}(v)\cap V'} \dist_G(u,v)$. \Comment{Choose an arbitrary $u$ if there is a tie.} \label{sta:leader_mapping}
		\State Initialize $E'=\emptyset$.
		\State For $\{u,v\}\in E$, $E'\gets E'\cup\{q(u),q(v)\}$. \label{sta:original_edges}
		\State For $v\in V,u\in \B^{\circ}_{G,b}(v),$ $E'\gets E'\cup\{q(u),q(v)\}$.
		\label{sta:ball_edges}
		\State For $e'\in E'$, initialize $w'(e')\gets \infty$.
		\State For $\{u,v\}\in E$, consider $e'=\{q(u),q(v)\}\in E'$, 
		\begin{align*}
		w'(e')\gets \min(w'(e'),\dist_G(q(u),u)+w(u,v)+\dist_G(v,q(v))).
		\end{align*}\label{sta:weight_assign_original_edge}
		\State For $v\in V,u\in \B^{\circ}_{G,b}(v)$, consider $e'=\{q(u),q(v)\}$, 
		\begin{align*}
		w'(e')\gets \min(w'(e'),\dist_G(q(u),u)+\dist_G(u,v)+\dist_G(v,q(v))).
		\end{align*} \label{sta:weight_assign_ball_edge}
		\State Return $H=(V',E',w')$ and $q:V\rightarrow V'$.
		\EndProcedure
	\end{algorithmic}
\end{algorithm}

The next lemma shows the correctness of the construction of the vertices of the subemulator. 
It also gives an upper bound of the number of vertices in the subemulator.
We put the proof into Section~\ref{sec:proof_of_subemulator_vertex}.

\begin{lemma}[Construction of the vertices]\label{lem:subemulator_vertex}
Consider a connected $n$-vertex $m$-edge undirected weighted graph $G=(V,E,w)$ and a parameter $b\in[n]$. 
\textsc{Samples}$(G,b)$ (Algorithm~\ref{alg:leader_selection}) will output $V'\subseteq V$ such that $\forall v\in V,\B_{G,b}(v)\cap V'\not= \emptyset$.
Furthermore, $\E[|V'|]\leq \frac32 \cdot \min(50\log (n)/b,1/2) \cdot n$.
\end{lemma}

The next lemma shows the correctness of the construction of the edges of the subemulator.
It also gives an upper bound of the number of edges in the subemulator.

\begin{lemma}[Construction of the edges]\label{lem:subemulator_edge}
Consider a connected $n$-vertex $m$-edge undirected weighted graph $G=(V,E,w)$, a vertex set $V'\subseteq V$ and a parameter $b\in[n]$. 
If $\forall v\in V, \B_{G,b}(v)\cap V'\not =\emptyset$, then the output graph $H=(V',E',w')$ of \textsc{Connects}$(G,V',b)$ (Algorithm~\ref{alg:edge_construction}) will be a strong $(8,b,22)$-subemulator (Definition~\ref{def:subemulator}) of $G$, and the output $q:V\rightarrow V'$ is a leader mapping.  Furthermore, $|E'|\leq m+nb$.
\end{lemma}

\begin{proof}
Firstly, let us consider the size of $|E'|$. 
The number of edges added to $E'$ by line~\ref{sta:original_edges} of Algorithm~\ref{alg:edge_construction} is at most $m$.
By the definition of $\B_{G,b}^\circ(v)$, we have $|\B_{G,b}^\circ(v)|<b$. 
The number of edges added to $E'$ by line~\ref{sta:ball_edges} is at most $n\cdot b$.
Thus, we can conclude $|E'|\leq m + nb$.

In the following, we will show that $H$ is actually a good subemulator of $G$.
The first two properties of Definition~\ref{def:subemulator} are automatically satisfied by the guarantees of the input $V',b$.
Let us prove the remaining properties.

Consider two arbitrary vertices $u,v\in V'$. 
Let $p=(u=x_0,x_1,\cdots,x_h=v)$ be an arbitrary shortest path between $u,v$ in the graph $H$.
Then $\dist_H(u,v)=w'(p)=\sum_{i=1}^h w'(x_{i-1},x_i)$.
By line~\ref{sta:weight_assign_original_edge} and line~\ref{sta:weight_assign_ball_edge}, $\forall i\in[h],$ there should be $y_i,z_i\in V$ with $q(y_i)=x_{i-1},q(z_i)=x_i$ such that $w'(x_{i-1},x_i)\geq\dist_G(x_{i-1},y_i)+\dist_G(y_i,z_i)+\dist_G(z_i,x_i)\geq \dist_G(x_{i-1},x_i)$. 
Then, $\dist_H(u,v)=\sum_{i=1}^h w'(x_{i-1},x_i)\geq \sum_{i=1}^h \dist_G(x_{i-1},x_i)\geq \dist_G(x_0,x_h)=\dist_G(u,v)$.

Next, we show how to upper bound $\dist_H(u,v)$.
Consider two arbitrary vertices $u,v\in V'$. 
If $v\in \B^{\circ}_{G,b}(u)$, then by line~\ref{sta:weight_assign_ball_edge}, 
\begin{align*}
\dist_H(u,v)\leq w'(u,v)\leq \dist_G(q(u),u)+\dist_G(u,v)+\dist_G(v,q(v))=\dist_G(u,v),
\end{align*}
where the last equality follows from $q(u)=u,q(v)=v$.
Otherwise, we use the following procedure to find a sequence of vertices which are on the shortest path between $u$ and $v$ in the graph $G$.
\begin{enumerate}
\item $y_0\gets u,t\gets 0$.
\item If $v\in \B^{\circ}_{G,b}(y_t)$, finish the procedure.
\item Otherwise, find an edge $\{x_{t+1},y_{t+1}\}\in E$ on the shortest path between $y_t$ and $v$ in $G$ such that $x_{t+1}\in \B_{G,b}^{\circ}(y_t),y_{t+1}\not\in\B_{G,b}^{\circ}(y_t)$.
\item $t\gets t+1$. Go to step 2.
\end{enumerate}
By the above procedure, it is easy to show that
\begin{align}\label{eq:path_length}
 \dist_G(u,v)=\dist_G(y_t,v)+\sum_{i=1}^t (\dist_G(y_{i-1},x_i)+w(x_i,y_i)).
\end{align}

\begin{claim}\label{cla:bound_sum_r}
 $\sum_{i=1}^t r_{G,b}(y_{i-1})\leq \dist_G(u,v)$.
\end{claim}
\begin{proof}
	By our construction of $x_i,y_i$, we know that $\forall i\in[t]$, $y_i\not\in\B^{\circ}_{G,b}(y_{i-1})$. 
	Thus, $\forall i\in[t],$ $r_{G,b}(y_{i-1})\leq \dist_G(y_{i-1},y_i)$.
	We have 
	\begin{align*}
	\sum_{i=1}^t r_{G,b}(y_{i-1})\leq \sum_{i=1}^t \dist_G(y_{i-1},y_i)=\dist_G(u,y_t)\leq \dist_G(u,v).
	\end{align*}
\end{proof}


\begin{claim}\label{cla:bound_ball}
$\forall u,v\in V, |r_{G,b}(u)-r_{G,b}(v)|\leq \dist_G(u,v)$.
\end{claim}
\begin{proof}
Since $\B_{G}(u,r_{G,b}(u))\subseteq \B_{G}(v,r_{G,b}(u)+\dist_G(u,v))$, we have $|\B_{G}(v,r_{G,b}(u)+\dist_G(u,v))|\geq b$ which implies that $r_{G,b}(v)\leq r_{G,b}(u)+\dist_G(u,v)$. 
Similarly, we have $r_{G,b}(u)\leq r_{G,b}(v)+\dist_G(u,v)$.
\end{proof}

\begin{claim}\label{cla:bound_yx}
$\forall i\in[t]$, $w'(q(y_{i-1}),q(x_i))\leq 2r_{G,b}(y_{i-1})+2\dist_G(y_{i-1},x_i)$.
\end{claim}
\begin{proof}
Since $x_i\in\B^{\circ}_{G,b}(y_{i-1})$, we have $w'(q(y_{i-1}),q(x_i))\leq \dist_G(q(y_{i-1}),y_{i-1})+\dist_G(y_{i-1},x_i)+\dist_G(x_i,q(x_i))\leq r_{G,b}(y_{i-1})+\dist_G(y_{i-1},x_i)+r_{G,b}(x_i)$ by line~\ref{sta:weight_assign_ball_edge}. 
Due to Claim~\ref{cla:bound_ball}, $r_{G,b}(x_i)\leq r_{G,b}(y_{i-1})+\dist_G(y_{i-1},x_i)$.
We can conclude $w'(q(y_{i-1}),q(x_i))\leq 2r_{G,b}(y_{i-1})+2\dist_G(y_{i-1},x_i)$.
\end{proof}

\begin{claim}\label{cla:bound_xy}
$\forall i\in[t]$, $w'(q(x_i),q(y_i))\leq 2r_{G,b}(y_{i-1})+2\dist_G(y_{i-1},x_i)+2w(x_i,y_i)$.
\end{claim}
\begin{proof}
Since $\{x_i,y_i\}\in E$, we have $w'(q(x_i),q(y_i))\leq \dist_G(q(x_i),x_i)+w(x_i,y_i)+\dist_G(y_i,q(y_i))\leq r_{G,b}(x_i)+w(x_i,y_i)+r_{G,b}(y_i)$ by line~\ref{sta:weight_assign_original_edge}.
By Claim~\ref{cla:bound_ball}, $r_{G,b}(x_i)\leq r_{G,b}(y_{i-1})+\dist_G(y_{i-1},x_i)$ and $r_{G,b}(y_i)\leq r_{G,b}(y_{i-1})+\dist_G(y_{i-1},x_i)+w(x_i,y_i)$.
We can conclude that $w'(q(x_i),q(y_i))\leq 2r_{G,b}(y_{i-1})+2\dist_G(y_{i-1},x_i)+2w(x_i,y_i)$.
\end{proof}

\begin{claim}\label{cla:bound_final_segement}
$w'(q(y_t),v)\leq 2\dist_G(y_t,v)$.
\end{claim}
\begin{proof}
By our procedure of finding $x_i,y_i$, we know that $v\in \B^{\circ}_{G,b}(y_t)$. 
Notice that $q(v)=v\in V'$.
By line~\ref{sta:leader_mapping}, we know that $\dist_G(y_t,q(y_t))\leq \dist_G(y_t,v)$.
Then by line~\ref{sta:weight_assign_ball_edge}, $w'(q(y_t),v)\leq \dist_G(q(y_t),y_t)+\dist_G(y_t,v)+\dist_G(v,q(v))\leq 2\dist_G(y_t,v)$.
\end{proof}

By Claim~\ref{cla:bound_yx}, Claim~\ref{cla:bound_xy}, and Claim~\ref{cla:bound_final_segement}, we have:
\begin{align}
\dist_{H}(u,v)&\leq w'(q(y_t),v)+\sum_{i=1}^t (w'(q(y_{i-1}),q(x_i))+w'(q(x_i),q(y_i))) \notag\\
&\leq 2\dist_G(y_t,v)+\sum_{i=1}^t(4r_{G,b}(y_{i-1})+4\dist_G(y_{i-1},x_i)+2w(x_i,y_i))\notag\\
&\leq 4\dist_G(u,v)+4\sum_{i=1}^t r_{G,b}(y_{i-1})\notag\\
&\leq 8\dist_G(u,v), \label{sta:dist_uv_in_H}
\end{align}
where the third inequality follows from Equation~\eqref{eq:path_length}, and the last inequality follows from Claim~\ref{cla:bound_sum_r}.

Next, we will show that $H$ is actually a strong subemulator of $G$, and $q:V\rightarrow V'$ is a corresponding leader mapping. 
By line~\ref{sta:leader_mapping}, $\forall v\in V,$ we have $q(v)\in \B_{G,b}(v)\cap V'$. 
Consider two arbitrary vertices $u,v\in V$.
There are two cases. 
In the first case, $u\in \B^{\circ}_{G,b}(v)$ or $v\in\B^{\circ}_{G,b}(u)$. 
In this case, we have $\dist_H(q(u),q(v))\leq w'(q(u),q(v))\leq \dist_G(q(u),u)+\dist_G(u,v)+\dist_G(v,q(v))$, where the last inequality follows from line~\ref{sta:weight_assign_ball_edge}.
In the second case, neither $u\in \B^{\circ}_{G,b}(v)$ nor $v\in \B^{\circ}_{G,b}(u)$.
In this case, since $q(u)\in \B_{G,b}(v),q(v)\in \B_{G,b}(u)$, we know that $\dist_G(u,v)\geq \max(\dist_G(v,q(v)),\dist_G(u,q(u)))$.
Thus, we have
\begin{align*}
\dist_H(q(u),q(v)) & \leq 8\dist_G(q(u),q(v))\\
&\leq 8(\dist_G(q(u),u)+\dist_G(u,v)+\dist(v,q(v)))\\
&= \dist_G(q(u),u)+\dist_G(q(v),v)+(7\dist_G(q(u),u)+7\dist_G(q(v),v)+8\dist_G(u,v))\\
&\leq \dist_G(q(u),u)+\dist_G(q(v),v) + 22\dist_G(u,v),
\end{align*}
where the first inequality follows from Equation~\eqref{sta:dist_uv_in_H}, the second inequality follows from triangle inequality, and the last inequality follows from $\dist_G(u,v)\geq \max(\dist_G(v,q(v)),\dist_G(u,q(u)))$.
\end{proof}

\begin{algorithm}[h]
	\caption{Construction of the Subemulator} \label{alg:subemulator_construction}
	\begin{algorithmic}[1]
		\small
		\Procedure{\textsc{Subemulator}}{$G=(V,E,w),b\in[|V|]$} 
		\State Output: $H=(V',E',w'),q:V\rightarrow V'$
		\State $V'\gets$\textsc{Samples}$(G,b)$. \Comment{Algorithm~\ref{alg:leader_selection}.}
		\State $H,q\gets$\textsc{Connects}$(G,V',b)$. \Comment{Algorithm~\ref{alg:edge_construction}.}
		\State Return $H,q$.
		\EndProcedure
	\end{algorithmic}
\end{algorithm}

\begin{theorem}[Construction of the subemulator]\label{thm:construction_subemulator}
Consider a connected $n$-vertex $m$-edge undirected weighted graph $G=(V,E,w)$ and a parameter $b\in[n]$. 
\textsc{Subemulator}$(G,b)$ (Algorithm~\ref{alg:subemulator_construction}) will output an undirected weighted graph $H=(V',E',w')$ and $q:V\rightarrow V'$ such that $H$ is a strong $(8,b,22)$-subemulator of $G$, and $q$ is a corresponding leader mapping (Definition~\ref{def:subemulator}). 
Furthermore, $\E[|V'|]\leq \min(75\log(n)/b,3/4)n$, $|E'|\leq m + nb$.
\end{theorem}
\begin{proof}
	Follows directly from Lemma~\ref{lem:subemulator_vertex} and Lemma~\ref{lem:subemulator_edge}.
\end{proof}
	\subsection{A Warm-up Algorithm: Distance Oracle}
Given a weighted undirected graph, a distance oracle is a static data structure which uses small space and can be used to efficiently return an approximate distance between any pair of query vertices.
In this section, we give a warm-up algorithm which is a direct application of subemulator.
In section~\ref{sec:low_hop_emulator}, we will show how to apply the preprocessing procedure $\textsc{PreProc}$ (Algorithm~\ref{alg:dis_oracle}) in our construction of low hop emulator.

\begin{algorithm}[h]
	\caption{Distance Oracle} \label{alg:dis_oracle}
	\begin{algorithmic}[1]
		\small
		\Procedure{\textsc{PreProc}}{$G=(V,E,w),k$} 
		\State $n\gets |V|,m\gets |E|$.
		\State $t\gets 0,H_0=(V_0,E_0,w_0)\gets G, b_0\gets \max\left(\lceil(75\log n)^2\rceil,n^{1/(2k)}\right)$.
		\State $n_0\gets |V_0|,m_0\gets |E_0|$
		\While{$n_t\geq b_t$}\label{sta:dis_oracle_end_condition}
			\State $H_{t+1}=(V_{t+1},E_{t+1},w_{t+1}),q_{t}\gets$\textsc{Subemulator}$(H_t, b_t)$. \Comment{See Algorithm~\ref{alg:subemulator_construction}.} \label{sta:use_of_subemulator}
			\State $\forall v\in V_t$, let $B_t(v)\gets \B^{\circ}_{H_t,b_t}(v)\cup \{q_t(v)\}$ and compute and store $\dist_{H_t}(v,u)$ for every $u\in B_t(v)$.\label{sta:set_Bt_mid}
			\State $n_{t+1}\gets |V_{t+1}|,m_{t+1}\gets |E_{t+1}|$.
			\State $b_{t+1}\gets b_t^{1.25}$.\label{sta:set_bt}
			\State $t\gets t+1$.
		\EndWhile
		\State For $v\in V_t$, $B_t(v)\gets V_t$, compute $\dist_{H_t}(v,u)$ for $u\in V_t$, and  $q_t(v)\gets x$ where $x\in V_t$ is smallest.\label{sta:set_Bt_final}
		\EndProcedure
		\Procedure{\textsc{Query}}{$u,v$} 
		\State Output: $d\in \mathbb{Z}_{\geq 0}$
		\State $l\gets 0,d_0\gets 0,u_0\gets u,v_0\gets v$. 
		\While{$v_l\not\in B_l(u_l)$ and $u_l\not\in B_l(v_l)$}\label{sta:query_finish_condition}
			\State $d_l\gets \dist_{H_l}(u_l,q_l(u_l))+\dist_{H_l}(v_l,q_l(v_l))$. \label{sta:assignment_of_dl}
			\State $u_{l+1}=q_l(u_l),v_{l+1}=q_l(v_l)$.
			\State $l\gets l+1$
		\EndWhile
		\State $d_l\gets \dist_{H_l}(u_l,v_l)$. \label{sta:direct_distance}
		\State Return $d=\sum_{i=0}^l d_i$.
		\EndProcedure
	\end{algorithmic}
\end{algorithm}

\begin{lemma}[Properties of the preprocessing algorithm]\label{lem:dis_oracle_space}
Given a connected weighted graph $G=(V,E,w)$ with $|V|=n,$ $|E|=m$, and a parameter $ k\in [0.5,0.5\log n]$, let $t$ be the value at the end of \textsc{PreProc}$(G,k)$ (Algorithm~\ref{alg:dis_oracle}). 
For $i>t$, define $n_i=m_i=0,b_i=b_{i-1}^{1.25},V_i=\emptyset$.
%
 We have following properties:
 \begin{enumerate}
 	\item $t\leq 4\lceil\log(k)+1\rceil$.
 	\item For $i\in \mathbb{Z}_{\geq 0}$, 
 	\begin{itemize}
 	\item $\E[n_{i}]\leq \max(n^{1+1/k},n\cdot(75\log n)^4)/b_i^2$,
 	\item $\E[m_i]\leq  m+2\cdot \max(n^{1+1/(2k)},n\cdot(75\log n)^2)$,
 	\item $\E\left[\sum_{v\in V_i} |B_i(v)|\right]\leq \max(n^{1+1/k},n\cdot(75\log n)^4)/b_i$.
 	\end{itemize}
 \end{enumerate}
\end{lemma}
\begin{proof}
By line~\ref{sta:set_bt} and the definition of $b_i$ for $i>t$, we have $\forall i\in\mathbb{Z}_{\geq 0}, b_i=b_0^{1.25^i}$.
	
Consider $i=1+\lceil\log_{1.25}\log_{b_0} n\rceil$. 
We have $b_i=b_0^{1.25^i}>n$.
By line~\ref{sta:dis_oracle_end_condition}, we can conclude $t<i$.
Thus, $t\leq \lceil\log_{1.25}\log_{b_0} n\rceil\leq \lceil(\log 2k)/\log 1.25\rceil\leq 4\lceil\log(k) + 1\rceil$.
	
Consider $n_i$, we have $\forall i\in \mathbb{Z}_{\geq 1},$ 
\begin{align}
\E[n_{i}]&\leq (75\log n)/b_{i-1}\cdot \E[n_{i-1}]\notag\\
&\leq \E[n_{i-1}]/b_{i-1}^{0.5}\notag\\
&\leq n/\left(\prod_{j=0}^{i-1} b_j\right)^{0.5}\notag\\
&\leq n/\left(b_0^{\sum_{j=0}^{i-1} 1.25^j}\right)^{0.5}\notag\\
&=n/b_0^{(1.25^i-1)\cdot 2}\notag\\
&=n/b_i^2\cdot b_0^2\notag\\
&\leq \max(n^{1+1/k},n\cdot(75\log n)^4)/b_i^2, \label{eq:size_of_ni}
\end{align}
where the first inequality follows from Theorem~\ref{thm:construction_subemulator}, the second inequality follows from that $b_i$ is increasing and $b_0^{0.5}\geq 75\log n$, the forth inequality follows from $\forall j\in\mathbb{Z}_{\geq 0},b_j=b_0^{1.25^j}$.

Consider $m_i$, we have $\forall i\in \mathbb{Z}_{\geq 1}$,
\begin{align*}
\E[m_{i}]&\leq \E[m_{i-1}]+\E[n_{i-1}]\cdot b_i\\
&=m+\sum_{j=0}^{i-1}\E[n_j]\cdot b_j\\
&\leq m+\sum_{j=0}^{i-1} \max(n^{1+1/k},n\cdot(75\log n)^4)/b_j\\
&\leq m+2\max(n^{1+1/(2k)},n\cdot(75\log n)^2),
\end{align*}
where the first inequality follows from Theorem~\ref{thm:construction_subemulator}, the second inequality follows from Equation~\eqref{eq:size_of_ni}, and the last inequality follows from $b_{j+1}\geq 2b_j$ and $b_0=\max(n^{1/(2k)},(75\log n)^2)$.

Consider $\sum_{v\in V_i}|B_i(v)|$, we have $\forall i\in \mathbb{Z}_{\geq 1}$,
\begin{align*}
\E\left[\sum_{v\in V_i}|B_i(v)|\right]\leq \E\left[\sum_{v\in V_i}b_i\right]=E[n_i]\cdot b_i\leq \max(n^{1+1/k},n\cdot(75\log n)^4)/b_i,
\end{align*}
where the first inequality follows from line~\ref{sta:set_Bt_mid}, line~\ref{sta:set_Bt_final} and the definition of $\B_{H_i,b_i}$, and the last inequality follows from Equation~\eqref{eq:size_of_ni}.
\end{proof}

\begin{lemma}[Correctness of the query algorithm]\label{lem:correct_distance_oracle}
Given a connected weighted graph $G=(V,E,w)$ with $|V|=n$, $|E|=m$, and a parameter $k\in[0.5,0.5\log n]$, run preprocessing \textsc{PreProc}$(G,k)$ (Algorithm~\ref{alg:dis_oracle}). 
Then $\forall u,v\in V$, the output $d$ of \textsc{Query}$(u,v)$ (Algorithm~\ref{alg:dis_oracle}) satisfies $\dist_G(u,v)\leq d\leq 26^{4\lceil\log(k)+1\rceil}\dist_G(u,v)$.
The running time of \textsc{Query}$(u,v)$ is $O(\log (4k))$.
\end{lemma}
\begin{proof}
Let $t$ be the value at the end of the preprocessing procedure \textsc{PreProc}$(G,k)$.
Let $l$ be the value at the end of the query procedure \textsc{Query}$(u,v)$.
By induction, we can show that $\forall i\in\{0,1,\cdots,l\},u_i,v_i\in V_i$.
Since $\forall v\in V_t, B_t(v)=V_t$ and the condition of line~\ref{sta:query_finish_condition}, we have $l\leq t$, i.e., the query procedure will terminate. 
By Lemma~\ref{lem:dis_oracle_space}, we have $t\leq 4\lceil\log(k) + 1\rceil$.
Thus, the running time of \textsc{Query}$(u,v)$ is $O(\log (4k))$.

In the following we will show that $\forall i\in\{0,1,\cdots,l\},$ $\dist_{H_i}(u_i,v_i)\leq \sum_{j=i}^l d_j\leq 26^{l-i}\dist_{H_i}(u_i,v_i)$.
Our proof is by induction.
The base case is $i=l$. 
By line~\ref{sta:direct_distance}, $d_l=\dist_{H_l}(u_l,v_l)$.
Suppose $\dist_{H_{i+1}}(u_{i+1},v_{i+1})\leq \sum_{j=i+1}^l d_j\leq 26^{l-i-1}\dist_{H_{i+1}}(u_{i+1},v_{i+1})$. 
For the contraction,
\begin{align*}
\sum_{j=i}^l d_j &= \dist_{H_i}(u_i,u_{i+1})+\dist_{H_i}(v_i,v_{i+1})+\sum_{j=i+1}^l d_j\\
&\geq  \dist_{H_i}(u_i,u_{i+1})+\dist_{H_i}(v_i,v_{i+1})+\dist_{H_{i+1}}(u_{i+1},v_{i+1})\\
&\geq \dist_{H_i}(u_i,u_{i+1})+\dist_{H_i}(v_i,v_{i+1})+\dist_{H_i}(u_{i+1},v_{i+1})\\
&\geq \dist_{H_i}(u_i,v_i),
\end{align*}
where the first equality follows from line~\ref{sta:assignment_of_dl}, the first inequality follows from $\sum_{j=i+1}^l d_j\geq \dist_{H_{i+1}}(u_{i+1},v_{i+1})$, the second inequality follows from Theorem~\ref{thm:construction_subemulator} that $H_{i+1}$ is a subemulator of $H_i$ and Definition~\ref{def:subemulator}, and the last inequality follows from triangle inequality.

For the expansion,
\begin{align*}
\sum_{j=i}^l d_j &= \dist_{H_i}(u_i,u_{i+1})+\dist_{H_i}(v_i,v_{i+1})+\sum_{j=i+1}^l d_j\\
&\leq  \dist_{H_i}(u_i,u_{i+1})+\dist_{H_i}(v_i,v_{i+1})+26^{l-i-1}\dist_{H_{i+1}}(u_{i+1},v_{i+1})\\
&\leq \dist_{H_i}(u_i,u_{i+1})+\dist_{H_i}(v_i,v_{i+1})+8\cdot 26^{l-i-1}\dist_{H_{i}}(u_{i+1},v_{i+1})\\
&\leq \dist_{H_i}(u_i,u_{i+1})+\dist_{H_i}(v_i,v_{i+1})+8\cdot 26^{l-i-1}(\dist_{H_{i}}(u_{i+1},u_{i})+\dist_{H_{i}}(u_{i},v_{i})+\dist_{H_{i}}(v_{i},v_{i+1}))\\
&= (8\cdot 26^{l-i-1}+1)(\dist_{H_i}(u_i,u_{i+1})+\dist_{H_i}(v_i,v_{i+1})) + 8\cdot 26^{l-i-1}\dist_{H_i}(u_i,v_i)\\
&\leq 24\cdot 26^{l-i-1}\dist_{H_i}(u_i,v_i) + 2\dist_{H_i}(u_i,v_i)\\
&\leq 26\cdot 26^{l-i-1}\dist_{H_i}(u_i,v_i)=26^{l-i}\dist_{H_i}(u_i,v_i),
\end{align*}
where the first equality follows from line~\ref{sta:assignment_of_dl}, the first inequality follows from 
\begin{align*}
\sum_{j=i+1}^l d_j\leq 26^{l-i-1}\dist_{H_{i+1}}(u_{i+1},v_{i+1}),
\end{align*}
the second inequality from Theorem~\ref{thm:construction_subemulator} that $H_{i+1}$ is an $(8,b_i)$-subemulator of $H_i$ and Definition~\ref{def:subemulator},
the third inequality follows from triangle inequality, and the forth inequality follows from that $\dist_{H_i}(u_i,v_i)\geq \max(\dist_{H_i}(u_i,u_{i+1}),\dist_{H_i}(v_i,v_{i+1}))$ since neither $u_i\in\B^{\circ}_{H_i,b_i}(v_i)$ nor $v_i\in\B^{\circ}_{H_i,b_i}(u_i)$.

By Lemma~\ref{lem:dis_oracle_space}, we have $t\leq 4\lceil\log(k) + 1\rceil$.
Since $l\leq t\leq 4\lceil\log(k) + 1\rceil$ and $\dist_G(u,v)=\dist_{H_0}(u_0,v_0)$, $\dist_{G}(u,v)\leq d\leq 26^{4\lceil\log(k) + 1\rceil} \dist_G(u,v)$.
\end{proof}

	\subsection{Low Hop Emulator}\label{sec:low_hop_emulator}
\subsubsection{Distance Preserving Graph with Low Hop Diameter}\label{sec:distance_preserving_graph}
In this section, we construct a new graph with more vertices and edges such that the distance is approximately preserved and there always exists a low hop shortest path between any pair of vertices in the new graph.
In the next section, we will show how to refine the construction to make the new graph be an emulator.

\begin{algorithm}[h]
	\caption{Low Hop Diameter Distance Preserving Graph} \label{alg:low_hop_graph}
	\begin{algorithmic}[1]
		\small
		\Procedure{\textsc{LowHopDimGraph}}{$G=(V,E,w),k$} 
		\State Output: $G'=(V',E',w')$
		\State Run the processing procedure \textsc{PreProc}$(G,k)$, and let $t$ be the value at the end of the procedure. $\forall i\in\{0,1,\cdots,t\},$ let $H_i=(V_i,E_i,w_i)$, $q_i:V_i\rightarrow V_{i+1},B_i:V_i\rightarrow 2^{V_i}$, $b_i$ be computed by the such procedure. \Comment{See Algorithm~\ref{alg:dis_oracle}.}
		\State Initialize $V'\gets \emptyset,E'\gets \emptyset$.
		\State For $v\in V$, if $v\in V_i,v\not\in V_{i+1},$ add $i+1$ copies of $v$ into $V'$, i.e., $V'\gets V'\cup\{v^{(0)},v^{(1)},\cdots,v^{(i)}\}$.
		\State For $i\in\{0,1,\cdots,t-1\}$, for each $v\in V_i$, $E'\gets E'\cup\{v^{(i)},u^{(i+1)}\}$, where $u=q_i(v)$. \label{sta:setup_edge_different_level}
		\State For $i\in\{0,1,\cdots,t\}$, for each $\{u,v\}\in E_i$, $E'\gets E'\cup\{u^{(i)},v^{(i)}\}$. \label{sta:setup_original_edge_same_level}
		\State For $i\in\{0,1,\cdots,t\}$, for each $v\in V_i$, for each $u\in B_i(v)$, $E'\gets E'\cup\{u^{(i)},v^{(i)}\}$.  \label{sta:setup_ball_edge_same_level}
		\State For each $e'\in E'$, initialize $w'(e')\gets \infty$.
		\State For $i\in\{0,1,\cdots,t-1\}$, for each $v\in V_i$, consider $e'=\{v^{(i)},u^{(i+1)}\}$ where $u=q_i(v)$. Let
		\begin{align*}
		w'(e')\gets \min(w'(e'),27^{t-i-1}\cdot \dist_{H_i}(u,v)).
		\end{align*} \label{sta:setup_weight_different_level}
		\State For $i\in\{0,1,\cdots,t\}$, for each $\{u,v\}\in E_i$, consider $e'=\{u^{(i)},v^{(i)}\}$. Let
		\begin{align*}
		w'(e')\gets \min(w'(e'),27^{t-i}\cdot w_i(u,v)).
		\end{align*}\label{sta:setup_original_weight_same_level}
		\State For $i\in\{0,1,\cdots,t\}$, for each $v\in V_i$, for each $u\in B_i(v)$, consider $e'=\{u^{(i)},v^{(i)}\}$. Let
		\begin{align*}
		w'(e')\gets \min(w'(e'),27^{t-i}\cdot \dist_{H_i}(u,v)).
		\end{align*}\label{sta:setup_ball_weight_same_level}
		\State Output $G'=(V',E',w')$.
		\EndProcedure
	\end{algorithmic}
\end{algorithm}

\begin{lemma}[Size of the graph]
Consider a connected undirected weighted graph $G=(V,E,w)$ and  $k\in[0.5,0.5\log n]$, where $n=|V|$. 
Let $G'=(V',E',w')$ be the output of \textsc{LowHopDimGraph}$(G,k)$ (Algorithm~\ref{alg:low_hop_graph}).
Then, $\E[|V'|]\leq O(n)$, $\E[|E'|]\leq O((n^{1+1/(2k)}+n\log^2 n+m)\log k)$.
\end{lemma}
\begin{proof}
For $i\in\{0,1,\cdots,t\}$, let $n_i=|V_i|,m_i=|E_i|$.
We have
\begin{align}
\E[|V'|]=\sum_{i=0}^t \E[n_i]\leq \sum_{i=0}^t\max(n^{1+1/k},n(75\log n)^4)/b_i^2\leq \sum_{i=0}^t\max(n^{1+1/k},n(75\log n)^4)/(b_0^2\cdot 2^i)\leq 2n, \label{eq:bound_size_Vprime}
\end{align}
where the first inequality follows from Lemma~\ref{lem:dis_oracle_space}, the second inequality follows from $b_i/b_0>2^i$, and the last inequality follows from $b_0=\max\left((75\log n)^2,n^{1/(2k)}\right)$.

Consider $|E'|$, we have
\begin{align*}
\E[|E'|]&\leq \sum_{i=0}^{t-1}\E[n_i] + \sum_{i=0}^t\E\left[\sum_{v\in V_i} |B_i(v)|\right]+\sum_{i=0}^t \E[m_i]\\
&\leq 2n + \sum_{i=0}^t\E\left[\sum_{v\in V_i} |B_i(v)|\right]+\sum_{i=0}^t \E[m_i]\\
&\leq 2n + \sum_{i=0}^t \max(n^{1+1/k},n\cdot (75\log n)^4)/b_i+\sum_{i=0}^t \E[m_i]\\
&\leq 2n + \sum_{i=0}^t \max(n^{1+1/k},n\cdot (75\log n)^4)/(2^i\cdot b_0)+\sum_{i=0}^t \E[m_i]\\
&\leq 2n + 2\cdot\max(n^{1+1/(2k)},n\cdot (75\log n)^2)+\sum_{i=0}^t \E[m_i]\\
&\leq 2n + 2\cdot\max(n^{1+1/(2k)},n\cdot (75\log n)^2)+ 4(\lceil\log(k) + 1\rceil\cdot(m+2\cdot \max(n^{1+1/(2k)},n\cdot (75\log n)^2))\\
&\leq 8 \lceil\log (k) + 1\rceil\cdot(m+2\cdot \max(n^{1+1/(2k)},n\cdot (75\log n)^2)),
\end{align*}
where the first inequality follows from line~\ref{sta:setup_edge_different_level}, line~\ref{sta:setup_original_edge_same_level} and line~\ref{sta:setup_ball_edge_same_level}, the second inequality follows from Equation~\eqref{eq:bound_size_Vprime}, the third inequality follows from Lemma~\ref{lem:dis_oracle_space}, the forth inequality follows from $b_i>b_0\cdot 2^i$, the sixth inequality follows from
that $t\leq 4\lceil\log(k)+1\rceil$ and 
\begin{align*}
\E[m_i]\leq m+2\cdot \max(n^{1+1/(2k)},n\cdot (75\log n)^2)
\end{align*}
 by Lemma~\ref{lem:dis_oracle_space}.
\end{proof}

\begin{lemma}[Low hop diameter]\label{lem:low_hop_diameter}
Consider a connected undirected weighted graph $G=(V,E,w)$ and $k\in[0.5,0.5\log n]$, where $n=|V|$. 
Let $G'=(V',E',w')$ be the output of \textsc{LowHopDimGraph}$(G,k)$ (Algorithm~\ref{alg:low_hop_graph}).
Then, $\forall u,v\in V,$
\begin{align*}
\dist_G(u,v)\leq \dist_{G'}(u^{(0)},v^{(0)})\leq 27^{4\lceil\log(k) + 1\rceil}\cdot \dist_G(u,v).
\end{align*} 
Furthermore, $\forall x,y\in V'$,
\begin{align*}
\dist_{G'}(x,y)=\dist_{G'}^{(16\lceil\log(k) + 1\rceil)}(x,y).
\end{align*}
\end{lemma}

\begin{proof}
Consider two vertices $u,v\in V$.
Since $H_0=G$ and the construction of the weights in line~\ref{sta:setup_original_weight_same_level}, we have $\dist_{G'}(u^{(0)},v^{(0)})\leq 27^t\cdot \dist_G(u,v)$.
Due to Lemma~\ref{lem:dis_oracle_space}, $t\leq 4\lceil\log(k) + 1\rceil$.
We have $\dist_{G'}(u^{(0)},v^{(0)})\leq 27^{4\lceil\log(k) + 1\rceil}\cdot \dist_G(u,v)$.
Because $\forall i\in[t]$, $H_i$ is a subemulator of $H_{i-1}$, we have $\dist_{G'}(u^{(0)},v^{(0)})\geq \dist_G(u,v)$ by Definition~\ref{def:subemulator} and the construction of the weights in line~\ref{sta:setup_weight_different_level}, line~\ref{sta:setup_original_weight_same_level} and line~\ref{sta:setup_ball_weight_same_level}.

For a vertex $v^{(i)}\in V'$, we define the level of $v^{(i)}$ as $i$, i.e., $\lev(v^{(i)})=i$.
Consider two arbitrary vertex $x,y\in V'$.
Let $p=(x=z_0,z_1,z_2,\cdots,z_h=y)$ be a shortest path with minimum number of hops between $x$ and $y$ in $G'$.

\begin{claim}\label{cla:continuous_change_of_level}
$\forall j\in[h]$, $|\lev(z_{j-1})-\lev(z_j)|\leq 1$.
\end{claim}
\begin{proof}
It follows directly from the construction of $E'$ in Algorithm~\ref{alg:low_hop_graph}.
\end{proof}

\begin{claim}\label{cla:level_cannot_continue_too_long}
$\forall j\in[h-1]$, either $\lev(z_{j-1})\not = \lev(z_j)$ or $\lev(z_j)\not = \lev(z_{j+1})$.
\end{claim}
\begin{proof}
We prove by contradiction.
Suppose $\lev(z_{j-1})=\lev(z_j)=\lev(z_{j+1})=c$.
Suppose $z_{j-1},z_j,z_{j+1}$ are copies of $u,a,v$ respectively, i.e, $z_{j-1}=u^{(c)},z_j=a^{(c)},z_{j+1}=v^{(c)}$.
By the construction, we know that $w'(z_{j-1},z_j)\geq 27^{t-c} \dist_{H_c}(u,a)$ and $w'(z_j,z_{j+1})\geq 27^{t-c}\dist_{H_c}(a,v)$.
Thus, $w'(z_{j-1},z_j)+w'(z_j,z_{j+1})\geq 27^{t-c}\dist_{H_c}(u,v)$.
There are two cases.
In the first case, either $v\in B_c(u)$ or $u\in B_c(v)$.
In this case, $\{z_{j-1},z_{j+1}\}\in E'$ by line~\ref{sta:setup_ball_edge_same_level} and $w'(z_{j-1},z_{j+1})\leq 27^{t-c}\dist_{H_c}(u,v)$ by line~\ref{sta:setup_ball_weight_same_level}.
Thus, $\{z_{j-1},z_{j+1}\}\in E',w'(z_{j-1},z_{j+1})\leq w'(z_{j-1},z_j)+w'(z_j,z_{j+1})$, and it contradicts to that $p$ has the minimum number of hops.
In the second case, neither $v\in \B^{\circ}_{H_c,b_c}(u)$ nor $u\in \B^{\circ}_{H_c,b_c}(v)$. 
Consider this case.
Let $u'=q_c(u),v'=q_c(v)$. 
We have
\begin{align*}
&w'(z_{j-1},u'^{(c+1)})+\dist_{G'}(u'^{(c+1)},v'^{(c+1)})+w'(v'^{(c+1)},z_{j+1})\\
\leq~& 27^{t-c-1}\cdot (\dist_{H_c}(u,u')+\dist_{H_{c+1}}(u',v')+\dist_{H_c}(v',v))\\
\leq~& 27^{t-c-1}\cdot(2\dist_{H_c}(u,v)+\dist_{H_{c+1}}(u',v'))\\
\leq~& 27^{t-c-1}\cdot(2\dist_{H_c}(u,v)+8\dist_{H_{c}}(u',v'))\\
\leq~& 27^{t-c-1}\cdot(2\dist_{H_c}(u,v)+8(\dist_{H_{c}}(u,u')+\dist_{H_{c}}(u,v)+\dist_{H_{c}}(v',v)))\\
\leq~& 27^{t-c-1}\cdot 26 \dist_{H_c}(u,v)\\
<~& 27^{t-c}\dist_{H_c}(u,v)\\
\leq~& w'(z_{j-1},z_j)+w'(z_j,z_{j+1}),
\end{align*}
where the first inequality follows from the construction of the edges and the corresponding weights in $G'$, the second inequality follows from that $\dist_{H_c}(u,v)\geq \max(\dist_{H_c}(u,u'),\dist_{H_c}(v,v'))$ implied by neither $v\in \B^{\circ}_{H_c,b_c}(u)$ nor $u\in \B^{\circ}_{H_c,b_c}(v)$, the third inequality follows from $H_{c+1}$ is a $(8,b_c)$-subemulator of $H_c$ (see Theorem~\ref{thm:construction_subemulator} and Definition~\ref{def:subemulator}), the forth inequality follows from triangle inequality, the fifth inequality follows from $\dist_{H_c}(u,v)\geq \max(\dist_{H_c}(u,u'),\dist_{H_c}(v,v'))$ again.
In this case, we can find a shorter path which contradicts to that $p$ is the shortest path.
\end{proof}
\begin{claim}\label{cla:level_should_be_a_peek}
$\forall j\in\{0,1,\cdots,h-2\}$, if $\lev(z_{j+1})<\lev(z_j)$, then $\forall j'\in\{j+2,\cdots,h\},\lev(z_{j'})<\lev(z_j)$.
\end{claim}
\begin{proof}
We prove it by contradiction.
Suppose $\lev(z_{j+1})<\lev(z_j)$ and $\exists j''>j$ such that $\lev(z_{j''})\geq \lev(z_j)$.
We can find an $z_{j'}$ such that $\lev(z_{j'})$ is the minimum among $\lev(z_j),\lev(z_{j+1}),$ $\cdots,\lev(z_{j''})$.
Let $f<j'$ be the largest value such that $\lev(z_f)>\lev(z_{j'})$.
Let $g>j'$ be the smallest value such that $\lev(z_g)>\lev(z_{j'})$.
Due to Claim~\ref{cla:continuous_change_of_level}, we have $\lev(z_f)=\lev(z_{f+1})+1=\cdots=\lev(z_{j'})+1=\cdots=\lev(z_{g-1})+1=\lev(z_g)$.
Suppose $\lev(z_{j'})=c$.
Let $z_{f+1},z_{g-1}$ be the copies of $u,v$ respectively, i.e., $z_{f+1}=u^{(c)}$ and $z_{g-1}=v^{(c)}$.
Let $u'=q_c(u),v'=q_c(v)$.
Then $z_f=u'^{(c+1)}$ and $z_g=v'^{(c+1)}$.
We have
\begin{align*}
\dist_{G'}(u',v')&=\sum_{a=f}^{g-1} w'(z_a,z_{a+1})\\
&=27^{t-c-1}\dist_{H_c}(u',u)+27^{t-c}\dist_{H_c}(u,v)+27^{t-c-1}\dist_{H_c}(v,v')\\
&=27^{t-c-1}(\dist_{H_c}(u',u)+\dist_{H_c}(v',v)+27\dist_{H_c}(u,v)).
\end{align*}
According to Theorem~\ref{thm:construction_subemulator}, $H_{c+1}$ is a strong $(8,b_c,22)$-subemulator. 
By Definition~\ref{def:subemulator}, we have $\dist_{H_{c+1}}(u',v')\leq \dist_{H_c}(u',u)+\dist_{H_c}(v',v)+22\dist_{H_c}(u,v)$.
Thus, we have
\begin{align*}
\dist_{G'}(u',v')\leq 27^{t-c-1}\dist_{H_{c+1}}(u',v')
\leq 27^{t-c-1}(\dist_{H_c}(u',u)+\dist_{H_c}(v',v)+22\dist_{H_c}(u,v))
\end{align*}
which leads to a contradiction.
\end{proof}
By Claim~\ref{cla:level_should_be_a_peek}, there must exist $j\in[h]$ such that $\forall j'\in\{0,1,\cdots,j-1\},\lev(z_{j'})\leq \lev(z_{j'+1})$ and $\forall j'\in\{j,j+1,\cdots,h-1\},\lev(z_{j'})\geq\lev(z_{j'+1})$.
Together with Claim~\ref{cla:level_cannot_continue_too_long}, we can conclude that $h\leq 4\cdot t\leq 16\lceil\log(k)+1\rceil$ where the last inequality follows from $t\leq 4\lceil\log(k)+1\rceil$ by lemma~\ref{lem:dis_oracle_space}.
\end{proof}

\subsubsection{Low Hop Emulator}
In this section, we show how to simplify Algorithm~\ref{alg:low_hop_graph} to obtain an emulator which has smaller size than the graph obtained in the previous section.
We have two observations.
The first observation is that line~\ref{sta:setup_original_edge_same_level} and line~\ref{sta:setup_original_weight_same_level} of Algorithm~\ref{alg:low_hop_graph} are useless.
\begin{observation}\label{obs:useless_edges}
Consider a connected undirected weighted graph $G=(V,E,w)$ and $k\in[0.5,0.5\log n]$, where $n=|V|$. Let $G'=(V',E',w')$ be the output of \textsc{LowHopDimGraph}$(G,k)$ (Algorithm~\ref{alg:low_hop_graph}), and let $t,H_i=(V_i,E_i,w_i),B_i(\cdot)$ be the same as described in Algorithm~\ref{alg:low_hop_graph}.
Then for any $i\in\{0,1,\cdots,t\}$ and for any $\{u,v\}\in E_i$, we have either $\dist_{G'}(u^{(i)},v^{(i)})<27^{t-i}\cdot w_i(u,v)$, $u\in B_i(v),$ or $v\in B_i(u)$.
\end{observation}
\begin{proof}
Suppose $u\not\in B_i(v)$ and $v\not\in B_i(u)$.
Let $u'=q_i(u),v'=q_i(v)$.
We have 
\begin{align*}
\dist_{G'}(u^{(i)},v^{(i)})& \leq w'(u^{(i)},u'^{(i+1)}) + \dist_{G'}(u'^{(i+1)},v'^{(i+1)}) + w'(v'^{(i+1)},v^{(i)})\\ 
&\leq 27^{t-i-1}\cdot (\dist_{H_i}(u,u')+\dist_{H_{i+1}}(u',v')+\dist_{H_i}(v',v))\\
&\leq 27^{t-i-1}\cdot (2\dist_{H_i}(u,v)+\dist_{H_{i+1}}(u',v'))\\
&\leq 27^{t-i-1}\cdot (2\dist_{H_i}(u,v)+8\dist_{H_{i}}(u',v'))\\
&\leq 27^{t-i-1}\cdot (2\dist_{H_i}(u,v)+8(\dist_{H_{i}}(u,u')+\dist_{H_{i}}(u,v)+\dist_{H_{i}}(v',v)))\\
&\leq 27^{t-i-1}\cdot 26\dist_{H_i}(u,v)\\
&< 27^{t-i} \cdot\dist_{H_i}(u,v)\\
&\leq 27^{t-i}\cdot w_i(u,v),
\end{align*}
where the second inequality follows from the construction of the edges and their corresponding weights in Algorithm~\ref{alg:low_hop_graph}, the third inequality follows from $\dist_{H_i}(u,u'),\dist_{H_i}(v',v)\leq \dist_{H_i}(u,v)$ since $v\not\in B_i(u)$ and $u\not\in B_i(v)$, the forth inequality follows from that $H_{i+1}$ is a $(8,b_i)$-subemulator of $H_i$ (see Theorem~\ref{thm:construction_subemulator} and Definition~\ref{def:subemulator}), the fifth inequality follows from triangle inequality, the sixth inequality follows from $\dist_{H_i}(u,u'),\dist_{H_i}(v',v)\leq \dist_{H_i}(u,v)$ again.
\end{proof}
If $u\in B_i(v)$ or $v\in B_i(u)$, then by line~\ref{sta:setup_ball_edge_same_level} and line~\ref{sta:setup_ball_weight_same_level} of Algorithm~\ref{alg:low_hop_graph}, there is an edge in $G'$ with weight at most $27^{t-i}\cdot \dist_{H_i}(u,v)\leq 27^{t-i}\cdot w_i(u,v)$. 
Otherwise $\dist_{G'}(u^{(i)},v^{(i)})<27^{t-i}\cdot w_i(u,v)$.
Together with the above observation, we can avoid using the edges constructed by line~\ref{sta:setup_original_edge_same_level} of Algorithm~\ref{alg:low_hop_graph} in any  shortest path.
The second observation is that $\forall i\in [t],v\in V_i$, $w'(v^{(i-1)},v^{(i)})=0$. 
Thus, we can use a single vertex $v$ to denote $v^{(0)},v^{(1)},\cdots,v^{(i)}$.

By removing line~\ref{sta:setup_original_edge_same_level} and line~\ref{sta:setup_original_weight_same_level} of Algorithm~\ref{alg:low_hop_graph}, and contraction of the vertices from different levels, we obtain Algorithm~\ref{alg:low_hop_graph_2}.

\begin{algorithm}[h]
	\caption{Low Hop Emulator} \label{alg:low_hop_graph_2}
	\begin{algorithmic}[1]
		\small
		\Procedure{\textsc{LowHopDimEmulator}}{$G=(V,E,w),k$} 
		\State Output: $G'=(V',E',w')$
		\State Run the processing procedure \textsc{PreProc}$(G,k)$, and let $t$ be the value at the end of the procedure. $\forall i\in\{0,1,\cdots,t\},$ let $H_i=(V_i,E_i,w_i)$, $q_i:V_i\rightarrow V_{i+1},B_i:V_i\rightarrow 2^{V_i}$, $b_i$ be computed by the such procedure. \Comment{See Algorithm~\ref{alg:dis_oracle}.} \label{sta:preproc_in_low_hop_emu}
		\State Initialize $E'\gets \emptyset$.
		\State For $i\in\{0,1,\cdots,t-1\}$, for each $v\in V_i$, $E'\gets E'\cup\{v,u\}$, where $u=q_i(v)$.\label{sta:setup_edge_different_level_emulator}
		\State For $i\in\{0,1,\cdots,t\}$, for each $v\in V_i$, for each $u\in B_i(v)$, $E'\gets E'\cup\{u,v\}$.  \label{sta:setup_original_edge_same_level_emulator}
		\State For each $e'\in E'$, initialize $w'(e')\gets \infty$.
		\State For $i\in\{0,1,\cdots,t-1\}$, for each $v\in V_i$, consider $e'=\{v,u\}$ where $u=q_i(v)$. Let
		\begin{align*}
		w'(e')\gets \min(w'(e'),27^{t-i-1}\cdot \dist_{H_i}(u,v)).
		\end{align*} \label{sta:setup_edge_weight_different_level_emulator}
		\State For $i\in\{0,1,\cdots,t\}$, for each $v\in V_i$, for each $u\in B_i(v)$, consider $e'=\{u,v\}$. Let
		\begin{align*}
		w'(e')\gets \min(w'(e'),27^{t-i}\cdot \dist_{H_i}(u,v)).
		\end{align*}\label{sta:setup_original_edge_weight_same_level_emulator}
		\State Output $G'=(V,E',w')$.
		\EndProcedure
	\end{algorithmic}
\end{algorithm}

\begin{theorem}[Low hop emulator]\label{thm:low_hop_emulator}
	Consider a connected undirected weighted graph $G=(V,E,w)$ and $k\in[0.5,0.5\log n]$, where $n=|V|$. 
	Let $G'=(V,E',w')$ be the output of \textsc{LowHopDimEmulator}$(G,k)$ (Algorithm~\ref{alg:low_hop_graph_2}).
	Then, $\E[|E'|]\leq O(n^{1+1/(2k)}+n\log^2 n)$ and $\forall u,v\in V,$
	\begin{align*}
	\dist_G(u,v)\leq \dist_{G'}(u,v)\leq 27^{4\lceil\log(k) + 1\rceil}\cdot \dist_G(u,v).
	\end{align*} 
	Furthermore, $\forall u,v\in V$,
	\begin{align*}
	\dist_{G'}(u,v)=\dist_{G'}^{(16\lceil\log(k) + 1\rceil)}(u,v).
	\end{align*}
\end{theorem}

\begin{proof}
	Consider $|E'|$, we have
\begin{align*}
\E[|E'|]&\leq \sum_{i=0}^{t-1}\E[n_i] + \sum_{i=0}^t\E\left[\sum_{v\in V_i} |B_i(v)|\right]\\
&\leq 2n + \sum_{i=0}^t\E\left[\sum_{v\in V_i} |B_i(v)|\right]\\
&\leq 2n + \sum_{i=0}^t \max(n^{1+1/k},n\cdot (75\log n)^4)/b_i\\
&\leq 2n + \sum_{i=0}^t \max(n^{1+1/k},n\cdot (75\log n)^4)/(2^i\cdot b_0)\\
&\leq 2n + 2\cdot\max(n^{1+1/(2k)},n\cdot (75\log n)^2),
\end{align*}
where the first inequality follows from line~\ref{sta:setup_edge_different_level_emulator} and line~\ref{sta:setup_original_edge_same_level_emulator} of Algorithm~\ref{alg:low_hop_graph_2}, the second inequality follows from Equation~\eqref{eq:bound_size_Vprime}, the third inequality follows from Lemma~\ref{lem:dis_oracle_space}, the forth inequality follows from $b_i>b_0\cdot 2^i$.	
	
The only two differences between Algorithm~\ref{alg:low_hop_graph} and Algorithm~\ref{alg:low_hop_graph_2} are as follows.
The first difference is that we remove line~\ref{sta:setup_original_edge_same_level} and line~\ref{sta:setup_original_weight_same_level} from Algorithm~\ref{alg:low_hop_graph}.
This change does not affect $\dist_{G'}(u,v)$ for any $u,v\in V$ because of Observation~\ref{obs:useless_edges}.
The second difference is that we contract $v^{(0)},v^{(1)},\cdots,$ to vertex $v$.
We can do this operation because we have $w'(v^{(i-1)},v^{(i)})=0$ for any $v^{(i-1)},v^{(i)}$ in Algorithm~\ref{alg:low_hop_graph}.
Then the statement follows directly from Lemma~\ref{lem:low_hop_diameter}
\end{proof}
\section{Uncapacitated Minimum Cost Flow}\label{sec:min_cost_flow}
Given an undirected graph $G=(V,E,w)$ with $|V|=n$ vertices and $|E|=m$ edges, the vertex-edge incidence matrix $A\in\mathbb{R}^{n\times m}$ is defined as the following: 
\begin{align*}
\forall i\in[n],j\in[m], A_{i,j}=\left\{
\begin{array}{ll}
1 & \{i,v\}\in E\text{ is the }j\text{-th edge of }G \text{ and }i<v,\\
-1 & \{i,v\}\in E\text{ is the }j\text{-th edge of }G \text{ and } i>v,\\
0 & \text{Otherwise.}
\end{array}
\right.
\end{align*}
The weight matrix $W\in\mathbb{R}^{m\times m}$ is a diagonal matrix.
The $i$-th diagonal entry of $W$ is $w(e)$, where $e\in E$ is the $i$-th edge.
Given a demand vector $b\in\mathbb{R}^n$ with $\one_n^\top b=0,$ i.e., $\sum_{i=1}^n b_i=0$, the uncapacitated minimum cost flow (transshipment) problem is to solve the following problem:
\begin{align*}
&\min_{f\in\mathbb{R}^m} \|Wf\|_1\\
s.t.~&~Af = b.
\end{align*}
If $b$ only has two non-zero entries $b_i=1$ and $b_j=-1$, then the optimal cost is the length of the shortest path between vertex $i$ and vertex $j$.
Without loss of generality, we can suppose that each edge has positive weight. 
Otherwise, we can contract the edges with weight $0$, and the contraction will not affect the value of the solution.
Let $x=Wf$, then the problem becomes
\begin{align}
&\min_{x\in\mathbb{R}^m} \|x\|_1 \label{eq:min_cost_flow}\\
s.t.~&~AW^{-1}x =b. \notag
\end{align}
In this section, we will focus on finding a $(1+\varepsilon)$-approximation to problem~\eqref{eq:min_cost_flow}.

\subsection{Sherman's Framework}
Before we present our algorithm, let us review Sherman's algorithm~\cite{s17}, and completely open his black box.

\begin{definition}[$\ell_1$ Non-linear condition number]\label{def:condition_number}
Given a matrix $B\in\mathbb{R}^{r\times m}$, the $\ell_1$ non-linear condition number of $B$ is defined as 
\begin{align*}
\kappa(B) =  \inf_S \|B\|_{1\rightarrow 1}\cdot \sup_{x\in\mathbb{R}^m:Bx\not=0} \frac{\|S(Bx)\|_1}{\|Bx\|_1},
\end{align*}
where the range of $S:\mathbb{R}^r\rightarrow \mathbb{R}^m$ is over all maps such that $\forall x\in\mathbb{R}^m$, $B\cdot S(Bx)=Bx$.
\end{definition}
By above definition, an alternative way to define $\kappa(B)$ is as the following:
\begin{align*}
\kappa(B)=\|B\|_{1\rightarrow 1}\cdot \max_{g\in \{y\in\mathbb{R}^r\mid y=Bx,x\in\mathbb{R}^m\}\setminus\{0\}} \min_{x:Bx=g} \frac{\|x\|_1}{\|g\|_1}.
\end{align*}

\begin{definition}[$(\alpha,\beta)$-Solution]\label{def:alpha_beta_solver}
Given a matrix $B\in\mathbb{R}^{r\times m}$ and a vector $g\in\{y\in\mathbb{R}^r\mid y=Bx,x\in\mathbb{R}^m\}$, let $x^*=\arg\min_{x:Bx=g}\|x\|_1$. 
If $\|x\|_1\leq \alpha \|x^*\|_1$ and $\|Bx-g\|_1\leq \beta \|B\|_{1\rightarrow 1}\|x^*\|_1$, then $x$ is called an $(\alpha,\beta)$-solution with respect to $(B,g)$.
Given a matrix $B\in\mathbb{R}^{r\times m}$, if an algorithm can output an $(\alpha,\beta)$-solution with respect to $(B,g)$ for any vector $g\in\{y\in\mathbb{R}^r\mid y=Bx,x\in\mathbb{R}^m\}$, then the algorithm is called an $(\alpha,\beta)$-solver for $B$.
\end{definition}

\begin{definition}[Composition of the solvers]
Suppose $F_1$ is an $(\alpha_1,\beta_1)$-solver for $B\in\mathbb{R}^{r\times m}$ and $F_2$ is an $(\alpha_2,\beta_2)$-solver for $B$.
For any input vector $g\in\{y\in\mathbb{R}^r\mid y=Bx,x\in\mathbb{R}^m\}$, the composition $F_2\circ F_1$ firstly runs $F_1$ to obtain an $(\alpha_1,\beta_1)$-solution $x\in\mathbb{R}^m$ with respect to $(B,g)$, then runs $F_2$ to obtain an $(\alpha_2,\beta_2)$-solution $x'\in\mathbb{R}^m$ with respect to $(B,g-Bx)$, and finally outputs $x+x'$.
\end{definition}

\begin{lemma}[\cite{s17}]
Suppose $F_1$ is an $(\alpha_1,\beta_1/\kappa)$-solver for $B\in\mathbb{R}^{r\times m}$ and $F_2$ is an $(\alpha_2,\beta_2/\kappa)$-solver for $B$, where $\kappa$ is the $\ell_1$ non-linear condition number of $B$, i.e., $\kappa=\kappa(B)$.
Then $F_2\circ F_1$ is an $(\alpha_1+\alpha_2\beta_1,\beta_1\beta_2/\kappa)$-solver for $B$.
\end{lemma}

\begin{corollary}[\cite{s17}]\label{cor:boost_error}
Let $\varepsilon\in(0,0.5)$.
Suppose $F$ is an $(1+\varepsilon,\varepsilon/\kappa)$ solver for $B\in\mathbb{R}^{r\times m}$, where $\kappa$ is the $\ell_1$ non-linear condition number of $B$, i.e., $\kappa=\kappa(B)$.
Define $F^1=F$, and $F^t=F^{t-1}\circ F$.
Then $F^t$ is an $(1+4\varepsilon,\varepsilon^t/\kappa)$ solver.
\end{corollary}

\begin{corollary}[\cite{s17}]\label{cor:exact_solver}
Let $\varepsilon\in(0,0.5),t,M\in\mathbb{R}_{\geq 0}$. Suppose $F_1$ is an $(1+4\varepsilon,\varepsilon^t/\kappa)$-solver for $B\in\mathbb{R}^{r\times m},$ and $F_2$ is an $(M,0)$-solver for $B$, where $\kappa=\kappa(B)$.
Then $F_2\circ F_1$ is an $(1+4\varepsilon+M\varepsilon^t,0)$-solver for $B$.
\end{corollary}

Let us come back to the minimum cost flow problem, problem~\eqref{eq:min_cost_flow}. 
One observation is that if a matrix $P\in\mathbb{R}^{r\times m}$ has full column rank, then $PAW^{-1}x=Pb\Leftrightarrow AW^{-1}x =b$.
So, instead of solving Equation~\eqref{eq:min_cost_flow} directly, we can design a matrix $P\in\mathbb{R}^{r\times m}$ with full column rank, and try to solve 
\begin{align}
&\min_{x\in\mathbb{R}^m} \|x\|_1\label{eq:preconditioned_min_cost_flow}\\
s.t.~&~PAW^{-1}x = Pb.\notag
\end{align}
Notice that since $P$ has full column rank, problem~\eqref{eq:preconditioned_min_cost_flow} is exactly the same as problem~\eqref{eq:min_cost_flow}.
Although an $(\alpha,0)$-solver for $PAW^{-1}$ is also an $(\alpha,0)$-solver for $AW^{-1}$, an $(\alpha,\beta)$-solver for $PAW^{-1}$ may not be an $(\alpha,\beta)$-solver for $AW^{-1}$ for $\beta>0$. 
As shown in~\cite{s17}, if $\kappa(PAW^{-1})$ is smaller, then it is much easier to design a $(1+\varepsilon,\varepsilon/\kappa(PAW^{-1}))$-solver for $PAW^{-1}$.
If $\kappa(PAW^{-1})$ is small, then we say $P$ is a good preconditioner for $AW^{-1}$.
Before we discuss how to construct $P$, let us assume $\kappa(PAW^{-1})\leq \kappa$, and review how to solve problem~\eqref{eq:preconditioned_min_cost_flow}.

As introduced by~\cite{s17}, there is a simple $(n,0)$-solver to problem~\eqref{eq:preconditioned_min_cost_flow}.

\begin{algorithm}[h]
	\caption{An $(n,0)$-Solver}\label{alg:mst_routing}
	\begin{algorithmic}[1]
		\small
		\Procedure{\textsc{MSTRouting}}{$G=(V,E,w),b\in\mathbb{R}^{n}$} 
		\State Output: $f\in\mathbb{R}^{|E|}$
		\State Compute a minimum spanning tree $T=(V,E',w)$ of $G$.
		\State Choose an arbitrary vertex as the root of $T$. Initialize $f\in\mathbb{R}^m$.
		\State Consider the $i$-th edge $\{u,v\}\in E$ $(u<v)$. If $\{u,v\}\not\in E'$, set $f_i\gets0$.
		Otherwise, if $u$ is the parent of $v$, set $f_i\gets-\sum_{z\text{ is in the subtree of }v} b_z$; otherwise, set $f_i\gets\sum_{z\text{ is in the subtree of }u} b_z$. \label{sta:subtree_sum}
		\State Return $f$.
		\EndProcedure
	\end{algorithmic}
\end{algorithm}
\begin{lemma}[\cite{s17}]\label{lem:exact_solver_with_large_approx}
Given a connected undirected weighted graph $G=(V,E,w)$, let $A\in\mathbb{R}^{n\times m}$ be the corresponding vertex-edge incidence matrix, and let $W\in\mathbb{R}^{m\times m}$ be the corresponding diagonal weight matrix, where $n=|V|,m=|E|$. 
For any demand vector $b\in\mathbb{R}^n$ with $\one_n^{\top} b = 0$, the output $f\in\mathbb{R}^m$ of \textsc{MSTRouting}$(G,b)$ (Algorithm~\ref{alg:mst_routing}) satisfies $Af=b$ and $\|Wf\|_1\leq n\cdot \min_{f':Af'=b}\|Wf'\|_1$.
\end{lemma}
By above lemma, if we set $x=Wf$, we have
$PAW^{-1}x=Pb$, and $\|x\|_1\leq n\cdot \min_{x':PAW^{-1}x'=Pb} \|x'\|_1$.
Thus, $x$ is an $(n,0)$-solution to problem~\eqref{eq:preconditioned_min_cost_flow}.
Suppose $\varepsilon<0.5$. 
By Corollary~\ref{cor:exact_solver}, if we have a $(1+4\varepsilon, \varepsilon^{1+\log n}/\kappa)$ solver for $PAW^{-1}$, then together with Lemma~\ref{lem:exact_solver_with_large_approx}, we can obtain a $(1+5\varepsilon,0)$-solver for $PAW^{-1}$, and thus we can finally find a $(1+5\varepsilon)$ approximation to problem~\eqref{eq:min_cost_flow}. 
If we have a $(1+\varepsilon,\varepsilon/\kappa)$-solver for $PAW^{-1}$, then  according to Corollary~\ref{cor:boost_error}, we can apply $(1+\varepsilon,\varepsilon/\kappa)$-solver $1+\log n$ times to obtain a $(1+4\varepsilon,\varepsilon^{1+\log n}/\kappa)$-solver.
It suffices to design a $(1+\varepsilon,\varepsilon/\kappa)$ solver for $PAW^{-1}$.

\subsubsection{A $(1+\varepsilon,\varepsilon/\kappa)$-Solver}
In this section, we will have a detailed discussion of how \cite{s17,knp19} used multiplicative weights update algorithm~\cite{ahk12} to find a $(1+\varepsilon,\varepsilon/\kappa)$-solution with respect to $(PAW^{-1},Pb)$, where $\kappa\geq \kappa(PAW^{-1})$ is an upper bound of the condition number (see Definition~\ref{def:condition_number}) of $PAW^{-1}$, and $\varepsilon\in(0,0.5)$ is an arbitrary real number.

Let $x^*=\arg\min_{x:PAW^{-1}x=Pb}\|x\|_1$.
We have
\begin{align*}
\frac{\|Pb\|_1}{\|PAW^{-1}\|_{1\rightarrow 1}}\leq \|x^*\|_1\leq \kappa\cdot  \frac{\|Pb\|_1}{\|PAW^{-1}\|_{1\rightarrow 1}},
\end{align*}
where the lower bound of $\|x^*\|_1$ follows from that $PAW^{-1}x^*=Pb$ and the definition of the operator $\ell_1$ norm, and the upper bound of $\|x^*\|_1$ follows from the definition of condition number (see Definition~\ref{def:condition_number}) and thus
$
\|PAW^{-1}\|_{1\rightarrow 1}\cdot \frac{\|x^*\|_1}{\|Pb\|_1}\leq \kappa(PAW^{-1})\leq \kappa.
$
Then, we can reduce the optimization problem to a feasibility problem.
We want to binary search $s\in\{1,1+\varepsilon,(1+\varepsilon)^2,\cdots,(1+\varepsilon)^{\lceil\log_{1+\varepsilon} \kappa\rceil}\}$, and 
want to find $s$ such that $s\cdot\frac{\|Pb\|_1}{\|PAW^{-1}\|_{1\rightarrow 1}}\leq (1+\varepsilon)\|x^*\|_1$ and find $x\in\mathbb{R}^{m}$ which satisfies $\|x\|_1\leq s\cdot\frac{\|Pb\|_1}{\|PAW^{-1}\|_{1\rightarrow 1}}$ and $\|PAW^{-1}x-Pb\|_1\leq \frac{\varepsilon}{2\kappa}\cdot \|PAW^{-1}\|_{1\rightarrow 1}\cdot  s\cdot \frac{\|Pb\|_1}{\|PAW^{-1}\|_{1\rightarrow 1}}$.
The binary search will takes $O(\log(\log_{1+\varepsilon}\kappa))$ rounds.

Now the problem becomes the following feasibility problem:
given $s\geq 1$, either find $x\in\mathbb{R}^m$ such that
\begin{align*}
\begin{array}{ccc}
\|x\|_1\leq s\cdot \frac{\|Pb\|_1}{\|PAW^{-1}\|_{1\rightarrow 1}} & \text{ and }& \|PAW^{-1}x-Pb\|_1\leq  \frac{\varepsilon}{2\kappa}\cdot \|PAW^{-1}\|_{1\rightarrow 1}\cdot  s\cdot \frac{\|Pb\|_1}{\|PAW^{-1}\|_{1\rightarrow 1}},
\end{array}
\end{align*}
or find a certificate such that 
\begin{align*}
\begin{array}{ccc}
\|x\|_1\leq s\cdot \frac{\|Pb\|_1}{\|PAW^{-1}\|_{1\rightarrow 1}} & \text{ and }& PAW^{-1}x=Pb
\end{array}
\end{align*}
 is not feasible. 
Let $x'= x\cdot \frac{\|PAW^{-1}\|_{1\rightarrow 1}}{\|Pb\|_1}\cdot \frac{1}{s}$.
Then we have the following equivalent feasibility problem:
given $s\geq 1$, either find $x'\in\mathbb{R}^m$ such that 
\begin{align}\label{eq:feas2}
\begin{array}{ccc}
\|x'\|_1\leq 1& \text{ and }& \left\| \frac{PAW^{-1}}{\|PAW^{-1}\|_{1\rightarrow 1}} x' - \frac1s\cdot \frac{Pb}{ \|Pb\|_1} \right\|_1\leq \frac{\varepsilon}{2\kappa}
\end{array}
\end{align}
or find a certificate such that 
\begin{align}\label{eq:feas2_ex}
\begin{array}{ccc}
\|x'\|_1\leq 1& \text{ and }&  \frac{PAW^{-1}}{\|PAW^{-1}\|_{1\rightarrow 1}} x' = \frac1s\cdot \frac{Pb}{ \|Pb\|_1} 
\end{array}
\end{align}
 is not feasible.

Next, we will show how to use multiplicative weights update algorithm~\cite{ahk12,s17,knp19} to solve problem~\eqref{eq:feas2}-\eqref{eq:feas2_ex}.

\begin{algorithm}[h]
	\caption{Solving the Feasibility Problem}\label{alg:feasibility}
	\begin{algorithmic}[1]
		\small
		\Procedure{\textsc{MWU}}{$P\in\mathbb{R}^{r\times n},A\in\mathbb{R}^{n\times m},W\in\mathbb{R}^{m\times m},b\in\mathbb{R}^{n},s\geq 1,\varepsilon\in(0,0.5),\kappa\geq 1$} 
		\State Output: $x'\in\mathbb{R}^m$
		\State Initialize weights: $\forall i\in[m],\psi^+_1(i)\gets 1,\psi^-_1(i)\gets 1$.
		\State Initialize $T\gets  \frac{64\kappa^2\ln(2m)}{\varepsilon^2},\eta\gets \frac{\varepsilon}{8\kappa},$ $B\in\mathbb{R}^{n\times 2m}:$
		\begin{align*}
		B\gets \left(\begin{matrix}
		\frac{AW^{-1}}{\|PAW^{-1}\|_{1\rightarrow 1}}-\frac{1}{s}\cdot\frac{b\cdot \one_m^\top}{\|Pb\|_1}
		&
		-\frac{AW^{-1}}{\|PAW^{-1}\|_{1\rightarrow 1}}-\frac{1}{s}\cdot\frac{b\cdot \one_m^\top}{\|Pb\|_1}
		\end{matrix}\right).
		\end{align*}
		\For{$t = 1\rightarrow T$}
		\State $\Psi_t\gets\sum_{i=1}^m \psi^+_{t}(i)+\sum_{i=1}^m\psi^-_{t}(i).$
		\State For $i\in [m]$, $p_t^+(i)\gets \psi^+_t(i)/\Psi_t, p_t^-(i)\gets \psi^-_t(i)/\Psi_t$.
		\State Set $p_t\in\mathbb{R}^{2m}$ s.t. $\forall i\in[m]$, the $i$-th entry of $p_t$ is $p_t^+(i)$, and the $(i+m)$-th entry of $p_t$ is $p_t^-(i)$.
		\State If
		$
		\left\|PBp_t\right\|_1\leq \frac{\varepsilon}{2\kappa},
		$
		return $x'\in\mathbb{R}^m$ such that  $\forall i\in[m],x'_i=p_t^+(i)-p_t^-(i)$. \label{sta:return_xprime}
		\State Otherwise, set $y_t\in\{+1,-1\}^r$ such that 
		$\forall i\in[r],\left(y_t\right)_i=\sgn\left(\left(PBp_t\right)_i\right)$.
		\State For $i\in [m]$, 
		$\phi_t^+(i)\gets y_t^\top PB_i/2,\phi_t^-(i)\gets y_t^\top PB_{i+m}/2$.
		
		\State For $i\in[m]$, $\psi_{t+1}^+(i)\gets \psi_{t}^+(i)\cdot \left(1-\eta \phi^+_{t}(i)\right), \psi_{t+1}^-(i)\gets \psi_{t}^-(i)\cdot \left(1-\eta \phi^-_{t}(i)\right)$.
		\EndFor
		\State Return \textrm{FAIL}.
		\EndProcedure
	\end{algorithmic}
\end{algorithm}

\begin{lemma}[\cite{s17,knp19}]\label{lem:mwu_feas_prob}
Consider $P\in\mathbb{R}^{r\times n},A\in\mathbb{R}^{n\times m},W\in\mathbb{R}^{m\times m},b\in\mathbb{R}^n,s\geq 1,\varepsilon\in (0,0.5),\kappa\geq 1$.
\textsc{MWU}$(P,A,W,b,s,\varepsilon,\kappa)$ (Algorithm~\ref{alg:feasibility}) takes $T=O(\kappa^2\varepsilon^{-2}\log m)$ iterations.
If \textsc{MWU}$(P,A,W,b,s,\varepsilon,\kappa)$ does not return FAIL, the output $x'\in\mathbb{R}^m$ satisfies Equation~\eqref{eq:feas2}.
Otherwise, $\bar{y}=\frac{1}{T}\sum_{t=1}^T y_t$ is a certificate that Equation~\eqref{eq:feas2_ex} is not feasible.
In particular,
\begin{align*}
\begin{array}{ccc}
\forall j\in[m], & \frac{1}{s} \cdot \frac{Pb}{\|Pb\|_1}<  \frac{\bar{y}^\top(PAW^{-1})_j}{\|PAW^{-1}\|_{1\rightarrow 1}}, & \frac{1}{s} \cdot \frac{Pb}{\|Pb\|_1}<  -\frac{\bar{y}^\top(PAW^{-1})_j}{\|PAW^{-1}\|_{1\rightarrow 1}}.
\end{array}
\end{align*}
\end{lemma}
For completeness, we put the proof of Lemma~\ref{lem:mwu_feas_prob} into Section~\ref{sec:proof_of_mwu_feas_prob}.

\subsection{Preconditioner Construction}
As discussed in the previous section, if we can find a good preconditioner such that $\kappa(PAW^{-1})$ is small, then we can use a small number of iterations to compute a good solution.
Before we describe how to choose a good preconditioner, let us introduce the following Lemma.
\begin{lemma}[\cite{s17,knp19}]\label{lem:property_of_preconditioner}
Given $P\in\mathbb{R}^{r\times n}$ with full column rank, $A\in\mathbb{R}^{n\times m}, W\in\mathbb{R}^{m\times m}$, if $\forall b\in\{y\in\mathbb{R}^n\mid y=AW^{-1}x,x\in\mathbb{R}^m\}$,
\begin{align*}
\|x^*\|_1\leq \|Pb\|_1\leq \gamma \|x^*\|_1,
\end{align*}
where $x^*=\arg\min_{x\in\mathbb{R}^m:AW^{-1}x=b}\|x\|_1$, then $\kappa(PAW^{-1})\leq \gamma$.
\end{lemma}
\begin{proof}
For any $x\in\mathbb{R}^m$, $\|PAW^{-1}x\|_1\leq \gamma \|x\|_1$.
Thus, $\|PAW^{-1}\|_{1\rightarrow 1}\leq \gamma$.
By Definition~\ref{def:condition_number}, we have:
\begin{align*}
\kappa(PAW^{-1})=\|PAW^{-1}\|_{1\rightarrow 1}\cdot \max_{b\in\{y\in\mathbb{R}^n\mid y=AW^{-1}x,x\in\mathbb{R}^m\}:Pb\not = 0} \min_{x\in\mathbb{R}^m:PAW^{-1}x=Pb}\frac{\|x\|_1}{\|Pb\|_1}\leq \gamma,
\end{align*}
where the inequality follows from
$\|PAW^{-1}\|_{1\rightarrow 1}\leq \gamma$, $P$ has full column rank and $\forall b\in\{y\in\mathbb{R}^n\mid y=AW^{-1}x,x\in\mathbb{R}^m\}$,
\begin{align*}
\min_{x\in\mathbb{R}^m:AW^{-1}x=b} \|x\|_1\leq \|Pb\|_1.
\end{align*}
\end{proof}
By above lemma, our goal is to find a linear operator $P$ such that for any demand vector $b$, $\|Pb\|_1$ can approximate the minimum cost flow with demand vector $b$ very well.
Instead of using Sherman's original lattice algorithm,
we propose to use randomly shifted grids based algorithm~\cite{it03}.

\subsubsection{Embedding Minimum Cost Flow into $\ell_1$ via Randomly Shifted Grids}
In this section, we review the embedding method of~\cite{it03} and describe how to construct the preconditioner.
Suppose we have a mapping $\varphi:V\rightarrow [\Delta]^d$ such that $\forall u,v\in V$,
\begin{align*}
\dist_{G}(u,v)\leq \|\varphi(u)-\varphi(v)\|_1\leq \alpha\cdot \dist_G(u,v).
\end{align*}
We can reduce estimating the minimum cost flow on $G$
to approximating the cost of the geometric transportation problem. 
The geometric transportation problem is also called Earth Mover's Distance (EMD) problem.
In particular, it is the following minimization problem:
\begin{align}\label{eq:emd}
&\min_{\pi:V\times V\rightarrow \mathbb{R}_{\geq 0}} \sum_{(u,v)\in V\times V} \pi(u,v)\cdot \|\varphi(u)-\varphi(v)\|_1\\
s.t.~&\forall u\in V, \sum_{v\in V}\pi(u,v) - \sum_{v\in V}\pi(v,u) = b_u.\notag
\end{align}
It is obvious that if we can obtain a $\beta$-approximation to the optimal cost of \eqref{eq:emd}, we can obtain an $\alpha\beta$-approximation to the cost of original minimum cost flow problem on $G$.

For a sequential algorithm, the such embedding $\varphi$ can be obtained by Bourgain's Embedding.
\begin{lemma}[Bourgain's Embedding~\cite{b85}]\label{lem:bourgain}
Given an undirected graph $G=(V,E,w)$ with $|V|=n$ vertices and $|E|=m$ edges, there is a randomized algorithm which can output a mapping $\varphi:V\rightarrow [\Delta]^d$ for $d=O(\log^2 n)$ with probability $0.99$ in $O(m\log^2 n)$ time, such that 
\begin{align*}
\forall u,v\in V, \dist_{G}(u,v)\leq \|\varphi(u)-\varphi(v)\|_1\leq O(\log n)\cdot \dist_G(u,v),
\end{align*}
where $\Delta\leq\sum_{e\in E} w(e)$.
\end{lemma}
In the remaining of this section, we focus on approximating~\eqref{eq:emd}.
Without loss of generality, we suppose $\Delta$ is a power of $2$.
Let $L=1+\log \Delta$.
We create $L$ levels grids $G_0,G_1,\cdots,G_{L-1}$, where $G_i$ partitions $[2\Delta]^d$ into disjoint cells with side length $2^i$. 
In particular, $\forall i\in\{0,1\cdots,L-1\},$ the $i$-th level grid $G_i$ is:
\begin{align*}
 \left\{C ~\big|~ C = \{a_1,\cdots,a_1+2^i-1\}\times\cdots\times \{a_d,\cdots,a_d+2^i-1\}, \forall j\in[d],a_j\bmod 2^i=1, a_j\in[2\Delta]\right\}.
\end{align*}
Instead of shifting the gird, we shift the points.
For each dimension, we can use the same shift value $\tau$~\cite{bi14}.
Let $\tau$ be a random variable with uniform distribution over $[\Delta]$.
We can construct a vector $h\in\mathbb{R}^{\sum_{i=0}^{L-1}|G_i|}$ with one entry per cell in $G_0\cup G_1\cup\cdots \cup G_{L-1}$.
Let $h_{(i,C)}$ correspond to the cell $C\in G_i$.
For each $i\in\{0,1,\cdots,L-1\}$ and each cell $C\in G_i$,
we set $h_{(i,C)}$ as:
\begin{align*}
h_{(i,C)} = d\cdot 2^i\cdot \sum_{v\in V:\varphi(v)+\tau\cdot \one_d\in C} b_v.
\end{align*}
Let $\OPT_{\EMD}(b)$ denote the optimal solution of the EMD problem~\eqref{eq:emd}.
As shown by~\cite{it03}, $\|h\|_1$ is a good approximation to $\OPT_{\EMD}(b)$.
\begin{lemma}\label{lem:random_shift_embed}
Let $h\in\mathbb{R}^{\sum_{i=0}^{L-1}|G_i|}$ be constructed as above. Then,
\begin{enumerate}
\item $\E_{\tau}[\|h\|_1]\leq 2Ld\cdot \OPT_{EMD}$,
\item $\|h\|_1\geq \OPT_{\EMD}(b)$. 
\end{enumerate}
\end{lemma}
\begin{proof}
Consider the upper bound. 
Let $\pi^*:V\times V\rightarrow \mathbb{R}_{\geq 0}$ be the optimal solution of problem~\eqref{eq:emd}.
\begin{align*}
\E_{\tau}[\|h\|_1] & \leq \sum_{i=0}^{L-1} \sum_{(u,v)\in V\times V} 2\cdot d\cdot 2^i\cdot \pi^*(u,v)\cdot \Pr_{\tau}\left[\varphi(u)+\tau\cdot \one_d\text{ and }\varphi(v)+\tau\cdot \one_d\text{ are in different cells of }G_i\right]\\
& \leq \sum_{i=0}^{L-1}\sum_{(u,v)\in V\times V} 2\cdot d\cdot 2^i\cdot \pi^*(u,v) \cdot \sum_{j=1}^d \frac{|\varphi(u)_j-\varphi(v)_j|}{2^i}\\
&\leq 2Ld\cdot \sum_{(u,v)\in V\times V} \pi^*(u,v) \|\varphi(u)-\varphi(v)\|_1\\
& = 2Ld\cdot \OPT_{\EMD}(b),
\end{align*}
where the second step follows from union bound on all dimensions.

Consider the lower bound. 
We can build a tree with one node per cell in $G_0\cup G_1\cup\cdots\cup G_{L-1}$.
For a cell $C\in G_i$, there is a unique cell $C'\in G_{i+1}$ such that $C\subset C'$. 
We connect the nodes corresponding to $C$ and $C'$ with an edge of which weight is $d\cdot 2^i$.
For $u,v\in V$, there are two cells $C_1,C_2\in G_0$ such that $\varphi(u)+\tau\cdot \one_d\in C_1$ and $\varphi(v)+\tau\cdot \one_d\in C_2$.
The distance between two nodes corresponding to $C_1,C_2$ on the tree is at least $\|\varphi(u)-\varphi(v)\|_1$.
The cost of the minimum cost flow on the such tree is 
\begin{align*}
\sum_{i=0}^{L-1}\sum_{C\in G_i} d\cdot 2^i\cdot\left|\sum_{v\in V:\varphi(v)+\tau\cdot \one_d\in C} b_v\right| = \|h\|_1.
\end{align*}
Thus, $\|h\|_1\geq \OPT_{\EMD}(b)$.
\end{proof}

An observation is that since each cell in $G_i$ has side length $2^i$, shifting each point by $\tau\cdot \one_d$ is equivalent to shifting each point by $(\tau\bmod 2^i)\cdot \one_d$ for the cells in $G_i$. Thus, if we modify the construction of $h$ as the following:
\begin{align*}
\forall i\in\{0,1,\cdots,L-1\}, C\in G_i, h_{(i,C)} =d\cdot 2^i\cdot \sum_{v\in V:\varphi(v)+(\tau\bmod 2^i)\cdot \one_d\in C} b_v,
\end{align*}
Lemma~\ref{lem:random_shift_embed} still holds.
Next, we describe how to construct $h'\in\mathbb{R}^{\sum_{i=0}^{L-1}2^i|G_i|}$.
The entry $h'_{(i,C,\tau)}$ corresponds to the cell $C\in G_i$ and the shift value $\tau$. 
For each $i\in\{0,1,\cdots,L-1\}$, each cell $C\in G_i$ and each shift value $\tau\in [2^i]$, we set $h'_{(i,C,\tau)}$ as:
\begin{align*}
h'_{(i,C,\tau)}=\frac{1}{2^i}\cdot d\cdot 2^i\cdot \sum_{v\in V:\varphi(v)+\tau\cdot \one_d \in C} b_v=d\cdot \sum_{v\in V:\varphi(v)+\tau\cdot \one_d\in C} b_v.
\end{align*}
It is clear that $\|h'\|_1=\E[\|h\|_1]$.
By Lemma~\ref{lem:random_shift_embed}, we have
\begin{align*}
\OPT_{\EMD}(b)\leq \|h'\|_1\leq 2Ld\cdot \OPT_{\EMD}(b).
\end{align*}
Observe that $h'$ can be written as a linear map of $b$, i.e., $h'=P'b$, where $P'\in\mathbb{R}^{(\sum_{i=0}^{L-1}2^i|G_i|)\times n}$.
Each row of $P'$ is indexed by a tuple $(i,C,\tau)$ for $i\in\{0,1,\cdots,L-1\},C\in G_i$ and $\tau\in[2^i]$, and each column of $P'$ is indexed by a vertex $v\in V$.
For $i\in\{0,1,\cdots,L-1\},C\in G_i,\tau\in[2^i],v\in V$,
\begin{align*}
P'_{(i,C,\tau),v}=\left\{\begin{array}{ll}d & \varphi(v)+\tau\cdot \one_d\in C,\\ 0 & \text{Otherwise.}\end{array}\right.
\end{align*}
Consider $i=0,\tau=1$, $\forall v\in V$, there is a unique cell $C\in G_0$ which contains $\varphi(v)+\one_d$.
Thus, $P'$ has full column rank.
According to Lemma~\ref{lem:property_of_preconditioner}, since $\forall b\in\{y\in\mathbb{R}^n\mid y=AW^{-1}x,x\in\mathbb{R}^m\}$, 
\begin{align*}
\min_{x\in\mathbb{R}^m:AW^{-1}x=b}\|x\|_1\leq \OPT_{\EMD}(b)\leq \|P'b\|_1\leq 2Ld\cdot\OPT_{\EMD}(b)\leq 2Ld\alpha\cdot\min_{x\in\mathbb{R}^m:AW^{-1}x=b} \|x\|_1,
\end{align*}
we have $\kappa(P'AW^{-1})\leq 2Ld\alpha$.
However, since the size of $P'$ is too large, we cannot apply $P'$ directly in Algorithm~\ref{alg:feasibility}, and thus it is unclear how to construct a $(1+\varepsilon,\varepsilon/\kappa(P'AW^{-1}))$-solver for $P'AW^{-1}$.

\subsection{Fast Operations for the Preconditioner}
One of our main contributions is to develop several fast operations for $P'$ such that we can implement Algorithm~\ref{alg:feasibility} efficiently.

\subsubsection{Preconditioner Compression}
\paragraph{Removing useless cells.}
The first observation is that though $P'$ has a large number of rows, most rows of $P'$ are zero.
Thus, we can remove them.
Precisely, for each $i\in\{0,1,\cdots,L-1\}$, let $\mathcal{C}_i=\{C\in G_i\mid \exists v\in V,\tau\in[2^i],~s.t.~\varphi(v)+\tau\cdot \one_d\in C\}$.
Then we can set $P\in\mathbb{R}^{(\sum_{i=0}^{L-1}2^i|\mathcal{C}_i|)\times n}$ such that $\forall i\in\{0,1,\cdots,L-1\},C\in \mathcal{C}_i,\tau\in[2^i],v\in V$,
\begin{align*}
P_{(i,C,\tau),v}=\left\{\begin{array}{ll}d & \varphi(v)+\tau\cdot \one_d\in C,\\ 0 & \text{Otherwise.}\end{array}\right.
\end{align*}
\begin{lemma}\label{lem:non_empty_cells}
$\forall i\in\{0,1,\cdots,L-1\}$, $|\mathcal{C}_i|\leq n\cdot (d+1)$.
\end{lemma}
\begin{proof}
Since each cell has side length $2^i$, for a dimension $j\in[d]$ and a vertex $v\in V$, $\varphi(v)+\one_d,\varphi(v)+2\cdot \one_d,\cdots, \varphi(v)+2^i\cdot \one_d$ can cross the boundary in the $j$-th dimension at most once.
Therefore, $\varphi(v)+\one_d,\varphi(v)+2\cdot \one_d,\cdots, \varphi(v)+2^i\cdot \one_d$ can be in at most $d+1$ different cells in $G_i$.
Because $V$ has size $n$, we can conclude $|\mathcal{C}_i|\leq n\cdot (d+1)$.
\end{proof}
By Lemma~\ref{lem:non_empty_cells}, we know that $P$ has at most $2\Delta\cdot n(d+1)$ rows. This is still too large.

\paragraph{Compressed representation.} Another observation is that, $P$ may have many identical rows. 
Thus, we want to handle these rows simultaneously.
To achieve this goal, we introduce a concept called compressed representation.
\begin{definition}[Compressed representation of a vector]\label{def:implicit_vector}
Let $I=\{([a_1,b_1],c_1),([a_2,b_2],c_2),\cdots,([a_s,b_s],c_s)\}$, where $c_i\in\mathbb{R}$, $[a_i,b_i]\subseteq[1,r]$ for some $r\in \mathbb{Z}_{\geq 1}$, and $\forall i\not =j\in[s],[a_i,b_i]\cap[a_j,c_j]=\emptyset$. 
Let $x\in\mathbb{R}^r$.
If $\forall i\in[s],j\in[a_i,b_i],x_j=c_i$ and $\forall j\in [1,r]\setminus\bigcup_{i\in[s]}[a_i,b_i],x_j = 0$, then $I$ is an compressed representation of $x$.
The size of the compressed representation $I$ is $|I|=s$.
\end{definition}
By above definition, the compressed representation of $x$ is not unique.
\begin{definition}[Compressed representation of a matrix]\label{def:implicit_matrix}
Let $I=(I_1,I_2,\cdots,I_n)$.
Given a matrix $P\in\mathbb{R}^{r\times n}$, if $\forall i\in[n]$, $I_i$ is an compressed representation of $P_i$, then $I$ is called an compressed representation of $P$.
Furthermore, the size of the compressed representation $I$ is defined as $\sum_{i=1}^n |I_i|$.
\end{definition}

\begin{algorithm}[h]
	\caption{Computing an compressed representation of $P$}\label{alg:implicit_preconditioner}
	\begin{algorithmic}[1]
		\small
		\Procedure{\textsc{ImplicitP}}{$\varphi:V\rightarrow [\Delta]^d$} 
		\State Output: $I$
		\State $n\gets |V|,L\gets 1+\log \Delta,\forall i\in\{0,1,\cdots,L-1\},\mathcal{C}_i\gets \emptyset$, and create grids $G_0,G_1,\cdots,G_{L-1}$.
		\State $\forall i\in\{0,1,\cdots,L-1\}, v\in V,\mathcal{C}_i\gets \mathcal{C}_i\cup \{C\in G_i\mid \exists \tau\in[2^i],\varphi(v)+\tau\cdot \one_d\in C\}$. \label{sta:init_non_empty_cells}
		\For{the $i$-th vertex $v\in V$} \label{sta:implicit_preconditioner_for_loop}
			\State $I_i\gets \emptyset$.
			\For{$l\in \{0,1,\cdots,L-1\}$}\label{sta:implicit_preconditioner_inner_loop}
				\State For each $C\in\mathcal{C}_l$ with $\exists \tau\in[2^l],\varphi(v)+\tau\cdot\one_d\in C$, find $\tau_1,\tau_2\in[2^l]$ such that
				\begin{align*}
				\begin{array}{cc}
				\tau_1=\min_{\tau\in[2^l]:\varphi(v)+\tau\cdot \one_d\in C} \tau, & \tau_2=\max_{\tau\in[2^l]:\varphi(v)+\tau\cdot \one_d\in C} \tau.
				\end{array}
				\end{align*} \label{sta:compute_taus}
				\State Suppose $C$ is the $k$-th cell in $\mathcal{C}_l$. $a\gets (k-1)2^l+\sum_{j=0}^{l-1}2^j|\mathcal{C}_j|$. 
				\begin{align*}
				I_i\gets I_i\cup \left\{\left(\left[a+\tau_1,a+\tau_2\right],d\right)\right\}.
				\end{align*} \label{sta:implicit_final_step}
			\EndFor
		\EndFor
		\State Return $I=(I_1,I_2,\cdots,I_n)$.
		\EndProcedure
	\end{algorithmic}
\end{algorithm}

\begin{lemma}[Computing an compressed representation of $P$]\label{lem:implicit_P}
Given an undirected graph $G=(V,E,w)$ with $|V|=n,|E|=m$ and a mapping $\varphi:V\rightarrow [\Delta]^d$ for some $\Delta,d$, such that 
\begin{align*}
\forall u,v\in V, \dist_{G}(u,v)\leq \|\varphi(u)-\varphi(v)\|_1\leq \alpha\cdot \dist_G(u,v),
\end{align*}
the output $I=(I_1,I_2,\cdots,I_n)$ of \textsc{ImplicitP}$(\varphi)$ (Algorithm~\ref{alg:implicit_preconditioner}) is an compressed representation of a matrix $P$ with full column rank and $\kappa(PAW^{-1})\leq O(\alpha Ld)$, where $L=1+\log \Delta$, $A\in\mathbb{R}^{n\times m}$ is the vertex-incidence matrix, and $W\in\mathbb{R}^{m\times m}$ is the diagonal weight matrix.
Furthermore, for $i\in[n]$, the size of $I_i$ is at most $(d+1)L$.
The running time of \textsc{ImplicitP}$(\varphi)$ is $n\cdot \poly(dL\log n)$.
\end{lemma}
\begin{proof}
As discussed previously, $P$ has full column rank and $\kappa(PAW^{-1})\leq O(\alpha Ld)$ by our construction.
As discussed in the proof of Lemma~\ref{lem:non_empty_cells}, $\forall l\in\{0,1,\cdots,L-1\},v\in V$, we know that $\varphi(v)+\one_d,\varphi(v)+2\cdot \one_d,\cdots, \varphi(v)+2^i\cdot \one_d$ can be in at most $d+1$ different cells in $G_i$.
Thus, $\forall i\in[n],|I_i|\leq (d+1)L$.
Notice that $\forall v\in V,\varphi(v)$ has dimension $d$, we only need to handle $L$ levels, and sorting will only induce additional $\log n$ factors, the running time will be at most $n\cdot \poly(dL\log n)$.
\end{proof}

\subsubsection{Operations under Compressed Representations}
In this section, we introduce how to implement some important operations under compressed representations.

\begin{fact}\label{fac:naive_operation}
Let  $I=\{([a_1,b_1],c_1),([a_2,b_2],c_2),\cdots,([a_s,b_s],c_s)\}$ be an compressed representation of a vector $x\in\mathbb{R}^r$. Then,
$
\|x\|_1=\sum_{i=1}^s (b_i-a_i+1)\cdot |c_i|.
$
Let $y\in\mathbb{R}^r$ be the vector satisfying $\forall i\in[r],y_i=\sgn(x_i)$.
Then $I'=\{([a_1,b_1],\sgn(c_1)),\cdots,([a_s,b_s],\sgn(c_s))\}$ is an compressed representation of $y$.
Let $z=t\cdot x$, where $t$ is a scalar.
Then $I''=\{([a_1,b_1],tc_1),\cdots,([a_s,b_s],tc_s)\}$ is an compressed representation of a vector $z$.
Furthermore, both $\|x\|_1,$ $I'$ and $I''$ can be computed in $O(s)$ time.
\end{fact}

\begin{algorithm}[h]
	\caption{Compressed Matrix-Vector Multiplication}\label{alg:matrix_vector}
	\begin{algorithmic}[1]
		\small
		\Procedure{\textsc{MatrixVec}}{$I=(I_1,I_2,\cdots,I_n),g\in\mathbb{R}^n$} 
		\State Output: $\hat{I}$
		\State $S\gets\emptyset,\hat{I}\gets\emptyset$.
		\For{$i\in[n]:g_i\not = 0$}\label{sta:matrixvec_first_loop}
			\State For each $([a,b],c)\in I_i$, $S\gets S\cup\{(a,cg_i),(b+1,-cg_i)\}$.
		\EndFor
		\State Sort $S=\{(q_1,z_1),(q_2,z_2),\cdots,(q_k,z_k)\}$ such that $q_1\leq q_2\leq \cdots\leq q_k$. \label{sta:matrixvec_sorting}
		\State For each $j\in\{2,3,\cdots,k\}:q_j>q_{j-1}$, $\hat{I}\gets \hat{I}\cup\{([q_{j-1},q_j-1],\sum_{t:q_t<q_j} z_t)\}$. \label{sta:matrixvec_prefix_sum}
		\State Return $\hat{I}$.
		\EndProcedure
	\end{algorithmic}
\end{algorithm}

\begin{lemma}[Compressed matrix-vector multiplication]\label{lem:matrix_vector}
Given an compressed representation $I=(I_1,I_2,\cdots,I_n)$ of a matrix $P\in\mathbb{R}^{r\times n}$ with $\forall i\in[n],|I_i|\leq s$, and a vector $g\in\mathbb{R}^n$, the output $\hat{I}$ of \textsc{MatrixVec}$(I,g)$ (Algorithm~\ref{alg:matrix_vector}) is an compressed representation of $Pg$.
Furthermore, $|\hat{I}|\leq 2s\cdot\nnz(g)$, and the running time is at most $O(s\nnz(g)\cdot\log(s\nnz(g)) )$.
\end{lemma}
\begin{proof}
Consider $j\in[r]$ such that $(Pg)_j\not = 0$.
\begin{align*}
(Pg)_j = \sum_{i\in[n]:g_i\not = 0} P_{j,i} g_i =\sum_{i\in [n]:g_i\not =0, \exists ([a,b],c)\in I_i,j\in[a,b]} cg_i.
\end{align*}
Notice that $\forall h\in[k]$, $\sum_{t:q_t<q_h}z_t = \sum_{i\in[n]:g_i\not=0,\exists ([a,b],c)\in I_i,a\leq q_t,b\geq q_h} cg_i$.
Since we can always find $h\in\{2,3,\cdots,k\}$ such that $j\in [q_{h-1},q_h-1]$, then for such $h$ we have $(Pg)_j=\sum_{t:q_t<q_h} z_t$.
Thus, $\hat{I}$ is an compressed representation of $Pg$.

For $i\in[n]:g_i\not =0$, we will add at most $2$ elements in $S$ for each tuple in $I_i$.
Since each element in $S$ can correspond to at most $1$ tuple in $\hat{I}$, we have $|\hat{I}|\leq 2s\nnz(g)$.
Sorting takes $O(|S|\log |S|)$ time, and maintaining prefix sum takes $O(|S|)$ time. 
Thus, total running time is at most $O(|S|\log|S|)=O(s\nnz(g)\log(s\nnz(g)))$.
\end{proof}

\begin{algorithm}[h]
	\caption{Compressed Vector-Matrix Multiplication}\label{alg:vector_matrix}
	\begin{algorithmic}[1]
		\small
		\Procedure{\textsc{VectorMat}}{$I, I'=(I_1,I_2,\cdots,I_n)$} 
		\State Output: $g^{\top}\in\mathbb{R}^n$
		\State $g\gets (0,0,\cdots,0)$.
		\State Fill $I$ such that $\forall j\in [r]$, $\exists ([a,b],c)\in I,j\in[a,b]$. \label{sta:vectormat_filling}
		\State Sort $I=\{([a_1,b_1],c_1),([a_2,b_2],c_2),\cdots,([a_s,b_s],c_s)\}$ such that $a_1<a_2<\cdots<a_s$. \label{sta:vectormat_sort}
		\State $\forall j\in[s]$, compute the prefix sum $p_j=\sum_{t=1}^j (b_t-a_t+1)\cdot c_t$. \label{sta:vectormat_prefix_sum}
		\For{$i\in[n]$} \label{sta:vectormat_loop}
			 \For {$([a,b],c)\in I_i$}
			 	\State Run binary search to find $j_1\leq j_2$ such that $a\in [a_{j_1},b_{j_1}],b\in[a_{j_2},b_{j_2}]$.
				\State If $j_1=j_2$, $g_i\gets g_i+ c\cdot c_{j_1}\cdot (b-a+1)$. \label{sta:vectormat_sumgi_1}
				\State If $j_1<j_2$, $g_i\gets g_i + c\cdot (c_{j_1}\cdot (b_{j_1}-a+1)+c_{j_2}\cdot (b-a_{j_2}+1)+(p_{j_2-1}-p_{j_1}))$. \label{sta:vectormat_sumgi_2}
			\EndFor
		\EndFor
		\State Return $g^{\top}$.
		\EndProcedure
	\end{algorithmic}
\end{algorithm}

\begin{lemma}[Compressed vector-matrix multiplication]\label{lem:vector_matrix}
Given an compressed representation $I$ of a vector $y\in\mathbb{R}^r$ with $|I|\leq s$ and an compressed representation $I'=(I_1,I_2,\cdots,I_n)$ of a matrix $P\in\mathbb{R}^{r\times n}$ with $\forall i\in[n],|I_i|\leq s'$, the output $g^\top\in\mathbb{R}^n$ of \textsc{VectorMat}$(I,I')$ (Algorithm~\ref{alg:vector_matrix}) is $P^{\top}y$.
Furthermore, the running time is $O((s+ns')\log s)$.
\end{lemma}
\begin{proof}
Consider the $i$-th entry of $y^\top P$,
\begin{align*}
(y^\top P)_i = y^\top P_i = \sum_{([a,b],c)\in I_i} c\cdot \sum_{a\leq t\leq b} y_t.
\end{align*}
By our algorithm, it is easy to show that $(y^\top P)_i=g_i$.
Sorting $I$ takes $O(s\log s)$ time.
For each $i\in[n]$, we need to take $O(s'\log s)$ time for $s'$ times binary search.
Thus, the total running time is $O((s+s'n)\log s)$.
\end{proof}

\subsection{Uncapacitated Minimum Cost Flow Algorithm}
\begin{theorem}
Given an $\varepsilon\in(0,0.5)$, a connected $n$-vertex $m$-edge undirected graph $G=(V,E,w)$ with $w:E\rightarrow \mathbb{Z}_{\geq 0}$, and a demand vector $b\in\mathbb{R}^n$ with $\one_n^\top b = 0$, there is a randomized algorithm which can output an $(1+\varepsilon)$-approximate solution to the uncapacitated minimum cost flow problem in $\varepsilon^{-2}m\cdot (\log n\log \Lambda)^{O(1)}$ time with probability at least $0.99$, where $\Lambda=\sum_{e\in E} w(e)$.
\end{theorem}
\begin{proof}
	Let $\Lambda=\sum_{e\in E} w(e)$.
	Let $A\in\mathbb{R}^{n\times m}$ be the vertex-edge incidence matrix of $G$, and let $W\in\mathbb{R}^{m\times m}$ be the weight matrix.
	By Lemma~\ref{lem:bourgain}, with $0.99$ probability, we can compute a mapping $\varphi:V\rightarrow [\Delta]^d$ with $\Delta\leq \Lambda,d\leq O(\log^2 n)$ in $O(m\log^2 n)$ time such that 
	\begin{align*}
	\forall u,v\in V, \dist_G(u,v)\leq \|\varphi(u)-\varphi(v)\|_1\leq O(\log n)\cdot \dist_G(u,v).
	\end{align*} 
	By Lemma~\ref{lem:implicit_P}, we can compute an compressed representation $I=\{I_1,I_2,\cdots,I_n\}$ of $P$ with $\forall i\in[n],|I_i|\leq O(d\log \Delta)=O(\log^2 n\log \Lambda)$.
	Furthermore, $\kappa(PAW^{-1})\leq O(\log^3 n\log \Lambda)$.
	The running time is $n\cdot (\log n\log \Lambda)^{O(1)}$.
	Now, we are able to implement Algorithm~\ref{alg:feasibility}.
	To compute matrix $B$, we need to compute $\|PAW^{-1}\|_{1\rightarrow 1}$ and $\|Pb\|_1$.
	Notice that $\|PAW^{-1}\|_{1\rightarrow 1}=\max_{i\in[m]} \|P(AW^{-1})_i\|_1$ and $(AW^{-1})_i$ only has two non-zero entries.
	By Lemma~\ref{lem:matrix_vector}, an compressed representation of $P(AW^{-1})_i$ can be computed in $O(\log^2 n\log \Lambda\cdot (\log \log n+\log\log \Lambda))$ time, and the size of the compressed representation is at most $O(\log^2 n\log \Lambda)$.
	By Fact~\ref{fac:naive_operation}, $\|P(AW^{-1})_i\|_1$ can be computed in $O(\log^2 n\log \Lambda(\log\log n+\log\log \Lambda))$ time. 
	Thus, $\|PAW^{-1}\|_{1\rightarrow 1}$ can be computed in $m\cdot (\log n\log \Lambda)^{O(1)}$ time.
	By Lemma~\ref{lem:matrix_vector} again, an compressed representation of $Pb$ can be computed in $n\cdot (\log n\log \Lambda)^{O(1)}$ time.
	Follows from Fact~\ref{fac:naive_operation}, $\|Pb\|_1$ can be computed in $n\cdot (\log n\log \Lambda)^{O(1)}$ time.
	$Bp_t$ can be computed in $O(n+m)$ time.
	By Lemma~\ref{lem:matrix_vector}, an compressed representation $\hat{I}$ of $PBp_t$ can be computed in $n\cdot (\log n\log \Lambda)^{O(1)}$ time, and we have $|\hat{I}|\leq n\cdot (\log n\log \Lambda)^{O(1)}$.
	By Fact~\ref{fac:naive_operation}, $\|PBp_t\|_1$ can be computed in $n\cdot (\log n\log\Lambda)^{O(1)}$ time.
	By Fact~\ref{fac:naive_operation} again, an compressed representation $I'$ of $y_t$ can also be computed in $n\cdot (\log n\log\Lambda)^{O(1)}$ time, and we have $|I'|\leq n\cdot (\log n\log \Lambda)^{O(1)}$.
	By Lemma~\ref{lem:vector_matrix}, $y_t^\top P$ can be computed in $n\cdot (\log n\log \Lambda)^{O(1)}$ time.
	To compute $(y_t^\top P)B_i$, we need to compute $(y_t^\top P) A_i$ and $(y_t^\top P) b$.
	Thus, the running time to compute all $\phi_t^+(i),\phi_t^-(i)$ is $n\cdot (\log n\log \Lambda)^{O(1)}+O(m+n)$.
	Thus, one iteration of Algorithm~\ref{alg:feasibility} takes $m\cdot (\log n\log\Lambda)^{O(1)}$ time.
	By Lemma~\ref{lem:mwu_feas_prob}, Algorithm~\ref{alg:feasibility} takes $\frac{1}{\varepsilon^2}\cdot (\log n\log\Lambda)^{O(1)}$ iterations.
	To construct a $(1+\varepsilon,\varepsilon/\kappa)$-solver, we need to call Algorithm~\ref{alg:feasibility} $\log(\varepsilon^{-1}\log \kappa)$ times.
	Thus, to find an $(1+\varepsilon,\varepsilon/\kappa)$-solution, the running time is $\frac{m}{\varepsilon^2}\cdot (\log n\log \Lambda)^{O(1)}\cdot \log(1/\varepsilon)$.
	Since $\varepsilon\geq 1/\Lambda$, the running time is $\frac{m}{\varepsilon^2}\cdot (\log n\log \Lambda)^{O(1)}$.
	Together with Corollary~\ref{cor:boost_error}, Corollary~\ref{cor:exact_solver} and Lemma~\ref{lem:exact_solver_with_large_approx}, we complete the proof.
\end{proof}
\section{Implementation in Parallel Setting}
In this section, we will have a detailed discussion of how to implement our algorithms in PRAM model. 
For convenience, we will describe our algorithms in the PRIORITY CRCW PRAM~\cite{frw88}. 
In this model, if multiple processors write to the same memory cell, the cell will take the minimum written value. 
According to \cite{v83,frw88}, algorithms in the PRIORITY CRCW PRAM model can be easily simulated in other PRAM models (including the weakest EREW PRAM model) with at most polylogarithmic factors blow-up in the depth.

\subsection{Computation of Balls}
For an $n$-vertex $m$-edge graph $G=(V,E,w)$, one subroutine needed in our algorithms is to compute $r_b(v),\B^{\circ}_b(v)$ and $\dist(v,u)$ for $u\in \B^{\circ}_b(v)$ for every vertex $v\in V$. 
As shown by \cite{ss99}, this can be done in $\log^{O(1)} n$ depth and $\wt{O}(nb^2 + m)$ work.
We describe their algorithm in Section~\ref{sec:computation_of_ball}. 
\begin{lemma}[\cite{ss99}]\label{lem:parallel_ball_computation}
Given an $n$-vertex $m$-edge undirected weighted graph $G=(V,E,w)$ and a parameter $b\in [n]$, there is a PRAM algorithm which computes $r_b(v), \B^{\circ}_b(v)$ and $\dist(v,u)$ for $u\in\B^{\circ}_b(v)$ for every $v\in V$ in $\log^{O(1)} n$ depth and $\wt{O}(nb^2 + m)$ work.
\end{lemma}

\subsection{Parallel Subemulator Construction}
\begin{theorem}[Parallel construction of subemulator]\label{thm:parallel_subemulator}
Given a connected $n$-vertex $m$-edge undirected weighted graph $G=(V,E,w)$ and a parameter $b\in[n]$, there is a PRAM algorithm (Algorithm~\ref{alg:subemulator_construction}) which outputs an undirected graph $H=(V',E',w')$ and $q: V\rightarrow V'$ such that $H$ is a strong $(8,b,22)$-subemulator of $G$, and $q$ is a corresponding leader mapping (Definition~\ref{def:subemulator}). 
Furthermore, $\E[|V'|]\leq \min(75\log(n)/b,3/4)n,|E'|\leq nb + m$.
The depth of the algorithm is $\log^{O(1)}n$ and the work is $\wt{O}(nb^2 + m)$.
\end{theorem}
\begin{proof}
The correctness and the size of $H$ is already shown by Theorem~\ref{thm:construction_subemulator}. 
Next, let us analyze the depth and the work of Algorithm~\ref{alg:subemulator_construction}.
	
Consider the depth and the work of Algorithm~\ref{alg:leader_selection} and Algorithm~\ref{alg:edge_construction}.
In Algorithm~\ref{alg:leader_selection}, sampling procedure (line~\ref{sta:sampled_vertices}) can be done in $O(1)$ depth and $O(n)$ work. 
For line~\ref{sta:check_each_vertex} of Algorithm~\ref{alg:leader_selection}, the implementation is described as the following:
\begin{enumerate}
	\item For each $v\in V$, compute $r_b(v),~\B^{\circ}_{b}(v)$ and $\dist(v,u)$ for $u\in\B^{\circ}_b(v)$.
	\item For $u\in V$, initialize $l(u)\gets \infty$. 
	\item If $v$ is sampled to be in $S$ by line~\ref{sta:sampled_vertices}, let $l(v)\gets 0$ and for each edge $\{v,u\}\in E$, mark $l(u)\gets w(u,v)$. If $l(u)$ is marked multiple times, only keep the minimum one.
	\item For $v\in V$, if $\forall u\in \B^{\circ}_b(v),\dist(v,u)+l(u)>r_b(v)$, mark $v$ to be in $V'$.
\end{enumerate}
Due to Lemma~\ref{lem:parallel_ball_computation}, $r_b(v),\B^{\circ}_{b}(v)$ and $\dist(v,u)$ for $u\in\B^{\circ}_b(v)$ can be computed in $\log^{O(1)} n$ depth and $\wt{O}(nb^2 + m)$ work for all $v\in V$. 
The last three steps of the above procedure only takes $O(1)$ depth and $O(nb+m)$ work.
Thus, Algorithm~\ref{alg:leader_selection} only uses $\log^{O(1)} n$ depth and $\wt{O}(nb^2+m)$ work.

In Algorithm~\ref{alg:edge_construction}, the implementation of line~\ref{sta:leader_mapping} is similar as line~\ref{sta:check_each_vertex} of Algorithm~\ref{alg:leader_selection}:
\begin{enumerate}
	\item For each $v\in V$, compute $r_b(v),~\B^{\circ}_{b}(v)$ and $\dist(v,u)$ for $u\in\B^{\circ}_b(v)$.
	\item For $u\in V$, initialize $l(u)\gets \infty$. 
	\item If $v\in V'$, let $l(v)\gets 0,q(v)\gets v$ and for each edge $\{v,u\}\in E$, mark $l(u)\gets w(u,v),q(u)\gets v$. If $l(u)$ is marked multiple times, only keep the minimum one and keep $q(u)$ to be the corresponding $v$ which minimizes $l(u)$ (if there is a tie, let $q(u)$ have the smallest label).
	\item For $v\in V$, set $q(v)\gets q(u)$ where $u\in \B^{\circ}_b(v)$ and $l(u)+\dist(v,u)$ is minimized.
	Set $\dist(v,q(v))\gets l(u)+\dist(v,u)$.
\end{enumerate}
By applying Lemma~\ref{lem:parallel_ball_computation} again, the above steps only take $\log^{O(1)}n$ depth and $\wt{O}(nb^2+m)$ work.
Notice that $\forall v\in V,u\in\B^{\circ}_b(v), \dist(v,u)$ is computed, and $\forall v\in V,\dist(v,q(v))$ is also computed.
Line~\ref{sta:original_edges} of Algorithm~\ref{alg:edge_construction} has $O(1)$ depth and $O(m)$ work.
Line~\ref{sta:ball_edges} of Algorithm~\ref{alg:edge_construction} has $O(1)$ depth and $O(nb)$ work.
Line~\ref{sta:weight_assign_original_edge} of Algorithm~\ref{alg:edge_construction} has $O(1)$ depth and $O(m)$ work.
Line~\ref{sta:weight_assign_ball_edge} of Algorithm~\ref{alg:edge_construction} has $O(1)$ depth and $O(nb)$ work.
Overall, Algorithm~\ref{alg:edge_construction} takes $\log^{O(1)}$ depth and $\wt{O}(nb^2 + m)$ work.
\end{proof}

\subsection{Parallel Construction of Low Hop Emulator}
Our emulator construction depends on a subroutine \textsc{PreProc}($G,k$) (Algorithm~\ref{alg:dis_oracle}). 
In the following lemma, we analyze the depth and the work of \textsc{PreProc}($G,k$).
\begin{lemma}[Depth and work of \textsc{PreProc}($G,k$)]\label{lem:depth_work_preproc}
Given a connected $n$-vertex $m$-edge undirected weighted graph $G=(V,E,w)$ with $w:E\rightarrow \mathbb{Z}_{\geq 0}$ and a parameter $k\in [0.5,0.5\log n]$, \textsc{PreProc}($G,k$) (Algorithm~\ref{alg:dis_oracle}) has $\log^{O(1)}(n)$ depth and $\wt{O}(m + n^{1+1/k})$ expected work.
\end{lemma}
\begin{proof}
Let $t$ be the value at the end of \textsc{PreProc}($G,k$).
According to Theorem~\ref{thm:parallel_subemulator}, for $i\in \{0,1,\cdots,t-1\}$, \textsc{Subemulator}($H_i,b_i$) in line~\ref{sta:use_of_subemulator} of Algorithm~\ref{alg:dis_oracle} takes $\log^{O(1)} n_i$ depth and $\wt{O}(n_ib_i^2+m_i)$ work.
Since $n_i\leq n$, the depth is at most $\log^{O(1)}n$.
According to Lemma~\ref{lem:dis_oracle_space}, $\E[n_ib_i^2+m_i]\leq \max(n^{1+1/k},n\cdot(75\log n)^4)+m+2\cdot \max(n^{1+1/(2k)}, n\cdot (75\log n)^2)=\wt{O}(m+n^{1+1/k})$.
Notice that all the information needed in line~\ref{sta:set_Bt_mid} of Algorithm~\ref{alg:dis_oracle} can be obtained during \textsc{Subemulator}($H_i,b_i$) (see Lemma~\ref{lem:parallel_ball_computation} and the proof of Theorem~\ref{thm:parallel_subemulator}).
Since $t$ is at most $O(\log k)$, the overall depth of \textsc{PreProc}($G,k$) is at most $\log^{O(1)}n$ and the overall expected work is at most $\wt{O}(m+n^{1+1/k})$.
\end{proof}

\begin{theorem}[Parallel construction of low hop emulator]\label{thm:parallel_low_hop_emulator}
Given a connected $n$-vertex $m$-edge undirected weighted graph $G=(V,E,w)$ with $w:E\rightarrow \mathbb{Z}_{\geq 0}$ and a parameter $k\in[0.5,0.5\log n]$, there is a PRAM algorithm (Algorithm~\ref{alg:low_hop_graph_2}) which outputs an undirected weighted graph $G'=(V,E',w')$ with $\E[|E'|]\leq O(n^{1+1/(2k)}+n\log^2 n),w':E'\rightarrow \mathbb{Z}_{\geq 0}$ and hop diameter at most $16\lceil\log (k) + 1\rceil$ such that $\forall u,v\in V$,
\begin{align*}
\dist_G(u,v)\leq \dist_{G'}(u,v)\leq 27^{4\lceil\log(k) + 1\rceil}\cdot \dist_G(u,v).
\end{align*}
The depth of the algorithm is at most $\log^{O(1)}n$ and the expected work is at most $\wt{O}(m+n^{1+1/k})$.
\end{theorem}
\begin{proof}
The correctness and the size of $E'$ is already analyzed by Theorem~\ref{thm:low_hop_emulator}.
Let us consider the depth and the work of implementing Algorithm~\ref{alg:low_hop_graph_2}.

According to Lemma~\ref{lem:depth_work_preproc}, line~\ref{sta:preproc_in_low_hop_emu} of Algorithm~\ref{alg:low_hop_graph_2} has depth $\log^{O(1)} n$ and expected work $\wt{O}(m+n^{1+1/k})$.
The information needed in line~\ref{sta:setup_edge_weight_different_level_emulator} and line~\ref{sta:setup_original_edge_weight_same_level_emulator} of Algorithm~\ref{alg:low_hop_graph_2} can be obtained during \textsc{PreProc}($G,k$) in line~\ref{sta:preproc_in_low_hop_emu} of Algorithm~\ref{alg:low_hop_graph_2} (see Algorithm~\ref{alg:dis_oracle}, Lemma~\ref{lem:parallel_ball_computation} and the proof of Theorem~\ref{thm:parallel_subemulator}).
Thus, the remaining steps of Algorithm~\ref{alg:low_hop_graph_2} have $O(1)$ depth and work at most $O(|E'|)$. 
The overall depth of Algorithm~\ref{alg:low_hop_graph_2} is at most $\log^{O(1)} n$.
Since $\E[|E'|]\leq O(n^{1+1/(2k)}+n\log^2n)$, the expected work of Algorithm~\ref{alg:low_hop_graph_2} is at most $\wt{O}(m+n^{1+1/k})$.
\end{proof}

A byproduct of the parallel implementation of \textsc{PreProc}$(G,k)$ (Algorithm~\ref{alg:dis_oracle}) is a parallel distance oracle.
\begin{theorem}[Parallel distance oracle]\label{thm:parallel_dis_oracle}
	Given a connected $n$-vertex $m$-edge undirected weighted graph $G=(V,E,w)$ with $w:E\rightarrow \mathbb{Z}_{\geq 0}$ and a parameter $k\in[0.5,0.5\log n]$, there is a PRAM algorithm (Algorithm~\ref{alg:dis_oracle}) which outputs a data structure with expected size $\wt{O}\left(n^{1+1/(2k)}\right)$ in depth $\log^{O(1)}(n)$ and expected work $\wt{O}(m+n^{1+1/k})$  such that for any pair of vertices $u,v\in V$, a value $d$ satisfying $\dist_G(u,v)\leq d\leq 26^{4\lceil\log(k)+1\rceil}\dist_G(u,v)$ can be computed in $O(\log (4k))$ time given the outputted data structure.
\end{theorem}
\begin{proof}
The correctness and the query time is shown by Lemma~\ref{lem:correct_distance_oracle}.
By Lemma~\ref{lem:depth_work_preproc}, \textsc{PreProc}($G,k$) has depth $\log^{O(1)}(n)$ and $\wt{O}(m+n^{1/k})$ expected work.
Consider \textsc{Query}$(u,v)$ (Algorithm~\ref{alg:dis_oracle}). 
Let $t$ be the value at the end of \textsc{PreProc}$(G,k)$ (Algorithm~\ref{alg:dis_oracle}).
For $l\in\{0,1,\dots, t\}$ and $u\in V_l$, we only need the information $B_l(u)$, $q_l(u)$ and $\dist_{H_l}(u,v)$ for $v\in B_l$(u).
By Lemma~\ref{lem:dis_oracle_space}, we have $t\leq 4\lceil\log(k) + 1\rceil$ and $\E[\sum_{l=0}^t\sum_{v\in V_l} |B_l(v)|]\leq t\cdot \max(n^{1+1/k},n\cdot (75\log n)^4)/b_0\leq \wt{O}(n^{1+1/(2k)})$.
Thus the space to store all required information is at most $\wt{O}\left(n^{1+1/(2k)}\right)$.
\end{proof}

\subsection{Applications of Parallel Low Hop Emulator}\label{sec:parallel_app}
\paragraph{$\poly(\log n)$-Approximate single source shortest paths (SSSP).} 
A direct application is to compute a $\poly(\log n)$-approximate distance from a given vertex $s$ to every other vertex $v$. 
We just need to compute a low hop emulator, and run $O(\log\log n)$ Bellman-Ford iterations starting from the source vertex $s$.
\begin{corollary}[Parallel $\poly(\log n)$-approximate SSSP]\label{cor:parallel_sssp}
Given a connected $n$-vertex $m$-edge undirected weighted graph $G=(V,E,w)$ with $w:E\rightarrow Z_{\geq 0}$ and a source vertex $s\in V$, there is a PRAM algorithm which can output an $\poly(\log n)$-approximation to $\dist_G(s,v)$ for every $v\in V$.
Furthermore, the depth of the algorithm is $\log^{O(1)}(n)$ and the expected work is $\wt{O}(m)$.
\end{corollary}

\paragraph{Embedding the graph metric into $\ell_1$.} The second application of our low hop emulator is an efficient parallel algorithm which can embed the graph metric into $\ell_1$ space. 
In particular, given a connected $n$-vertex $m$-edge undirected weighted graph $G=(V,E,w)$, we can in parallel find a mapping $\varphi:V\rightarrow \ell_1^d$ for $d=O(\log^2 n)$ such that
\begin{align*}
\forall u, v\in V, \dist_G(u,v)\leq \|\varphi(u)-\varphi(v)\|_1\leq \log^{O(1)}(n)\cdot \dist_G(u,v).
\end{align*}
This can be done by applying Bourgain's embedding~\cite{b85} on our low hop emulator.
The algorithm is described in Section~\ref{sec:parallel_bourgain}.
\begin{corollary}[Parallel embedding into $\ell_1$]\label{cor:parallel_embed_into_l1}
Given a connected $n$-vertex $m$-edge undirected weighted graph $G=(V,E,w)$ with $w:E\rightarrow Z_{\geq 0}$, there is a PRAM algorithm which can output a mapping $\varphi:V\rightarrow [\Delta]^d$ for $\Delta=\diam(G)\cdot \log^{O(1)}(n),d=O(\log^2 n)$ such that with probability at least $0.99$,
\begin{align*}
\forall u,v \in V, \dist_G(u,v) \leq \|\varphi(u)-\varphi(v)\|_1\leq \log^{O(1)}(n)\cdot \dist_G(u,v).
\end{align*}
Furthermore, the depth of the algorithm is $\log^{O(1)}(n)$ and the expected work is $\wt{O}(m)$.
\end{corollary}
\paragraph{Low diameter decomposition.}
Another application is low diameter decomposition.
This can be done by applying algorithm of \cite{mpx13} on our low hop emulator (see discussion in Section~\ref{sec:parallel_low_diam_decomp}).
\begin{corollary}[Low diameter decomposition]\label{cor:parallel_ldd}
Given a connected $n$-vertex $m$-edge undirected weighted graph $G=(V,E,w)$ with $w:E\rightarrow Z_{\geq 0}$ and a parameter $\beta\in (0,1]$, there is a PRAM algorithm which can partition $V$ into subsets $C_1,C_2,\cdots,C_k$ such that 
\begin{enumerate}
\item $\forall i\in[k],~\forall u,v\in C_i,~\dist_G(u,v)\leq \frac{\log^{O(1)} n}{\beta}$,
\item $\forall u,v\in V,~\Pr[u,v\text{ are not in the same subset}]\leq \beta\cdot \dist_G(u,v)\cdot \log^{O(1)} n$.
\end{enumerate}
Furthermore, the depth of the algorithm is $\log^{O(1)}(n)$ and the expected work is $\wt{O}(m)$.
\end{corollary}

\paragraph{Metric tree embedding.} By applying the parallel FRT embedding (Theorem 7.9 of \cite{fl18}) on our low hop emulator directly, we can obtain a more work-efficient parallel metric tree embedding algorithm.
\begin{corollary}[Metric tree embedding]\label{cor:parallel_metric_tree}
Given a connected $n$-vertex $m$-edge undirected weighted graph $G=(V,E,w)$ with $w:E\rightarrow \mathbb{Z}_{\geq 0}$, there is a PRAM algorithm which can output a tree $T=(V',E',w')$ where $V'\supseteq V$ such that $\forall u,v\in V$,
\begin{enumerate}
\item $\dist_G(u,v)\leq \dist_T(u,v)$,
\item $\E[\dist_T(u,v)]\leq \log^{O(1)}(n)\cdot \dist_G(u,v)$.
\end{enumerate}
The depth of the algorithm is $\log^{O(1)}(n)$ and the expected work is $\wt{O}(m\cdot \log(\diam(G)))$.
\end{corollary}

\subsection{Parallel Uncapacitated Minimum Cost Flow}
\begin{fact}[Parallel $(n,0)$-solver]\label{fac:parallel_exact_solver}
Given a connected $n$-vertex $m$-edge undirected weighted graph $G$ and a demand vector $b\in \mathbb{R}^n$, \textsc{MSTRouting}$(G,b)$ (Algorithm~\ref{alg:mst_routing}) can be implemented in PRAM with depth $\log^{O(1)}(n)$ and $\wt{O}(m)$ work.
\end{fact}
\begin{proof}
We can use Boruvka's algorithm to compute a minimum spanning tree $T$ of $G$, and it can be implemented in PRAM with $\log^{O(1)}(n)$ depth and $\wt{O}(m)$ work (see e.g., \cite{cc96}).
The Euler Tour of $T$ can be computed in $\log^{O(1)}(n)$ depth and $\wt{O}(n)$ work~\cite{tv84}.
A subtree of $T$ should appear in a consecutive subsequence of the Euler Tour.
Thus, the sum of weights in a subtree can be computed as a sum of weights of a subsequence of the Euler Tour.
We can use $\log^{O(1)}(n)$ depth and $\wt{O}(n)$ work to preprocess a prefix sum over the Euler Tour and hence the line~\ref{sta:subtree_sum} of Algorithm~\ref{alg:mst_routing} can be computed in $O(1)$ depth and $\wt{O}(n)$ work.
\end{proof}

As discussed in Section~\ref{sec:min_cost_flow}, we need to find a good preconditioner.

\begin{lemma}[Work and depth of parallel preconditioner construction]\label{lem:work_depth_implicit_P}
Given a mapping $\varphi: V\rightarrow [\Delta]^d$ for some $\Delta,d\in \mathbb{Z}_{\geq 0}$, \textsc{ImplicitP}$(\varphi)$ (Algorithm~\ref{alg:implicit_preconditioner}) can be implemented in PRAM with $(d\log(n\Delta))^{O(1)}$ depth and $n\cdot(d\log(n\Delta))^{O(1)}$ work.
\end{lemma}
\begin{proof}
Let $L=1+\log\Delta $.
Consider line~\ref{sta:init_non_empty_cells} of Algorithm~\ref{alg:implicit_preconditioner}.
We can simultaneously handle each pair $(l,v)\in \{0,1,\cdots,L-1\}\times V$. 
As discussed in the proof of Lemma~\ref{lem:non_empty_cells}, $\forall l\in\{0,1,\cdots,L-1\},v\in V$, we know that $\varphi(v)+\one_d,\varphi(v)+2\cdot \one_d,\cdots, \varphi(v)+2^i\cdot \one_d$ can be in at most $d+1$ different cells in $G_i$.
Notice that each cell can be denoted by one of its corner point and the side length.
Thus, the depth of line~\ref{sta:init_non_empty_cells} of Algorithm~\ref{alg:implicit_preconditioner} is $O(d^2)$ and the work is $O(nLd^2)$.
For the outer loop and the inner loop started from line~\ref{sta:implicit_preconditioner_for_loop} and line~\ref{sta:implicit_preconditioner_inner_loop} of Algorithm~\ref{alg:implicit_preconditioner}, we can handle each $(v,l)$ simultaneously.
Again, there are at most $d+1$ cells in line~\ref{sta:compute_taus} will be considered, and each cell can be indicated by its side length and one corner point which has $d$ coordinates.
In addition, for one cell $C$, $\tau_1$ and $\tau_2$ can be computed in $O(d)$ time.
Thus, line~\ref{sta:compute_taus} has depth $O(d^2)$ and work $O(nLd^2)$.
To implement line~\ref{sta:implicit_final_step} of Algorithm~\ref{alg:implicit_preconditioner}, for each $l\in\{0,1,\cdots,L-1\}$, we need to index all cells in $\mathcal{C}_l$ before the loop started from line~\ref{sta:implicit_preconditioner_for_loop}.
This can be done by sorting.
Since $|\mathcal{C}_l|\leq (d+1)\cdot n$ and each cell is represented by size at most $O(d)$.
The sorting has depth at most $\log^{O(1)}(nd)$ and total work $nL\cdot (d\log n)^{O(1)}$.
Notice that we can compute $\sum_{j=0}^{l-1} 2^j |\mathcal{C}_j|$ for every $l\in\{0,1,\cdots,L-1\}$ using depth $O(L)$ and work $O(L)$.
After preprocessing steps, the depth of line~\ref{sta:implicit_final_step} is $O(d)$ and the work is $O(nLd)$.
To conclude, the depth of Algorithm~\ref{alg:implicit_preconditioner} is $(d\log(n\Delta))^{O(1)}$ and the total work is $n\cdot(d\log(n\Delta))^{O(1)}$.
\end{proof}

\begin{lemma}[Work and depth of parallel compressed matrix-vector multiplication]\label{lem:work_depth_matrixvec}
Given an compressed representation (see Definition~\ref{def:implicit_matrix}) $I=(I_1,I_2,\cdots,I_n)$ of a matrix $P\in\mathbb{R}^{r\times n}$ with $\forall i\in[n],|I_i|\leq s$, and a vector $g\in\mathbb{R}^n$, \textsc{MatrixVec}$(I,g)$ (Algorithm~\ref{alg:matrix_vector}) can be implemented in PRAM with depth $\log^{O(1)}(s\cdot \nnz(g))$ and work $\wt{O}(s\cdot\nnz(g))$.
\end{lemma}
\begin{proof}
For the loop started from line~\ref{sta:matrixvec_first_loop} of Algorithm~\ref{alg:matrix_vector}, the set $S$ can be created in depth $O(1)$, and the work is $O(|S|) $.
The sorting in line~\ref{sta:matrixvec_sorting} has depth $\log^{O(1)} |S|$ and work $|S|\cdot \log^{O(1)}|S|$.
For line~\ref{sta:matrixvec_prefix_sum}, we need to preprocess a prefix sum $\sum_{t:q_t<q_j} z_t$ for $j\in \{2,3,\cdots, k\}$.
This can be done in depth $\log^{O(1)} |S|$ and work $|S|\log^{O(1)} |S|$.
Once we preprocess the prefix sum, we can implement line~\ref{sta:matrixvec_prefix_sum} in $O(1)$ depth and $O(|S|)$ work.
Since $|S| = O(s\cdot \nnz(g))$, the overall depth of Algorithm~\ref{alg:matrix_vector} is $\log^{O(1)}(s\cdot \nnz(g))$ and the work is $\wt{O}(s\cdot \nnz(g))$.
\end{proof}

\begin{lemma}[Work and depth of parallel compressed vector-matrix multiplication]\label{lem:work_depth_vectormat}
Given an compressed representation $I$ of a vector $y\in\mathbb{R}^r$ with $|I|\leq s$ and an compressed representation $I'=(I_1,I_2,\cdots,I_n)$ of a matrix $P\in\mathbb{R}^{r\times n}$ with $\forall i\in[n], |I_i|\leq s'$, \textsc{VectorMat}$(I,I')$ (Algorithm~\ref{alg:vector_matrix}) can be implemented in PRAM with depth $\log^{O(1)}(ss')$ and work $\wt{O}(s+ns')$. 
\end{lemma}
\begin{proof}
By sorting, we can implement line~\ref{sta:vectormat_filling} and line~\ref{sta:vectormat_sort} of Algorithm~\ref{alg:vector_matrix} in $\log^{O(1)}(s)$ depth and $s\log^{O(1)} s$ depth.
Line~\ref{sta:vectormat_prefix_sum} computes the prefix sum $p_j$ for $j\in[s]$.
Thus, it needs $\log^{O(1)}(s)$ depth and $s\cdot \log^{O(1)}(s)$ work.
For the loop started from line~\ref{sta:vectormat_loop}, we can handle each $i\in[n]$ and $([a,b],c)\in I_i$ simultaneously.
The binary search takes $\log^{O(1)} s$ depth.
Line~\ref{sta:vectormat_sumgi_1} and line~\ref{sta:vectormat_sumgi_2} compute $g_i$ which is the sum of at most $|I_i|\leq s'$ values.
Thus, Line~\ref{sta:vectormat_sumgi_1} and line~\ref{sta:vectormat_sumgi_2} have depth $\log^{O(1)}(s')$.
To conclude, the overall depth of Algorithm~\ref{alg:vector_matrix} is $\log^{O(1)}(ss')$ and the work is $\wt{O}(s+ns')$.
\end{proof}

\begin{theorem}[Parallel uncapacitated minimum cost flow]\label{thm:parallel_min_cost_flow}
Given an $\varepsilon\in (0,0.5)$, a connected $n$-vertex $m$-edge undirected weighted graph $G=(V,E,w)$ with $w:E\rightarrow \mathbb{Z}_{\geq 0}$, and a demand vector $b\in\mathbb{R}^n$ with $\one_n^{\top} b = 0,$ there is a PRAM algorithm which can output an $(1+\varepsilon)$-approximate solution to the uncapacitated minimum cost flow problem with probability at least $0.99$.
Furthermore, the depth is at most $\varepsilon^{-2}\log^{O(1)}(n\Lambda)$ and the expected work is at most $\varepsilon^{-2}m\cdot \log^{O(1)}(n\Lambda)$, where $\Lambda=\sum_{e\in E} w(e)$.
\end{theorem}
\begin{proof}
Let $A\in\mathbb{R}^{n\times m}$ be the vertex-edge incidence matrix of $G$, and let $W\in\mathbb{R}^{m\times m}$ be the weight matrix.
By Corollary~\ref{cor:parallel_embed_into_l1}, with probability at least $0.99$, we can compute a mapping $\varphi:V\rightarrow [\Delta]^d$ with $\Delta\leq \Lambda\cdot \log^{O(1)} n,d\leq O(\log^2 n)$ in $\log^{O(1)}(n)$ depth and $\wt{O}(m)$ expected work such that
\begin{align*}
\forall u,v\in V, \dist_G(u,v)\leq \|\varphi(u)-\varphi(v)\|_1\leq \log^{O(1)}(n)\cdot \dist_G(u,v).
\end{align*}
By Lemma~\ref{lem:implicit_P}, we can compute an compressed representation $I=\{I_1,I_2,\cdots,I_n\}$ of $P$ with $\forall i\in[n],|I_i|\leq O(d\log \Delta)=\log^{O(1)}(n\Lambda)$.
Furthermore, $\kappa(PAW^{-1})\leq \log^{O(1)}(n\Lambda)$.
By Lemma~\ref{lem:work_depth_implicit_P}, the depth of computing such compressed representation is $\log^{O(1)}(n\Lambda)$  and the work is $n\cdot\log^{O(1)}(n\Lambda)$.
Now, we are able to implement Algorithm~\ref{alg:feasibility} in parallel.
To compute matrix $B$, we need to compute $\|PAW^{-1}\|_{1\rightarrow 1}$ and $\|Pb\|_1$.
Notice that $\|PAW^{-1}\|_{1\rightarrow 1}=\max_{i\in[m]}\|P(AW^{-1})_i\|_1$ and $(AW^{-1})_i$ only has two non-zero entries.
By Lemma~\ref{lem:matrix_vector} and Lemma~\ref{lem:work_depth_matrixvec}, an compressed representation of $P(AW^{-1})_i$ can be computed in $\log^{O(1)}(n\Lambda)$ depth.
We can compute $P(AW^{-1})_i$ for all $i\in[m]$ simultaneously.
By Lemma~\ref{lem:work_depth_matrixvec}, the total work needed is $m\cdot \log^{O(1)}(n\Lambda)$.
By Lemma~\ref{lem:matrix_vector}, the size of the compressed representation of $P(AW^{-1})_i$ is at most $\log^{O(1)}(n\Lambda)$.
Thus, by Fact~\ref{fac:naive_operation}, $\|P(AW^{-1})_i\|_1$ can be computed in $\log^{O(1)}(n\Lambda)$ depth.
Thus, $\|PAW^{-1}\|_{1\rightarrow 1}$ can be computed in $\log^{O(1)}(n\Lambda)$ depth and $m\cdot \log^{O(1)}(n\Lambda)$ work.
By Lemma~\ref{lem:matrix_vector} and Lemma~\ref{lem:work_depth_matrixvec} again, an compressed representation of $Pb$ can be computed in $\log^{O(1)}(n\Lambda)$ depth and $n\cdot \log^{O(1)}(n\Lambda)$ work.
By Lemma~\ref{lem:matrix_vector}, the size of the compressed representation of $Pb$ is at most $n\cdot \log^{O(1)}(n\Lambda)$, thus $\|Pb\|_1$ can be computed by the summation of at most $n\cdot \log^{O(1)}(n\Lambda)$ values.
Such summation operation has at most $\log^{O(1)}(n\Lambda)$ depth and $n\cdot \log^{O(1)}(n\Lambda)$ work.
Algorithm~\ref{alg:feasibility} has $O(\varepsilon^{-2}\kappa^2\log n)=\varepsilon^{-2}\log^{O(1)}(n\Lambda)$ iterations.
In each iteration, we firstly need to compute $p_t$.
This can be done in $\log^{O(1)}(n)$ depth and $m\log^{O(1)}(n)$ work.
Notice that $AW^{-1}$ only has $O(m)$ non-zero entries and $b\cdot \one_m^{\top} p_t = b$. 
We can compute $Bp_t$ in depth $\log^{O(1)}(n)$ and work $m\cdot \log^{O(1)} n$.
By Lemma~\ref{lem:matrix_vector} and Lemma~\ref{lem:work_depth_matrixvec}, an compressed representation $\hat{I}$ of $PBp_t$ can be computed in depth $\log^{O(1)}(n\Lambda)$ and work $n\cdot \log^{O(1)}(n\Lambda)$.
By Lemma~\ref{lem:matrix_vector}, the size of $\hat{I}$ is at most $n\log^{O(1)}(n\Lambda)$. 
Similar as computing $\|Pb\|_1$, $\|PBp_t\|_1$ can be computed in depth $\log^{O(1)}(n\Lambda)$ and work $n\log^{O(1)}(n\Lambda)$.
Now, we want to compute an compressed representation $I'$ of $y_t$.
To achieve this, we can look at each $([a,b],c)\in \hat{I}$ simultaneously and put $([a,b],\sgn(c))$ into $I'$.
This step has depth $O(1)$ and work at most $n\log^{O(1)}(n\Lambda)$.
It is easy to see $|I'|=|\hat{I}|$.
By Lemma~\ref{lem:vector_matrix} and Lemma~\ref{lem:work_depth_vectormat}, $y_t^{\top}P$ can be computed in depth $\log^{O(1)}(n\Lambda)$ and work $n\cdot \log^{O(1)}(n\Lambda)$.
To compute $(y_t^{\top}P)B_i$, we need to compute $(y_t^{\top}P)A_i$ and $(y_t^{\top}P)b$.
Since $A_i$ has at most $2$ non-zero entries, we can simultaneously compute $(y_t^{\top}P)A_i$ for all $i\in[m]$ and $(y_t^{\top}P)b$.
It has depth $\log^{O(1)}(n\Lambda)$ and $m\log^{O(1)}(n\Lambda)$ work.
To construct a $(1+\varepsilon,\varepsilon/\kappa)$-solver, we need to invoke Algorithm~\ref{alg:feasibility} $\log(\varepsilon^{-1}\log\kappa)$ times. 
Thus, to find an $(1+\varepsilon,\varepsilon/\kappa)$-solution, the total depth is $\varepsilon^{-2}\log^{O(1)}(\varepsilon^{-1}n\Lambda)$ and the work is $\varepsilon^{-2}m\log^{O(1)}(\varepsilon^{-1}n\Lambda)$.
Since $\varepsilon\geq 1/\Lambda$, $\log(\varepsilon^{-1})\leq \log\Lambda$.
Together with Corollary~\ref{cor:boost_error}, Corollary~\ref{cor:exact_solver}, Lemma~\ref{lem:exact_solver_with_large_approx} and Fact~\ref{fac:parallel_exact_solver}, we complete the proof.
\end{proof}

\subsection{Parallel $s-t$ Approximate Shortest Path}
Given two vertices $s$ and $t$, a special case of uncapacitated minimum cost flow is when the demand vector $b$ only has $2$ non-zero entries: $b_s = 1$ and $b_t= -1$.
In this case, the value of the minimum cost flow is exactly the same as the distance between $s$ and $t$.
As shown previously, since we can compute a $(1+\varepsilon)$-approximation to the uncapacitated minimum cost flow, we can compute a $(1+\varepsilon)$-approximation to the distance between $s$ and $t$.
However, our flow algorithm can only output a flow from $s$ to $t$.
In this section, we will show how to obtain an $s-t$ path from the $s-t$ flow.

Before we present our algorithm, let us show a good property of the random walk corresponding to the flow.
Consider a connected $n$-vertex $m$-edge undirected weighted graph $G=(V,E,w)$.
Let $A\in\mathbb{R}^{n\times m}$ be the corresponding vertex-edge incidence matrix (see Section~\ref{sec:min_cost_flow}), and let $W\in\mathbb{R}^{m\times m}$ be the corresponding diagonal weight matrix.
Given a valid demand vector $b\in\mathbb{R}^n$, i.e., $\one^{\top}_n b = 0$, let $f\in\mathbb{R}^m$ be a feasible flow, i.e., $Af = b$.
Suppose $\{u,v\}\in E$ is the $i$-th edge.
We denote $f(u,v)$ as the flow from $u$ to $v$, i.e.,
\begin{align*}
f(u,v) = \left\{\begin{array}{ll} f_i& u < v,\\ -f_i & u > v.\end{array}\right.
\end{align*}
By the definition of $f(u,v)$, we have $f(u,v)= -f(v,u)$.
For $\{u,v\}\not\in E$, we denote $f(u,v)=0$.
If $f(u,v)$ is negative, then it means that there is $-f(u,v)$ units of flow from $v$ to $u$.  
Suppose the demand vector $b$ further satisfies $b_t = -1,$ and $\forall v\not =t, b_v\geq 0$.
Notice that since $\one_n^{\top} b = 0$, we have $\sum_{v\in V\setminus\{t\}} b_v = 1$.
We can generate a random walk by the following way:
\begin{enumerate}
	\item Set $i\gets 0$ and set $u_0$ to be $v\in V\setminus \{t\}$ with probability $b_v$.
	\item Set $u_{i+1}$ to be $v\in \{v'\in V\mid f(u_i,v')>0\}$ with probability $\frac{f(u_i,v)}{\sum_{v':f(u_i,v')>0}f(u_i,v')}$. \label{it:rand_walk_next}
	\item If $u_{i+1}\not=t$, set $i\gets i + 1$, and repeat step~\ref{it:rand_walk_next}. Otherwise, output the path $p=(u_0,u_1,\cdots,u_{i+1})$.
\end{enumerate}
We say the path $p$ is a random walk corresponding to the flow $f$.
If the flow $f$ contains no cycle, then the expected length of $p$ is exactly the same as the cost of $f$ (see e.g., \cite{bkkl17}).
The following lemma shows that if $f$ contains cycles then the expected length of $p$ is still the cost of $f$.

\begin{lemma}[Expected length of a random walk]\label{lem:expected_length_of_random_walk}
Consider a connected $n$-vertex $m$-edge undirected weighted graph $G=(V,E,w)$ with $w:E\rightarrow \mathbb{Z}_{\geq 0}$. 
Let $A\in\mathbb{R}^{n\times m}$ be the corresponding vertex-edge incidence matrix, and let $W\in\mathbb{R}^{m\times m}$ be the corresponding diagonal weight matrix.
Given a vertex $t\in V$ and a demand vector $b\in\mathbb{R}^n$ satisfying $b_t=-1$ and $\forall v\not=t,b_v\geq 0$, let $f\in\mathbb{R}^m$ be a feasible flow for $b$, i.e., $Af=b$.
The expected length of a random walk corresponding to the flow $f$ is $\|Wf\|_1$.
\end{lemma}
\begin{proof}
For $u\in V$, let $d(u)$ denote the expected length of a random walk starting from the vertex $u$.
Notice that $d(t) = 0$.
For each vertex $u\in V$, we have the following equation:
\begin{align*}
d(u)=\sum_{v:f(u,v)>0}\frac{f(u,v)}{\sum_{v':f(u,v')>0} f(u,v')} \cdot (d(v) + w(u,v)).
\end{align*}
By reordering the terms, we have:
\begin{align*}
\left(\sum_{v':f(u,v')>0}f(u,v')\cdot d(u)\right) - \left(\sum_{v:f(u,v)>0} f(u,v) \cdot d(v)\right) = \sum_{v:f(u,v)>0} f(u,v)\cdot w(u,v).
\end{align*}
By summation over all vertices $u\in V$, we have:
\begin{align*}
&\sum_{u\in V}\left(\left(\sum_{v':f(u,v')>0}f(u,v')\cdot d(u)\right) - \left(\sum_{v:f(u,v)>0} f(u,v) \cdot d(v)\right)\right) = \sum_{u\in V}\sum_{v:f(u,v)>0} f(u,v)\cdot w(u,v)\\
\Rightarrow& \sum_{u\in V} d(u)\cdot \left(\left(\sum_{v':f(u,v')>0} f(u,v')\right)-\left(\sum_{v:f(v,u)>0} f(v,u)\right)\right) = \|Wf\|_1\\
\Rightarrow & \sum_{u\in V} d(u)\cdot b_u = \|Wf\|_1.
\end{align*}
Since $d(t)=0$, the expected length of a random walk corresponding to $f$ is $\sum_{u\in V\setminus\{t\}} d(u)\cdot b_u= \sum_{u\in V} d(u)\cdot b_u = \|Wf\|_1$.
\end{proof}
According to the above lemma, if $f$ is a $(1+\varepsilon)$-approximation to the optimal $s-t$ flow, then an $s-t$ random walk corresponding to $f$ is a $(1+\varepsilon)$-approximate shortest path.
However, simulating a random walk in parallel may not be easy. 
Instead, we will show how to in parallel find a path of which expected length is at most the expected length of a random walk (Algorithm~\ref{alg:stpath}).

\begin{algorithm}[h]
	\caption{Finding an $s-t$ Path} \label{alg:stpath}
	\begin{algorithmic}[1]
		\small
		\Procedure{\textsc{FindPath}}{$G=(V,E,w),s,t\in V,\varepsilon\in(0,0.5)$} 
		\State Output: $p=(u_0,u_1,u_2,\cdots,u_h)$
		\State If $s=t$, return $p=(s)$.
		\State $n\gets |V|, m\gets |E|$. Initialize a demand vector $b\in\mathbb{R}^n$: $b_s\gets 1,b_t\gets-1,\forall v\not=s,t,b_v\gets 0$.
		\State Compute a $(1+\varepsilon)$-approximate uncapacitated minimum cost flow $f$ satisfying $b$.\label{sta:compute_flow} \Comment{Theorem~\ref{thm:parallel_min_cost_flow}.}
		\State For each vertex $u\in V\setminus \{t\}$, set the pointer $\p(u)\gets v\in\{v'\in V\mid f(u,v')>0\}$ with probability 
		\begin{align*}
			\frac{f(u,v)}{\sum_{v':f(u,v')>0}f(u,v')}.
		\end{align*}\label{sta:select_out_edge}
		\State Let $G'=(V,E')$, where $E'=\{\{u,\p(u)\}\mid u\in V\setminus\{t\}\}$. \label{sta:construct_Gprime}
		\State Compute a spanning forest of $G'$. For $u\in V$, if $u$ is in the same connected component as $t$, set $\rt(u)\gets t$; otherwise set $\rt(u)\gets v$ where $v$ is in the same connected component as $u$ and the edge $\{v,\p(v)\}$ does not appear in the spanning forest. \label{sta:compute_spanning_forest}
		\State For $u\in V$, compute $l(u)\gets w(u,\p(u))+w(\p(u),\p(\p(u)))+\cdots+w(\p(\cdots\p(u)),\rt(u))$.\label{sta:dis_to_root}
		\State Set $V''\gets \{v\in V\mid \rt(v) = v\},E''\gets \{\{\rt(u),\rt(v)\}\mid\{u,v\}\in E\}$.
		\State For each $e''=\{u'',v''\}\in E''$, set 
		\begin{align*}
		\map(e'')\gets \arg\min_{\underset{\rt(u)=u'',\rt(v)=v''}{\{u,v\}\in E:}} l(u)+w(u,v)+l(v).
		\end{align*}
		\State For each $e''\in E''$, set $w''(e'')\gets l(u)+w(u,v)+l(v)$, where $\{u,v\}=\map(e'')$. \label{sta:setup_weight_wdoubleprime}
		\State $p''=(u''_0,u''_1,\cdots,u''_{h''})\gets$\textsc{FindPath}$(G''=(V'',E'',w''),\rt(s),\rt(t),\varepsilon)$.\label{sta:recursion} \Comment{Recursion.}
		\State Create $p$ by replacing $u''_0$ with 
		\begin{align*}
		(s,\p(s),\p(\p(s)),\cdots,\rt(s)),
		\end{align*}
		and replacing each edge $\{u''_{i-1},u''_i\}$ of $p''$ with a path
		\begin{align*}
		(\rt(x_i),\p(\cdots\p(x_i)),\cdots,\p(x_i),x_i,y_{i},\p(y_{i}),\p(\p(y_{i})),\cdots,\rt(y_{i})),
		\end{align*}
		where $\{x_i,y_i\}=\map(\{u''_{i-1},u''_i\})$. \label{sta:create_p}
		\State Shortcut all cycles of $p$ and return $p$. \label{sta:shortcut_cycle}
		\EndProcedure
	\end{algorithmic}
\end{algorithm}

\begin{lemma}[Work and depth of parallel $s-t$ approximate shortest path]\label{lem:work_depth_stpath}
Given an $\varepsilon\in(0,0.5)$, a connected $n$-vertex $m$-edge undirected graph $G=(V,E,w)$ with $w:E\rightarrow \mathbb{Z}_{\geq 0}$, and two vertices $s,t\in V$, \textsc{FindPath}$(G,s,t,\varepsilon)$ (Algorithm~\ref{alg:stpath}) can be implemented in PRAM with $\varepsilon^{-2}\poly(\log(n\Lambda))$ depth and expected $\varepsilon^{-2}m\poly(\log(n\Lambda))$ work, where $\Lambda=\max_{e\in E} w(e)$.
\end{lemma}
\begin{proof}
Line~\ref{sta:compute_flow} can be done in $\varepsilon^{-2}\poly(\log(n\Lambda))$ depth using expected $\varepsilon^{-2}m\poly(\log(n\Lambda))$ work by Theorem~\ref{thm:parallel_min_cost_flow}.
We can repeat line~\ref{sta:compute_flow} $\Theta(\log n)$ times to boost the success probability to $1-n^{-10}$.
It only increases the work by a $O(\log n)$ factor.
Line~\ref{sta:select_out_edge} can be done in $\poly(\log n)$ depth using $m\poly(\log n)$ work.
In line~\ref{sta:compute_spanning_forest}, computing connected components and a spanning forest can be done in $\poly(\log n)$ depth using $m\poly(\log n)$ work~\cite{sv82}.
In line~\ref{sta:dis_to_root}, we can use doubling algorithm to compute $l(u)$ for all $u\in V$ simultaneously in $\poly(\log n)$ depth using $n\poly(\log n)$ work (see e.g., the doubling algorithm mentioned in~\cite{asswz18}).
Suppose the number of hops of $p$ before shortcutting cycles is $h$.
Then the path $p$ in line~\ref{sta:create_p} can be obtained in $\poly(\log(nh))$ depth using $(n+h)\cdot \poly(\log(nh))$ work by a doubling algorithm (see e.g.,~\cite{asswz18}).
Notice that the path $p''$ obtained by line~\ref{sta:recursion} has no cycle and thus each edge $e\in E$ can appear in $p$ obtained by line~\ref{sta:create_p} at most twice. 
Hence, line~\ref{sta:create_p} can be implemented in $\poly(\log n)$ depth using $m\poly(\log n)$ work.
In line~\ref{sta:shortcut_cycle}, we use the following way to shortcut cycles of $p$:
\begin{enumerate}
\item For each vertex $v$, find its last appearance in $p$. 
\item For the vertex $u_i$ appeared in $p$, if there is $j\in\{0,1,\cdots,i-1\}$ such that the last appearance of $u_j$ is $u_i$ or after $u_i$, then remove $u_i$ from $p$. 
\end{enumerate}
To do the second step, we can use doubling algorithm to preprocess a prefix max which can be done in $\poly(\log n)$ depth using $m\poly(\log n)$ work.

Next, consider the number of recursions (line~\ref{sta:recursion}) needed.
Except vertex $t$, every vertex $u$ has an edge $\{u,\p(u)\}\in E'$. 
Thus, each connected component without $t$ in $G'$ has size at most $2$.
Notice that $|V''|$ is exactly the same as the number of connected components in $G'$. 
Thus, $|V''|\leq \lceil n/2\rceil$.
It implies that the number of recursions is at most $\lceil\log n\rceil$.
Thus, the overall depth is at most $\varepsilon^{-2}\poly(\log(n\Lambda))$ and the overall expected work is at most $\varepsilon^{-2}m\poly(\log(n\Lambda))$.
\end{proof}

\begin{lemma}[Correctness of parallel approximate $s-t$ shortest path]\label{lem:correctness_stpath}
Given an $\varepsilon\in (0,0.5)$, a connected $n$-vertex $m$-edge undirected graph $G=(V,E,w)$ with $w:E\rightarrow \mathbb{Z}_{\geq 0}$, and two vertices $s,t\in V$, let $p$ be the output of \textsc{FindPath}$(G,s,t,\varepsilon)$ (Algorithm~\ref{alg:stpath}). $\E[w(p)]\leq (1+2\varepsilon)^{\lceil \log n\rceil} \cdot  \dist_G(s,t)$. 
\end{lemma}
\begin{proof}
Our proof is by induction on the number of vertices of $G$.
When $G$ only has one vertex, the statement is obviously true.
Now suppose the statement is true for any graph with less than $n$ vertices and consider $G$ with $n$ vertices.

For $u\in V$, let $l(u)$ be the same as that in line~\ref{sta:dis_to_root}.
Let $p$ be the output of \textsc{FindPath}$(G,s,t,\varepsilon)$.
Let $p''=(u_0'',u_1'',\cdots,u_{h''}'')$ be the output of line~\ref{sta:recursion}.
By line~\ref{sta:recursion}, $u_0''=\rt(s),u_{h''}''=\rt(t)$.
By line~\ref{sta:compute_spanning_forest}, we further have $u_{h''}''=\rt(t)=t$.
\begin{claim}\label{cla:output_path}
$w(p)\leq w''(p'')+l(s)$.
\end{claim}
\begin{proof}
	Notice that 
	\begin{align*}
	w(p) & \leq w((s,\p(s),\p(\p(s)),\cdots,\rt(s)))\\
	&+\sum_{\underset{\{x_i,y_i\}=\map(\{u_{i-1}'',u_i''\})}{x_i,y_i:}} w((\rt(x_i),\cdots,\p(x_i),x_i,y_{i},\p(y_{i}),\cdots,\rt(y_{i})))\\
	&= l(s) + \sum_{\underset{\{x_i,y_i\}=\map(\{u_{i-1}'',u_i''\})}{x_i,y_i:}} l(x_i) + w(x_i,y_i) + l(y_i)\\
	&= l(s) + \sum_{i=1}^{h''} w''(u_{i-1}'',u_i'') \\
	& = l(s) + w''(p''),
	\end{align*}
where the first inequality follows from line~\ref{sta:shortcut_cycle}, the first equality follows from line~\ref{sta:dis_to_root}, and the second equality follows from line~\ref{sta:setup_weight_wdoubleprime}.
\end{proof}

\begin{claim}\label{cla:small_shortest_path}
$\E_{\p:V\setminus\{t\}\rightarrow V}[l(s)+\dist_{G''}(\rt(s),\rt(t))]\leq (1+2\varepsilon) \dist_G(s,t)$.
\end{claim}
\begin{proof}
Let $f$ be the flow obtained by line~\ref{sta:compute_flow} satisfying $\sum_{\{u,v\}\in E} w(u,v)\cdot |f(u,v)|\leq (1+\varepsilon)\dist_G(s,t)$. 
Our proof is by coupling.
By Lemma~\ref{lem:expected_length_of_random_walk}, we only need to prove that $\E[l(s)+\dist_{G''}(\rt(s),\rt(t))]$ is almost upper bounded by the expected length of a random walk corresponding to the flow $f$.

For $u\in V$, let $\p(u)$ be the same as that in Algorithm~\ref{alg:stpath}. 
We conceptually generate a random walk $\wh{p}$ corresponding to $f$ in the following way:
\begin{enumerate}
\item Set $i\gets 0$ and set $\wh{u}_0\gets s$.
\item If $\forall j\in\{0,1,\cdots,i-1\},\wh{u}_j\not=\wh{u}_i$, then set $\wh{u}_{i+1}\gets \p(\wh{u}_i)$. 
Otherwise, set $\wh{u}_{i+1}$ to be $ v\in\{v'\in V\mid f(\wh{u}_i,v')>0\}$ with probability $\frac{f(\wh{u}_i,v)}{\sum_{v':f(\wh{u}_i,v')>0} f(\wh{u}_i,v')}$.  \label{it:coupling_next_step}
\item If $\wh{u}_{i+1}\not =t$, set $i\gets i+1$, and repeat step~\ref{it:coupling_next_step}. 
Otherwise, $\hat{h}\gets i+1$, and output the path $\wh{p}=(\wh{u}_0,\wh{u}_1,\cdots,\wh{u}_{\wh{h}})$ as the random walk.
\end{enumerate}
It is easy to see that $\wh{p}$ is a random walk corresponding to the flow $f$ and thus $\E[w(\wh{p})] = \sum_{\{u,v\}\in E} w(u,v)\cdot |f(u,v)|$ by Lemma~\ref{lem:expected_length_of_random_walk}.
For each edge $\{\wh{u}_{i-1},\wh{u}_i\}$ in the path $\wh{p}$, if $\exists j\in\{0,1,\cdots,i-1\},k\in\{i,i+1,\cdots,\wh{h}\}$ such that $\wh{u}_j=\wh{u}_k$, we call edge $\{\wh{u}_{i-1},\wh{u}_i\}$ a \emph{redundant} edge. 
Otherwise, we call $\{\wh{u}_{i-1},\wh{u}_i\}$ a \emph{crucial} edge.
By the above construction, $\hat{p}$ has several good properties.
\begin{fact}\label{fac:first_appearance}
If $\wh{u}_i$ is the first appearance of a vertex $v$, then $\wh{u}_{i+1}$ must be $\p(v)$.
\end{fact} 

\begin{fact}\label{fac:appeared_edges}
If vertex $v\in V$ appears in $\wh{p}$ and $v\not=t$, all edges $\{v,\p(v)\},\{\p(v),\p(\p(v))\},$ $\{\p(\p(v)),\p(\p(\p(v)))\},\cdots,$ should be in $\wh{p}$. 
\end{fact}
\begin{proof}
It directly follows from Fact~\ref{fac:first_appearance}.
\end{proof}

\begin{fact}\label{fac:twice_redundant}
If vertex $v$ appears at least twice in $\wh{p}$, $\{v,\p(v)\}$ is a redundant edge.
\end{fact}
\begin{proof}
Suppose $\wh{u}_i$ is the first appearance of $v$, then $\wh{u}_{i+1}=\p(v)$ (Fact~\ref{fac:first_appearance}).
Suppose $\wh{u}_j$ is the second appearance of $v$, we have $\wh{u}_i=\wh{u}_j$ and $j\geq i+1$.
$\{v,\p(v)\}=\{\wh{u}_i,\wh{u}_{i+1}\}$ is a redundant edge by definition. 
\end{proof}

\begin{fact}\label{fac:follows_redundant}
If $\{v,\p(v)\}$ is a redundant edge, $\{\p(v),\p(\p(v))\}$ is a redundant edge.
\end{fact}
\begin{proof}
If $\p(v)$ appears in $\wh{p}$ at least twice, $\{\p(v),\p(\p(v))\}$ is a redundant edge due to Fact~\ref{fac:twice_redundant}.
Otherwise, suppose $\hat{u}_i$ is the first appearance of $v$. 
$\hat{u}_{i+1}$ must be $\p(v)$ due to Fact~\ref{fac:first_appearance}.
Since $\p(v)$ only appears once, $\hat{u}_{i+1}$ is the first appearance of $\p(v)$.
By Fact~\ref{fac:first_appearance} again, $\hat{u}_{i+2}$ must be $\p(\p(v))$.
Since $\{\wh{u}_i,\wh{u}_{i+1}\}$ is a redundant edge and $\wh{u}_{i+1}$ only appears once, by the definition of redundant edge, there is $j\leq i$ and $k\geq i+2$ such that $\wh{u}_j=\wh{u}_k$.
By the definition, $\{\p(v),\p(\p(v))\}=\{\wh{u}_{i+1},\wh{u}_{i+2}\}$ is a redundant edge.
\end{proof}

\begin{fact}\label{fac:cost_redundant_edge}
	\begin{align*}
	\E_{\wh{p}}\left[\sum_{i:\{\wh{u}_{i-1},\wh{u}_i\}\text{ is redundant}} w(\wh{u}_{i-1},\wh{u}_i)\right] \leq \varepsilon\cdot \dist_G(s,t).
	\end{align*}
\end{fact}
\begin{proof}
Observe that if we remove all redundant edges from $\wh{p}$, the remaining edges still form a path from $s$ to $t$.
Thus, we always have:
\begin{align*}
\sum_{i:\{\wh{u}_{i-1},\wh{u}_i\}\text{ is crucial}} w(\wh{u}_{i-1},\wh{u}_i) \geq \dist_G(s,t).
\end{align*}
On the other hand, we have:
\begin{align*}
 &\E\left[\sum_{i:\{\wh{u}_{i-1},\wh{u}_i\}\text{ is crucial}} w(\wh{u}_{i-1},\wh{u}_i)\right] + \E\left[\sum_{i:\{\wh{u}_{i-1},\wh{u}_i\}\text{ is redundant}} w(\wh{u}_{i-1},\wh{u}_i)\right]\\
  =&\E\left[w(\wh{p})\right]=  \sum_{\{u,v\}\in E} w(u,v)\cdot |f(u,v)| \leq (1+\varepsilon)\dist_G(s,t),
\end{align*}
where the second equality follows from Lemma~\ref{lem:expected_length_of_random_walk} and the last inequality follows from that $f$ is a $(1+\varepsilon)$-approximate uncapacitated minimum cost $s-t$ flow.
Thus,
\begin{align*}
\E\left[\sum_{i:\{\wh{u}_{i-1},\wh{u}_i\}\text{ is redundant}} w(\wh{u}_{i-1},\wh{u}_i)\right] \leq \varepsilon\cdot \dist_G(s,t).
\end{align*}
\end{proof}
Next, we show how to find a path in $G''$ of which cost is at most $\sum_{i:\{\wh{u}_{i-1},\wh{u}_i\}\text{ is crucial}} w(\wh{u}_{i-1},\wh{u}_i) + 2 \sum_{i:\{\wh{u}_{i-1},\wh{u}_i\}\text{ is redundant}} w(\wh{u}_{i-1},\wh{u}_i)$.
Let us conceptually construct a path $\wt{p}$ in $G$ using the following way:
\begin{enumerate}
\item Initialize $\wt{p}=(s,\p(s),\p(\p(s)),\cdots,\rt(s)),i\gets 0$. Let $\wt{u}_0''\gets \rt(s)$.
\item Let $\wh{u}_{k_i}$ be the last vertex in $\wh{p}$ such that $\wh{u}_{k_i}$ is in the same connected component as $\wt{u}''_i$ in $G'$ ($G'$ is constructed by line~\ref{sta:construct_Gprime}).
\item If $\wt{u}''_i = t$, let $\wt{h}''\gets i$ and finish the procedure. \label{it:check_target}
\item Otherwise, concatenate 
\begin{align*}
	(\rt(\wh{u}_{k_i}) = \wt{u}_i'', \cdots \p(\p(\wh{u}_{k_i})),\p(\wh{u}_{k_i}),\wh{u}_{k_i},\wh{u}_{k_i+1},\p(\wh{u}_{k_i+1}),\p(\p(\wh{u}_{k_i+1})),\cdots,\rt(\wh{u}_{k_i+1}))
\end{align*}
to $\wt{p}$ and set $\wt{u}_{i+1}''\gets \rt(\wh{u}_{k_i+1})$.
\item Set $i\gets i + 1$ and go to step~\ref{it:check_target}.
\end{enumerate}

\begin{fact}\label{fac:root_appears_once}
$\forall i\not=j\in\{0,1,\cdots,\wt{h}''\}, \wt{u}_i''\not=\wt{u}_j''$.
\end{fact}
\begin{proof}
This follows from that $\wh{u}_{k_i}$ is the last vertex in $\hat{p}$ such that $\wh{u}_{k_i}$ is in the same connected component as $\wt{u}_i''$ in $G'$ and $k_0<k_1<k_2<\cdots<k_{\wh{h}''}$.
\end{proof}
\begin{fact}\label{fac:apears_in_phat}
Each edge in $\wt{p}$ appears in $\wh{p}$.
\end{fact}
\begin{proof}
since $\wh{u}_0=s$, edges $\{s,\p(s)\},\{\p(s),\p(\p(s))\},\{\p(\p(s)),\p(\p(\p(s)))\},\cdots$ must appear in $\wh{p}$ by Fact~\ref{fac:appeared_edges}.
By the choice of $\wh{u}_{k_i}$ and $\wh{u}_{k_i+1}$, edge $\{\wh{u}_{k_i},\wh{u}_{k_i+1}\}$ appears in $\wh{p}$.
By Fact~\ref{fac:appeared_edges} again, edges $\{\wh{u}_{k_i},\p(\wh{u}_{k_i})\}, \{\p(\wh{u}_{k_i}),\p(\p(\wh{u}_{k_i}))\},$ $\{\p(\p(\wh{u}_{k_i})),\p(\p(\p(\wh{u}_{k_i})))\},\cdots$ and edges $\{\wh{u}_{k_i+1},\p(\wh{u}_{k_i+1})\}, \{\p(\wh{u}_{k_i+1}),\p(\p(\wh{u}_{k_i+1}))\},$ $\{\p(\p(\wh{u}_{k_i+1})),\p(\p(\p(\wh{u}_{k_i+1})))\},\cdots$ also appear in $\wh{p}$.
\end{proof}

\begin{fact}\label{fac:appear_at_most_twice}
Each edge in $\wt{p}$ can appear at most twice. In addition, if an edge $e$ in $\wt{p}$ appears twice, $e$ is a redundant edge.
\end{fact}
\begin{proof}
Suppose $e=\{u,v\}$ appears in $\wt{p}$.
By Fact~\ref{fac:root_appears_once}, $\rt(u),\rt(v)$ only appears once in $\wt{p}$.
If $u$ and $v$ are not in the same connected component in $G'$, then either $u=\wh{u}_{k_i},v=\wh{u}_{k_i+1}$ or $v=\wh{u}_{k_i},u=\wh{u}_{k_i+1}$ for some $i\in\{0,1,\cdots,\wt{h}''\}$.
In this case, $\{u,v\}$ only appears once in $\wt{p}$.
If $u$ and $v$ are in the same connected component, then either $v=\p(u)$ or $u=\p(v)$.
Without loss of generality, suppose $v=\p(u)$.
If $(u,\p(u))$ appears in $\wt{p}$, the subpath $(u,\p(u),\p(\p(u)),\cdots,\rt(u))$ appears in $\wt{p}$.
If $(\p(u),u)$ appears in $\wt{p}$, the subpath $(\rt(u),\cdots,\p(\p(u)),\p(u),u)$ appears in $\wt{p}$.
As mentioned previously, $\rt(u)$ only appears once in $\wt{p}$.
Thus, $(u,\p(u))$ can appear at most once in $\wt{p}$ and $(\p(u),u)$ can appear at most once in $\wt{p}$ which means that $\{u,\p(u)\}$ can appear at most twice in $\wt{p}$.

Suppose $\{u,\p(u)\}$ appears twice in $\wt{p}$, then by the construction of $\wt{p}$, the subpath 
\begin{align}\label{eq:subpath_in_tree}
(u,\p(u),\p(\p(u)),\cdots,\rt(u),\cdots,\p(\p(u)),\p(u),u)
\end{align}
must appear in $\wt{p}$. 
More precisely, the subpath \eqref{eq:subpath_in_tree} should appear in the subpath 
\begin{align*}
(\wh{u}_{k_i+1},\p(\wh{u}_{k_i+1}),\cdots,u,\p(u),\cdots,\rt(u),\cdots,\p(u),u,\cdots,\p(\wh{u}_{k_{i+1}}),\wh{u}_{k_{i+1}})
\end{align*}
for some $i\in\{0,1,2,\cdots,\wt{h}''-1\}$.
If $\wh{u}_{k_i+1}=\wh{u}_{k_{i+1}}$, it means that $\wh{u}_{k_i+1}$ appears twice in $\wh{p}$.
In this case, by Fact~\ref{fac:twice_redundant}, $\{\wh{u}_{k_i+1},\p(\wh{u}_{k_i+1})\}$ is a redundant edge.
By Fact~\ref{fac:follows_redundant}, 
\begin{align*}
\{\wh{u}_{k_i+1},\p(\wh{u}_{k_i+1})\},\{\p(\wh{u}_{k_i+1}),\p(\p(\wh{u}_{k_i+1}))\},\{\p(\p(\wh{u}_{k_i+1})),\p(\p(\p(\wh{u}_{k_i+1})))\},\cdots
\end{align*}
 are redundant edges. 
Thus, $\{u,\p(u)\}$ is a redundant edge.
In the case when $\wh{u}_{k_i+1}\not=\wh{u}_{k_{i+1}}$, we can find two vertices $x\not=y$ such that $x=\p(\p(\cdots\p(\wh{u}_{k_i+1})))$, $y=\p(\p(\cdots\p(\wh{u}_{k_{i+1}})))$ and $\p(x)=\p(y)$.
By Fact~\ref{fac:apears_in_phat}, both $\{x,\p(x)\}$ and $\{y,\p(y)\}$ appear in $\wh{p}$ which means that $\p(x)$ appears twice in $\wh{p}$. 
By Fact~\ref{fac:twice_redundant}, $\{\p(x),\p(\p(x))\}$ is a redundant edge.
Since $u=\p(\p(\cdots\p(x)))$, $\{u,\p(u)\}$ must be a redundant edge according to Fact~\ref{fac:follows_redundant}. 
\end{proof}
By Fact~\ref{fac:apears_in_phat} and Fact~\ref{fac:appear_at_most_twice}, we have
\begin{align}\label{eq:ub_in_Gprime}
w(\wt{p})\leq \sum_{i:\{\wh{u}_{i-1},\wh{u}_i\}\text{ is crucial}} w(\wh{u}_{i-1},\wh{u}_i) + 2\sum_{i:\{\wh{u}_{i-1},\wh{u}_i\}\text{ is redundant}} w(\wh{u}_{i-1},\wh{u}_i).
\end{align}
Let $\wt{p}''=(\wt{u}_0'',\wt{u}_1'',\wt{u}_2'',\cdots,\wt{u}_{\wt{h}''}'')$. 
It is obvious that $\wt{p}''$ is a path in $G''$ connecting $\rt(s)$ and $\rt(t)$.
In addition, we have $w(\wt{p})=l(s)+w''(\wt{p}'')$.
Thus, to conclude,
\begin{align*}
&\E\left[l(s)+\dist_{G''}(\rt(s),\rt(t))\right]\\
\leq & \E[l(s)+w''(\wt{p}'')] \\
\leq &\E\left[\sum_{i:\{\wh{u}_{i-1},\wh{u}_i\}\text{ is crucial}} w(\wh{u}_{i-1},\wh{u}_i) + 2\sum_{i:\{\wh{u}_{i-1},\wh{u}_i\}\text{ is redundant}} w(\wh{u}_{i-1},\wh{u}_i)\right]\\
\leq & \E[w(\wh{p})] + \varepsilon\cdot \dist_G(s,t)\\
= & \left(\sum_{\{u,v\}\in E} w(u,v)\cdot |f(u,c)|\right) + \varepsilon\cdot \dist_G(s,t)\\
\leq & (1+2\varepsilon) \dist_G(s,t),
\end{align*}
where the first step follows from that $\wt{p}''$ is a path in $G''$ connecting $\rt(s)$ and $\rt(t)$, the second step follows from Equation~\ref{eq:ub_in_Gprime}, the third step follows from Fact~\ref{fac:cost_redundant_edge}, the forth step follows from Lemma~\ref{lem:expected_length_of_random_walk}, and the last step follows from that $f$ is a $(1+\varepsilon)$-approximate uncapacitated minimum cost flow from $s$ to $t$.
\end{proof}
As proved in the proof of Lemma~\ref{lem:work_depth_stpath}, $|V''|\leq \lceil n/2\rceil$. 
By induction hypothesis, we have $\E[w''(p'')]\leq (1+2\varepsilon)^{\lceil\log n\rceil-1}\cdot  \E[\dist_{G''}(\rt(s),\rt(t))]$.
Thus, we have
\begin{align*}
\E[w(p)] & \leq \E[w''(p'') + l(s)] \\
&\leq (1+2\varepsilon)^{\lceil\log n\rceil-1}\E[\dist_{G''}(\rt(s),\rt(t))] + \E[l(s)]\\
&\leq (1+2\varepsilon)^{\lceil\log n\rceil-1}\E[\dist_{G''}(\rt(s),\rt(t))+l(s)]\\
&\leq (1+2\varepsilon)^{\lceil\log n\rceil}\dist_G(s,t),
\end{align*}
where the first step follows from Claim~\ref{cla:output_path}, the second step follows from induction hypothesis, the last step follows from Claim~\ref{cla:small_shortest_path}.
\end{proof}

\begin{theorem}[Parallel approximate $s-t$  shortest path]\label{thm:st_shortest_path}
Given an $\varepsilon\in(0,0.5)$, a connected $n$-vertex $m$-edge undirected graph $G=(V,E,w)$ with $w:E\rightarrow \mathbb{Z}_{\geq 0}$, and two vertices $s,t\in V$, there is a PRAM algorithm which takes $\varepsilon^{-2}\poly(\log(n\Lambda))$ parallel time using expected $\varepsilon^{-3}m\poly(\log(n\Lambda))$ work and with probability at least $0.99$ outputs an $s-t$ path satisfying $w(p)\leq (1+\varepsilon)\cdot \dist_G(s,t)$.
\end{theorem}
\begin{proof}
We invoke \textsc{FindPath}$(G,s,t,\varepsilon')$  $\Theta(\varepsilon^{-1}\log n)$ times, where  $\varepsilon'=\frac{\varepsilon}{20\log n}$.
The depth and work is shown by Lemma~\ref{lem:work_depth_stpath}.
As mentioned in the proof of Lemma~\ref{lem:work_depth_stpath}, we can repeat line~\ref{sta:compute_flow} $\Theta(\log n)$ times to boost the success probability of computing the flow to $1-n^{-10}$. 
By taking union bound, all the flow computation succeed with probability at least $0.999$.
Condition on success of all the flow computation, by Lemma~\ref{lem:correctness_stpath}, \textsc{FindPath}$(G,s,t,\varepsilon')$ outputs a path $p$ satisfies $\E[w(p)]\leq (1+2\varepsilon')^{\lceil\log n\rceil} \cdot\dist_G(s,t)\leq (1+\varepsilon/2)\cdot \dist_G(s,t)$.
By repeating $\Theta(\varepsilon^{-1}\log n)$ times, with probability at least $0.999$, we can find an $s-t$ path $p$ such that $w(p)\leq (1+\varepsilon)\dist_G(s,t)$.
\end{proof}



\section*{Acknowledgements}
We thank Aaron Bernstein, Yan Gu, Hossein Esfandiari, Jakub Łącki, Vahab Mirrokni and Ruosong Wang for very helpful discussions.
Part of this work was done while Peilin Zhong was an intern at Google New York.

\newpage

\addcontentsline{toc}{section}{References}
\bibliographystyle{alpha}
\bibliography{ref}
\newpage

\appendix

\section{The Necessity of $2$ Types of Edges in the Subemulator}\label{sec:necessity_edges_subemulator}
We show that both types of edges constructed by line~\ref{sta:original_edges} and line~\ref{sta:ball_edges} in Algorithm~\ref{alg:edge_construction} are necessary for the construction of subemulator.
If we only contain the edges constructed by line~\ref{sta:original_edges} and miss the edges constructed by line~\ref{sta:ball_edges}, Figure~\ref{fig:only_original_edges} gives an example that the constructed graph can not be a subemulator.
\begin{figure}[t!]
	\centering
	\includegraphics[width=0.8\textwidth]{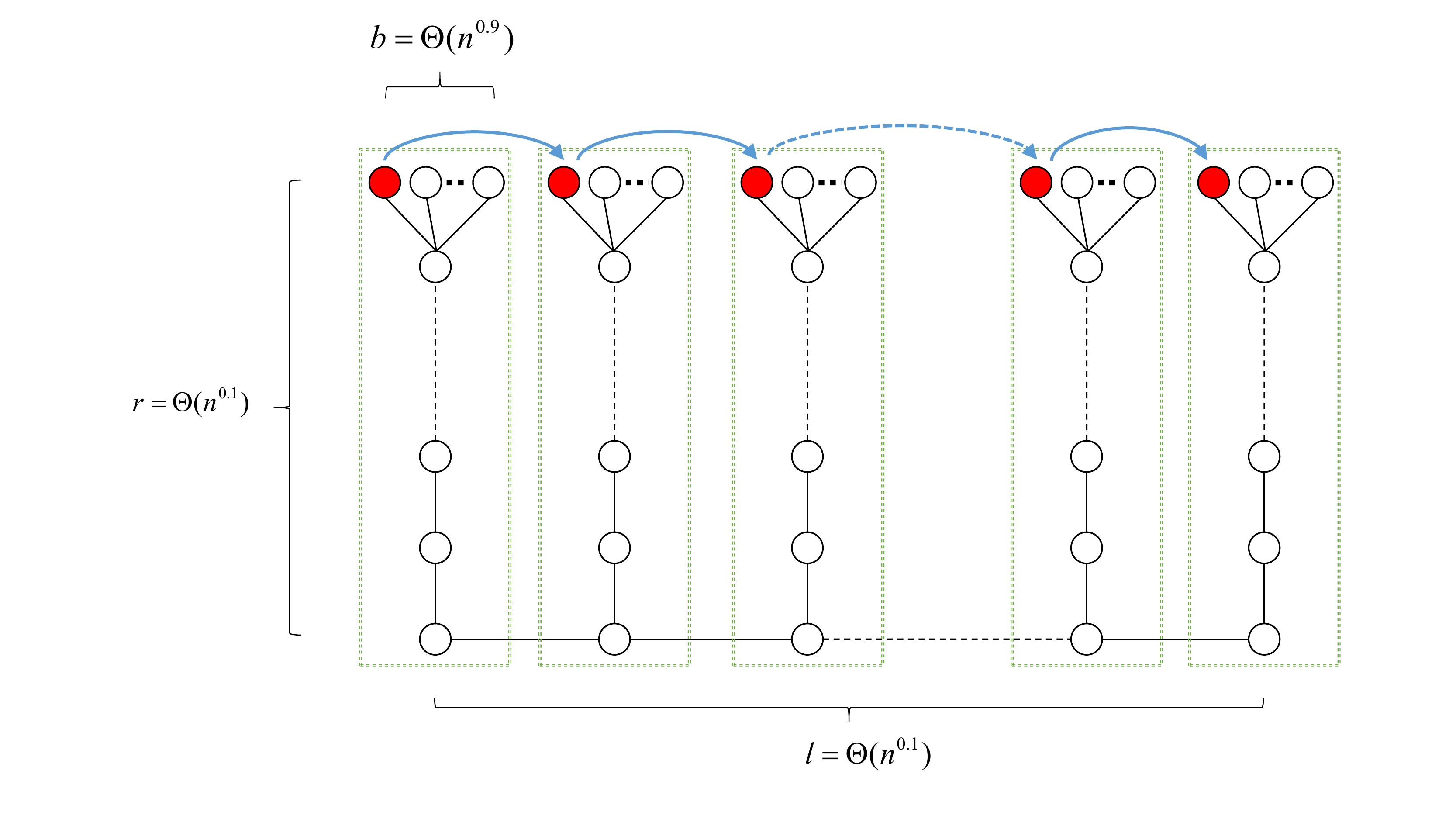}
	\caption{ \small The graph is unweighted and is a tree constructed by following steps. 
	We first create a path with length $l=\Theta(n^{0.1})$.
	For each vertex on the path, we create a brunch starting with a path with length $r=\Theta(n^{0.1})$ and ending with a star with $b=\Theta(n^{0.9})$ vertices.
	If we sample each vertex (solid red vertex) to be in the subemulator with probability $\log(n)/b$, with high probability, sampled vertices can only appear in stars and each brunch must have at least one sampled vertex.
	We condition on this event.
	It is clear that each vertex has at least one $(b+r)$-closest neighbor which is a sampled vertex, and that sampled vertex must be in the same dashed green box.
	If we only contain the edges constructed by line~\ref{sta:original_edges} of Algorithm~\ref{alg:edge_construction}, the result graph must be a length-$l$ path (represented by blue arcs) where each edge corresponds to an edge crossing two dashed green box above and has weight $2r+1$.
	Thus the diameter of the result graph is $l(2r+1)=\Theta(n^{0.2})$.
	However, the diameter of the original graph is $2r+l=\Theta(n^{0.1})$ which implies that the result graph is not a good subemulator.
 }\label{fig:only_original_edges}
\end{figure}
If we only contain the edges constructed by line~\ref{sta:ball_edges} and miss the edges constructed by line~\ref{sta:original_edges}, Figure~\ref{fig:only_ball_edges} gives an example that the constructed graph can not be a subemulator.
\begin{figure}[t!]
	\centering
	\includegraphics[width=0.8\textwidth]{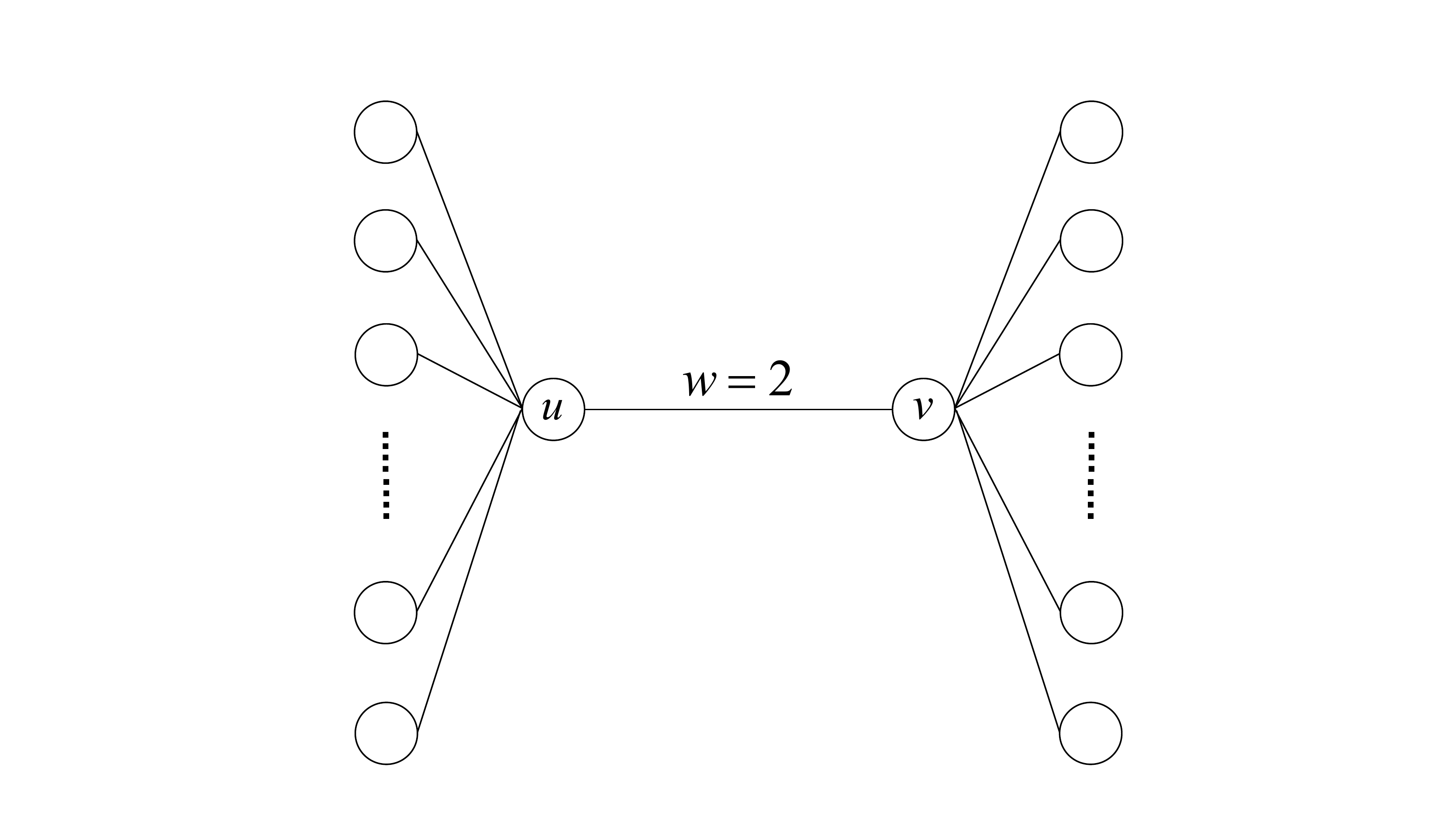}
	\caption{ \small The graph contains two stars connecting by an edge with weight $2$.
		Each star has $n/2$ vertices.
		One star has center $u$ and another has center $v$.
		Except the edge between $u,v$, all other edges have weights $1$.
		For $b<n/2$, neither $v$ is a $b$-closest neighbor of any vertex in the star with center $u$ nor $u$ is a $b$-closest neighbor of any vertex in the star with center $v$. 
		Thus, if we only contain the edges constructed by line~\ref{sta:ball_edges} of Algorithm~\ref{alg:edge_construction}, the result graph is disconnected which cannot be a subemulator.
	}\label{fig:only_ball_edges}
\end{figure}

\section{Proof of Lemma~\ref{lem:subemulator_vertex}} \label{sec:proof_of_subemulator_vertex}

\begin{proof}
	It is implied by \cite{asswz18}. For the completeness, we prove the lemma in the following.
	
	The correctness follows from line~\ref{sta:check_each_vertex} of Algorithm~\ref{alg:leader_selection}, i.e., $\forall v\in V,\B_{G,b}(v)\cap V'\not= \emptyset$.
	
	Let us consider the expected size of $V'$.
	We have 
	\begin{align*}
	\E[|V'|] = \E[|S|] + \E[|\{v\in V\mid \B_{G,b}(v)\cap S=\emptyset\}|].
	\end{align*}
	
	Consider the case when $1/2<50\log(n)/b$. 
	We have $\E[|S|]=1/2\cdot |V|$.
	Notice that $\B_{G,b}(v)\cap S=\emptyset$ implies that neither $v\in S$ nor $u\in S$, where $u\not=v$ is an arbitrary vertex in $\B_{G,b}(v)$.
	Thus, we have $\Pr[ \B_{G,b}(v)\cap S=\emptyset]\leq 1/4$, and it implies that $\E[|V'|]\leq 1/2\cdot |V|+ 1/4\cdot |V|=3/4\cdot n$.
	
	Consider the case when $1/2>50\log(n)/b$. 
	We have $\E[|S|]=50\log (n)/b\cdot |V|$.
	Since $\E[|\B_{G,b}(v)\cap S|]=b\cdot 50\log(n)/b=50\log(n)$,
	by Chernoff bound, we have
	\begin{align*}
	\Pr[\B_{G,b}(v)\cap S=\emptyset]\leq \Pr[ |\B_{G,b}(v)\cap S| \leq 25\log(n)]\leq e^{-50\log(n)/8}\leq 1/n^2.
	\end{align*}
	By union bound,
	\begin{align*}
	\Pr[\exists v,\B_{G,b}(v)\cap S=\emptyset]=1/n.
	\end{align*}
	Thus, $\E[|\{v\in V\mid \B_{G,b}(v)\cap S=\emptyset\}|]\leq 1/n\cdot n = 1$.
	We can conclude that $\E[|V'|]\leq 50\log(n)/b\cdot n + 1\leq 75\log(n)/b\cdot n$.
\end{proof}

\section{Proof of Lemma~\ref{lem:mwu_feas_prob}} \label{sec:proof_of_mwu_feas_prob}
\begin{proof}
	Let us firstly consider the case that \textsc{MWU}$(P,A,W,b,s,\varepsilon,\kappa)$ does not return FAIL.
	By line~\ref{sta:return_xprime}, we have $\|PBp_t\|_1\leq \frac{\varepsilon}{2\kappa}$, i.e.,
	\begin{align*}
	\left\|\sum_{i=1}^m \left( \frac{(PAW^{-1})_i}{\|PAW^{-1}\|_{1\rightarrow 1}}\cdot (p_t^+(i)-p_t^-(i))-\frac{1}{s}\cdot \frac{Pb}{\|Pb\|_1}\cdot(p_t^+(i)+p_t^-(i))\right)\right\|_1\leq \frac{\varepsilon}{2\kappa}.
	\end{align*}
	Since $\sum_{i=1}^m (p_t^+(i)+p_t^-(i))=1$ and $x'_i=p_t^+(i)-p_t^-(i)$, we have:
	\begin{align*}
	\left\|\frac{PAW^{-1}}{\|PAW\|_{1\rightarrow 1}}x'-\frac{1}{s}\cdot \frac{Pb}{\|Pb\|_1}\right\|_1\leq \frac{\varepsilon}{2\kappa}.
	\end{align*}
	Furthermore, because $\forall i\in[m],p_t^+(i),p_t^-(i)\geq 0$, we have $\|x'\|_1\leq 1$.
	Thus, $x'$ satisfies Equation~\eqref{eq:feas2}.
	
	Let us consider the case when \textsc{MWU}$(P,A,W,b,s,\varepsilon,\kappa)$ outputs FAIL.
	For $i\in[m],t\in[T]$, we have
	\begin{align*}
	|\phi_t^+(i)|&\leq \|y_t\|_{\infty}\|PB_i\|_1/2\\
	&=\left\|\frac{(PAW^{-1})_i}{\|PAW^{-1}\|_{1\rightarrow 1}}-\frac{1}{s}\cdot \frac{Pb}{\|Pb\|_1}\cdot \right\|_1/2\\
	&\leq \left\|\frac{(PAW^{-1})_i}{\|PAW^{-1}\|_{1\rightarrow 1}}\right\|_1/2+\left\|\frac{1}{s}\cdot \frac{Pb}{\|Pb\|_1}\cdot \right\|_1/2\\
	&\leq 1,
	\end{align*}
	where the first step follows from H\"oder's inequality, the second step follows from the construction of $B$, the third step follows from triangle inequality, and the last step follows from $\|(PAW^{-1})_i\|_1\leq \|PAW^{-1}\|_{1\rightarrow 1}$, $\|Pb/\|Pb\|_1\|_1=1$ and $s\geq 1$.
	Similarly, we also have $|\phi_t^-(i)|\leq 1$.
	By Theorem 2.1 of \cite{ahk12}:
	\begin{align*}
	&\forall j\in[m],\sum_{t=1}^T\sum_{i=1}^m (p_t^+(i)\phi_t^+(i) + p_t^-(i)\phi_t^-(i)) \leq \sum_{t=1}^T \phi_t^+ (j) + \eta \sum_{t=1}^T |\phi_t^+ (j)| + \frac{\ln(2m)}{\eta},\\
	&\forall j\in[m],\sum_{t=1}^T\sum_{i=1}^m (p_t^+(i)\phi_t^+(i) + p_t^-(i)\phi_t^-(i)) \leq \sum_{t=1}^T \phi_t^- (j) + \eta \sum_{t=1}^T |\phi_t^- (j)| + \frac{\ln(2m)}{\eta}.
	\end{align*}
	By the construction of $p_t^+,p_t^-,\phi_t^+,\phi_t^-$,
	\begin{align*}
	\sum_{t=1}^T\sum_{i=1}^m (p_t^+(i)\phi_t^+(i) + p_t^-(i)\phi_t^-(i)) = \sum_{t=1}^T y_t^\top PBp_t/2 > T\cdot \frac{\varepsilon}{4\kappa},
	\end{align*}
	where the inequality follows from $\forall l\in[r],\left(y_t\right)_l=\sgn\left((PBp_t)_i\right)$ and thus $y_t^\top PBp_t = \|PBp_t\|_1>\frac{\varepsilon}{2\kappa}$.
	Thus,
	\begin{align}
	&\forall j\in[m], T\cdot \frac{\varepsilon}{4\kappa} < \sum_{t=1}^T \phi_t^+ (j) + \eta T + \frac{\ln(2m)}{\eta},\label{eq:mwu_mid_1}\\
	&\forall j\in[m],T\cdot \frac{\varepsilon}{4\kappa} < \sum_{t=1}^T \phi_t^- (j) + \eta T + \frac{\ln(2m)}{\eta}.\label{eq:mwu_mid_2}
	\end{align}
	Let $\bar{y}=\frac{1}{T}\sum_{t=1}^T y_t$, then
	\begin{align*}
	&\forall j\in[m],\sum_{t=1}^T \phi_t^+(j) = T\cdot \bar{y}^\top \left(\frac{(PAW^{-1})_j}{\|PAW^{-1}\|_{1\rightarrow 1}} - \frac{1}{s} \cdot \frac{Pb}{\|Pb\|_1}\right),\\
	&\forall j\in[m],\sum_{t=1}^T \phi_t^-(j) = T\cdot \bar{y}^\top \left(-\frac{(PAW^{-1})_j}{\|PAW^{-1}\|_{1\rightarrow 1}} - \frac{1}{s} \cdot \frac{Pb}{\|Pb\|_1}\right).
	\end{align*}
	Recall that $\eta=\frac{\varepsilon}{8\kappa},T=\frac{64\kappa^2\ln(2m)}{\varepsilon^2}$. Thus, together with Equation~\eqref{eq:mwu_mid_1} and Equation~\eqref{eq:mwu_mid_2}, we have:
	
	\begin{align*}
	\forall j\in[m], 0< \bar{y}^\top \left(\frac{(PAW^{-1})_j}{\|PAW^{-1}\|_{1\rightarrow 1}} - \frac{1}{s} \cdot \frac{Pb}{\|Pb\|_1}\right), \\
	\forall j\in[m], 0< \bar{y}^\top \left(-\frac{(PAW^{-1})_j}{\|PAW^{-1}\|_{1\rightarrow 1}} - \frac{1}{s} \cdot \frac{Pb}{\|Pb\|_1}\right).
	\end{align*}
	For any $x'\in\mathbb{R}^m$ with $\|x'\|_1\leq 1$, we can always find $x'^+,x'^-\in\mathbb{R}^m$ such that $x'^+,x'^-\geq 0, x'=x'^+-x'^-$, and $\sum_{i=1}^m (x'^+_i+x'^-_i)=1$.
	If $x'$ satisfies Equation~\eqref{eq:feas2_ex}, then
	\begin{align*}
	0 & =\bar{y}^\top \left(\frac{PAW^{-1}}{\|PAW^{-1}\|_{1\rightarrow 1}}(x'^+-x'^-)-\frac{1}{s}\cdot \frac{Pb}{\|Pb\|_1}\right) \\
	& = \sum_{j=1}^m\left(\bar{y}^\top \left(\frac{(PAW^{-1})_j}{\|PAW^{-1}\|_{1\rightarrow 1}} - \frac{1}{s} \cdot \frac{Pb}{\|Pb\|_1}\right)\cdot x'^+_j + \bar{y}^\top \left(-\frac{(PAW^{-1})_j}{\|PAW^{-1}\|_{1\rightarrow 1}} - \frac{1}{s} \cdot \frac{Pb}{\|Pb\|_1}\right)\cdot x'^-_j\right)\\
	& > 0
	\end{align*}
	which leads to a contradiction.
\end{proof}

\section{Parallel Computation for $\B^{\circ}_b(v)$}\label{sec:computation_of_ball}
The PRAM algorithm is described as the following.
\begin{enumerate}
	\item For $v\in V$, initialize a list $L^{(0)}(v)$ containing $b$ closest neighbors (including $v$ itself) of $v$. 
	Let $t\gets \lceil\log n\rceil$.
	If $v$ has less than $b$ neighbors, let $L^{(0)}(v)$ be all of them. 
	For $u\in L^{(0)}(v)$, compute $\dist^{(1)}(v,u)$.
	\item For $i=1\rightarrow t$:
	\begin{enumerate}
		\item For  $v,u\in V$, initialize $\dist^{(2^i)}(v,u)\gets \infty$.
		\item Assign $b^2$ processors for each $v\in V$. Each processor reads a vertex $x\in L^{(i-1)}(v)$ and then reads a vertex $u\in L^{(i-1)}(x)$. 
		If $\dist^{(2^{i-1})}(v,x) + \dist^{(2^{i-1})}(x,u) < \dist^{(2^i)}(v,u)$, update $\dist^{(2^i)}(v,u)\gets \dist^{(2^{i-1})}(v,x) + \dist^{(2^{i-1})}(x,u)$.
		\item For $v\in V$, add $u$ to list $L^{(i)}(v)$ if $\dist^{(2^i)}(v,u)$ is one of the $b$ smallest values among $\dist^{(2^i)}(v,x)$ for $x\in V$. If there is a tie, take the vertex with a smaller label.
	\end{enumerate}
	\item For $v\in V$, suppose $u\in L^{(t)}(v)$ has the largest $\dist^{(2^{t})}(v,u)$. Let 
	\begin{align*}
	&r_b(v)\gets\dist^{(2^t)}(v,u),\\
	&\B^{\circ}_b(v)\gets \{x\in L^{(t)}(v) | \dist^{(2^{t})}(v,x) < r_b(v)\},\\
	&\forall x\in \B_b^{\circ}(v),\dist(v,x)\gets \dist^{(2^{t})}(v,x).
	\end{align*}
\end{enumerate}
By induction, it is easy to show that after $i$-th rounds, $L^{(i)}(v)$ contains vertices with $b$ smallest $2^i$-hop distance to $v$. Thus, after $t=\lceil\log n\rceil$ rounds, $L^{(t)}(v)$ contains $b$ closest vertices to $v$.

\section{Parallel Embedding into $\ell_1$} \label{sec:parallel_bourgain}
Consider a connected $n$-vertex $m$-edge undirected weighted graph $G=(V,E,w)$ with $w:E\rightarrow Z_{\geq 0}$. 
By Theorem~\ref{thm:parallel_low_hop_emulator}, we can use expected $\wt{O}(m)$ work and $\log^{O(1)}(n)$ depth to compute a low hop emulator $G'=(V,E',w')$, i.e., $\forall u,v \in V$,
\begin{align*}
\dist_G(u,v)\leq \dist_{G'}(u,v)\leq \log^{O(1)}(n) \cdot \dist_G(u,v),
\end{align*}
and the hop diameter of $G'$ is $O(\log\log n)$.
Obviously, the diameter of $G'$ is at most $\diam(G)\cdot \log^{O(1)}(n)$.
Furthermore, $\E[|E'|]\leq \wt{O}(n)$.
According to our construction of $G'$, since all weights in $G$ are integers, weights in $G'$ are also integers.

It suffices to embed $G'$ into $\ell_1$.
We apply Bourgain's embedding:
\begin{enumerate}
\item $t\gets \Theta(\log n) $.
\item For $i=1\rightarrow \lceil\log n\rceil,j = 1\rightarrow t$:
\begin{enumerate}
	\item Choose a set $S_{i,j}$ by sampling each $v\in V$ with probability $2^{-i}$.
	\item $\forall v\in V$, set the $((i-1)\cdot t + j)$-th coordinate of $\varphi(v)$ as $\dist_{G'}(S_{i,j},v)$, i.e., 
	\begin{align*}
	\varphi(v)_{(i-1)\cdot t + j}\gets \dist_{G'}(S_{i,j},v).
	\end{align*}\label{it:bellmanford}.
\end{enumerate}
\end{enumerate}
It is easy to see that for every $v\in V$ the coordinates of $\varphi(v)$ are non-negative integers and $\|\varphi(v)\|_{\infty}\leq \diam(G')\leq \diam(G)\cdot \log^{O(1)}(n)$.
The dimension of $\varphi(v)$ is $t\cdot \lceil\log n\rceil  \leq O(\log^2 n)$.
By Bourgain's theorem~\cite{b85}, with probability at least $0.999$, $\forall u,v \in V$,
\begin{align*}
\dist_{G'}(u,v)\leq \|\varphi(u)-\varphi(v)\|_1\leq \log^{O(1)}(n)\cdot \dist_{G'}(u,v),
\end{align*}
which implies that $\forall u,v\in V$,
\begin{align*}
\dist_G(u,v)\leq \|\varphi(u)-\varphi(v)\|_1\leq \log^{O(1)}(n)\cdot \dist_G(u,v).
\end{align*}
Now consider the depth and the work of the above procedure.
Step~\ref{it:bellmanford} can be implemented by Bellman-Ford algorithm. 
In particular, we add a super node which connects to every vertex in $S_{i,j}$ with weight zero.
By $h$ Bellman-Ford iterations, we can compute $\dist_{G'}^{(h)}(S_{i,j}, v)$ for every $v\in V$.
Since the hop diameter of $G'$ is $O(\log \log n)$, we only need $O(\log\log n)$ Bellman-Ford iterations in step~\ref{it:bellmanford}.
Thus, the depth of the above procedure is at most $\log^{O(1)}(n)$, and the work is at most $\wt{O}(|E'|)$.
Together with the computation of $G'$, the overall depth is $\log^{O(1)} n$, and the total work is at most $\wt{O}(m)$.

\section{Parallel Low Diameter Decomposition}\label{sec:parallel_low_diam_decomp}
By Theorem~\ref{thm:parallel_low_hop_emulator}, we can use expected $\wt{O}(m)$ work and $\log^{O(1)}(n)$ depth to compute a low hop emulator $G'=(V,E',w')$, i.e., $\forall u,v \in V$,
\begin{align*}
\dist_G(u,v)\leq \dist_{G'}(u,v)\leq \log^{O(1)}(n) \cdot \dist_G(u,v),
\end{align*}
and the hop diameter of $G'$ is $O(\log\log n)$.
Furthermore, $\E[|E'|]\leq \wt{O}(n)$.

It suffices to run low diameter decomposition \cite{mpx13} on $G'$:
\begin{enumerate}
\item For $v\in V$, draw $\delta_v$ independently from the exponential distribution with CDF $1-e^{-\beta x}$.
\item Compute the subset $C_u$ by assigning each $v$ to $u$ which minimizes $\dist_{G'}(v,u)-\delta_u$.
\item Remove empty subsets $C_u$ and return the remaining subsets $\{C_u\}$.
\end{enumerate}
By~\cite{mpx13}, $V$ will be partitioned into clusters such that
\begin{itemize}
\item for any two vertices $u,v$ from the same cluster, $\dist_{G'}(u,v)\leq O(\beta^{-1}\log n)$,
\item for any two vertices $u,v$, the probability that $u,v$ are not in the same cluster is at most $O(\beta\cdot \dist_{G'}(u,v))$.
\end{itemize}
Thus, it implies that the partition is also good for $G$:
\begin{itemize}
	\item for any two vertices $u,v$ from the same cluster, $\dist_{G}(u,v)\leq \beta^{-1} \log^{O(1)}(n)$,
	\item for any two vertices $u,v$, the probability that $u,v$ are not in the same cluster is at most $\beta\cdot \dist_{G}(u,v)\cdot \log^{O(1)}(n)$.
\end{itemize}
To implement the second step of the algorithm, we can add a super node which connects to every vertex $v$ with weight $ \max_{u\in V} \delta_u-\delta_v$.
Then we can use Bellman-Ford to compute the single source shortest path from the super node. 
Since the hop diameter of $G'$ is at most $O(\log\log n)$, the number of Bellman-Ford iterations is at most $O(\log\log n)$. 
Thus, the overall work is at most $\wt{O}(m)$ and the depth is at most $\log^{O(1)}(n)$.




\end{document}